\documentclass[USenglish,oneside,twocolumn]{article}

\newif\ifdraft
\draftfalse

\newif\iffullversioncontent
\fullversioncontenttrue

\newif\ifplots
\plotsfalse
\newif\ifpopets
\popetsfalse

\usepackage[T1]{fontenc}
\usepackage[utf8]{inputenc}
\usepackage[USenglish]{babel}

\ifpopets
    \usepackage[big]{dgruyter_NEW}
    \cclogo{\includegraphics{by-nc-nd.pdf}}
    \newcommand{\myparagraph}[1]{\noindent\textbf{#1.}}
\else
    \usepackage{geometry}
    \newcommand{\myparagraph}[1]{\vspace*{1em}\noindent\textbf{#1.}\newline}
\fi
\usepackage{lmodern}
 
\iffullversioncontent
    \newcommand{\fv}[1]{#1}
\else
    \newcommand{\fv}[1]{}
\fi

\ifdraft

    \newcommand{\new}[1]{{\color{blue-violet}#1}}
    \newcommand{\potentiallyfullversion}[1]{{\color{cyan}#1}}
    \newcommand{\redcomment}[1]{}
    \newcommand{\redemph}[1]{#1}
    
    \let\oldfv\fv
    \renewcommand{\fv}[1]{{\color{cordovan}\oldfv{#1}}}

\else
    \newcommand{\potentiallyfullversion}[1]{}
    \newcommand{\new}[1]{{#1}}
    \newcommand{\redcomment}[1]{}
    \newcommand{\redemph}[1]{{#1}}
\fi

\usepackage{amsmath}
\usepackage{amsthm}
\usepackage{amssymb}
\usepackage{amsfonts}
\usepackage{xcolor}
\usepackage{hyperref}
\usepackage{ifthen}
\usepackage{mathtools}
\usepackage{array}

\allowdisplaybreaks

\usepackage{stackengine}

\usepackage{xparse}
\usepackage{caption}
\usepackage{subcaption}
\usepackage{pgf}
\usepackage{pgfplots}
\usepackage{tikzscale}
\pgfplotsset{compat=1.16}
\usetikzlibrary{pgfplots.groupplots}
\usetikzlibrary{external}
\usepgfplotslibrary{external}
\usepgfplotslibrary{fillbetween}
\tikzset{external/system call= {pdflatex -save-size=800000 
                          -pool-size=100000000 
                          -extra-mem-top=500000000 
                          -extra-mem-bot=100000000 
                          -main-memory=900000000 
                          \tikzexternalcheckshellescape 
                          -halt-on-error 
                          -interaction=batchmode
                          -jobname "\image" "\texsource"}}
\usepackage{cleveref}
\usepackage[mode=buildnew]{standalone}

\usepackage{graphicx}
\usepackage{placeins}
\usepackage{dsfont}
\usepackage{enumitem}

\ifpopets
\usepackage[subtle,mathspacing=normal]{savetrees}
\fi

\definecolor{amethyst}{rgb}{0.6, 0.4, 0.8}
\definecolor{blue-violet}{rgb}{0.54, 0.17, 0.89}
\definecolor{cordovan}{rgb}{0.54, 0.25, 0.27}

\newcommand{\eps}{\varepsilon}
\newcommand{\expectation}[1]{\ensuremath{\mathop{\mathbb{E}}_{#1}}}
\newcommand{\indicatorfunc}[1]{\ensuremath{\text{\large$\mathds{1}$\normalfont}_{#1}}}

\newcommand{\pr}[2][\vphantom{}]{\ensuremath{\Pr_{#1} \left [#2 \right ]}}

\newcommand{\drawXfromY}[2]{\pr{#1\! \leftarrow\! #2 }}

\newcommand{\minfty}{{\ensuremath{\text{-}\infty}}}

\newcommand{\utilityweight}{\ensuremath{w}}
\newcommand{\utilityweightdecayrate}{\ensuremath{\gamma}}
\newcommand{\currentEpoch}{\ensuremath{t}}
\newcommand{\utilityloss}{\ensuremath{U}}

\newcommand{\rangebegin}{\ensuremath{r}}
\newcommand{\scales}{\ensuremath{B}}
\newcommand{\slopes}{\ensuremath{C}}
\newcommand{\sigmoidc}{\ensuremath{F}}
\newcommand{\bias}{\ensuremath{A}}

\newcommand{\LossWeMinimize}{{\ensuremath{L}}}
\newcommand{\learningRate}{{\ensuremath{\theta}}}

\newcommand{\Losssymbol}{{\ensuremath{\mathcal{L}}}}

\newcommand{\loss}[3]{{\ensuremath{\Losssymbol_{#1 / #2}\left( #3 \right)}}}
\newcommand{\losstwo}[2]{{\ensuremath{\Losssymbol_{#1 / #2}}}}

\newcommand{\checkInfty}[1]{\ifthenelse{\equal{#1}{\infty}}{\ensuremath{\hspace{-0.9pt}\infty\hspace{-0.9pt}}}{#1}}
\newcommand{\bucketSymbol}{\ensuremath{\omega}} 

\newcommand{\bucket}[2][\vphantom{}]{\ensuremath{\bucketSymbol\!_{#1}\left(\checkInfty{#2}\right)}}

\newcommand{\bucketd}[2][\vphantom{}]{\ensuremath{\sqcup_{#1}\!\left(\checkInfty{#2}\right)}}

\newcommand{\bucketdl}[2][\vphantom{}]{\ensuremath{\bar{\sqcup}_{#1}\!\left(\checkInfty{#2}\right)}}

\DeclareMathOperator*{\Sigmoid}{\ensuremath{\sigma}}
\DeclareMathOperator*{\SoftMax}{SoftMax}

\DeclareMathOperator{\erf}{erf}
\DeclareMathOperator{\smaller}{smallerThan}
\DeclareMathOperator{\equal}{equal}
\DeclareMathOperator{\ceil}{ceil}
\DeclareMathOperator*{\argmax}{arg\,max}

\newcounter{mathenv}

\newtheorem{theorem}[mathenv]{Theorem}
\newtheorem*{theorem*}{Theorem}
\newtheorem{lemma}[mathenv]{Lemma}
\newtheorem*{lemma*}{Lemma}
\newtheorem{proposition}[mathenv]{Proposition}

\newtheorem*{claim*}{Claim}  
\newtheorem{definition}[mathenv]{Definition}
\newtheorem{example}[mathenv]{Example}

\ifpopets\else
\providecommand{\keywords}[1]
{
  {\small	
  \vspace*{1em}\noindent\textbf{\textit{Keywords:}} #1
  }
}
\makeatletter
\fi
\renewenvironment{proof}[1][\proofname]{%
   \par\pushQED{\qed}\normalfont%
   \topsep6\p@\@plus6\p@\relax
   \trivlist\item[\hskip\labelsep\bfseries{Proof of #1}\@addpunct{.}]%
   \ignorespaces
}{%
   \popQED\endtrivlist\@endpefalse
}

\ifpopets\else
 \makeatother
\date{July 2021}
\fi

\begin{document}

\newif\ifauthors
\authorsfalse
\authorstrue

\ifpopets
\else
\author{
   David M. Sommer\\
   \texttt{david.sommer@inf.ethz.ch}
   \and
   Lukas Abfalterer\\
   \texttt{labfalterer@gmail.com}
   \and
   Sheila Zingg\\
   \texttt{zinggsh@gmail.com}
   \and
   Esfandiar Mohammadi\\
   \texttt{esfandiar.mohammadi@uni-luebeck.de}
 }
\fi

  \title{\huge Learning Numeric Optimal Differentially Private\\ Truncated Additive Mechanisms}

\ifpopets
\else
  \maketitle
\fi

  \begin{abstract}
{
Differentially private (DP) mechanisms face the challenge of providing accurate results while protecting their inputs: the privacy-utility trade-off.
A simple but powerful technique for DP adds noise to sensitivity-bounded query outputs to blur the exact query output: additive mechanisms. 
While a vast body of work considers infinitely wide noise distributions, some applications (e.g., real-time operating systems) require hard bounds on the deviations from the real query, and only limited work on such mechanisms exist. An additive mechanism with truncated noise (i.e., with bounded range) can offer such hard bounds.
We introduce a gradient-descent-based tool to learn truncated noise for additive mechanisms with strong utility bounds while simultaneously optimizing for differential privacy under sequential composition, i.e., scenarios where multiple noisy queries on the same data are revealed.
Our method can learn discrete noise patterns and not only hyper-parameters of a predefined probability distribution.
For sensitivity bounded mechanisms, we show that it is sufficient to consider symmetric and that\new{, for from the mean monotonically falling noise,} ensuring privacy for a pair of representative query outputs guarantees privacy for all pairs of inputs (that differ in one element). 
We find that the utility-privacy trade-off curves of our generated noise are remarkably close to truncated Gaussians and even replicate their shape for $l_2$ utility-loss.
\new{For a low number of compositions, we also improved DP-SGD (sub-sampling).}
Moreover, we extend Moments Accountant to truncated distributions, allowing to incorporate mechanism output events with varying input-dependent zero occurrence probability.
}
\end{abstract}


\redcomment{
\section{Paperarbeiten}
    \begin{itemize}
        \item " Other papers have shown that it is sometimes sufficient to reduce to a non-adaptive sequence of mechanisms that that will apply for the adaptive setting.  However, work from Dong, Durfee, Rogers’20 shows that this does not always hold true and in fact there is a gap between the privacy loss of adaptively chosen and non-adaptively chosen mechanisms from a specific class.  This point should be explicit in the paper with the stated limitation" 
    \end{itemize}
}


\section{Introduction}

Differentially private (DP) mechanisms for queries face the challenge of providing accurate approximations (i.e., high utility) while sufficiently protecting the input data points against any curious recipient of the response.
Many DP mechanisms consider a robust family of queries $q$ (sensitivity-bounded queries) and add noise $N$%
 (additive mechanisms: $D \mapsto q(D) + N$). For additive mechanisms, the noise distribution is independent of the input; hence, additive mechanisms can be easily modified by modifying the noise distribution, and finding strong utility-privacy trade-offs boils down to finding suitable noise distributions.

If combined with range-bounded noise distributions (truncated distributions), additive mechanisms additionally provide strong utility guarantees on the deviation from the query result $q(D)$, e.g., in real-time operating systems or for strengthening DP anonymous communication~\cite{vuvuzela,stadium,karaoke} with bounded latency overhead. 
However, utilizing truncated distributions comes with additional privacy challenges. 
With truncated noise, some outputs can only originate from specific inputs. 
We call all perturbed output events where the attacker can clearly exclude a certain set of inputs distinguishing events. No prior work provides methods for finding truncated additive mechanisms with optimal utility-privacy trade-offs.

In practical applications, privacy has to hold even when an adversary asks several queries on the same input dataset.
There is a line of work on so-called analytical (Moments Accountant~\cite{abadi2016deeplearning}) and numerical (PrivacyBuckets~\cite{meisermohammadi2018PB}) sequential composition bounds that show how differential privacy bounds are amplified if a DP mechanism responds to multiple queries on the same dataset. This line shows that a relaxation of differential privacy, called approximate differential privacy (ADP), leads to stronger sequential composition bounds than pure differential privacy. ADP accepts a (typically very small) error $\delta$ that characterizes the mass that violates the pure $\eps$-differential privacy guarantees. However, 
prior work does not provide a generic method for finding noise distributions with strong utility-privacy trade-offs for ADP under sequential composition.

\myparagraph{Summary of contribution} 
\vspace{-1em}
\begin{itemize}[wide, labelwidth=0pt, labelindent=0pt]
\item
We present a \textbf{gradient-descent-based tool} for learning additive DP mechanisms with strong utility-privacy trade-offs for ADP under sequential composition.
\item 
We prove with several \textbf{theoretical contributions} that learning additive DP mechanisms can be reduced to a \textbf{feasible optimization problem}. One results that might be of independent interest is the \textbf{extension of the Moments Accountant} (MA)
to mechanisms with \textbf{distinguishing events}. MA is formulated as an optimization problem in its moments. 
We identify \textbf{sufficient conditions} on mechanisms under which the optimization problem is feasible.
\new{
\item We learn optimal noise for \textbf{sensitivity-bounded queries and DP-SGD } (sub-sampling) and compare it to \textbf{truncated Gaussian noise}. 
We illustrate that our learned mechanisms have privacy-utility trade-offs that are close to the truncated Gaussian noise and \textbf{replicate} the latter with an increasing number of compositions in the non-sub-sampling scenario when considering an $l_2$ utility-penalty.
}
\end{itemize}

\myparagraph{Our contribution in detail}\fv{We divide our contribution into six parts.}
\begin{enumerate}[wide, labelwidth=0pt, labelindent=0pt]
%
\item We developed an optimization engine that learns differentially private additive noise distributions for sensitivity-bounded queries and an even more general set of queries, so-called queries with worst-case pairs of output distributions (intuitively, queries for which it suffices to analyze a pair of output distributions to derive privacy-guarantees for all inputs), while also maximizing utility.
%
\item We reduce the problem of learning additive mechanisms with strong utility-privacy trade-offs to learning noise distributions for a specific pair of outputs via several theoretical contributions. 
We proved that for sensitivity-bounded queries with noise distributions that are symmetric and monotonically falling from the mean, it suffices to validate ($\eps,\delta$)-ADP guarantees for pairs of neighboring inputs (differing in one element) where the unnoised response has the maximal distance (i.e., sensitivity). We show that such a validation already implies ($\eps$,$\delta$)-ADP for any query response for neighboring inputs.  
\item 
We derive an analytical bound for ($\eps$,$\delta$)-ADP under sequential composition and higher numbers of compositions. 
The original Moment Accountant (MA)~\cite{abadi2016deeplearning} is currently \new{one of} the best known analytical bounds and relies on minimizing the log-moment generating function $\alpha$ that characterizes the privacy-losses. This approach is, however, inherently incompatible with mechanisms that have distinguishing events. Incorporating distinguishing events poses two problems: 1) There are two partially dependent privacy leakage variables: the distinguishing events and $\alpha$. 2) The naïve approach to consider them separately distorts the privacy-leakage characterization of $\alpha$ due to the possibly unbalanced input-distribution renormalization when excluding distinguishing events. We solved these challenges for mechanisms with a worst-case reduction and provided an extended MA, able to handle truncated mechanisms. 
\item We implement PrivacyBuckets~\cite{meisermohammadi2018PB} our proven extension of MA as a differentiable function in PyTorch to estimate ($\eps$,$\delta$)-ADP under sequential composition. 
To achieve a strong utility-privacy trade-off for a given $\eps$, our optimization engine minimizes the sum of two loss terms: the numerical upper bound for $\delta$ and either a $L_1$-utility-loss (the mean absolute deviation) or a $L_2$-utility-loss (the standard deviation) of the noise distribution.
\new{We have published our source code on GitHub.}

\item 
We provide strong evidence for low numerical errors by showing that without composition, our tool finds a truncated (and improved) version of the staircase mechanism, which has been proven to be optimal in the non-truncated case. With a high number of sequential compositions, our tool finds the shape of truncated Gaussians when considering a $L_2$-utility-loss (maximal KL-divergence of $10^{-4}$ for $128$ compositions) while for $l_1$-loss it approaches the privacy-utility curve of the latter. With our Extended MA bound for a low number of compositions, we outperformed truncated Gaussians.
Our results suggest a strong utility-privacy trade-off for the truncated Gaussian mechanism.
%
\item 
We apply our engine to learn noise for DP-SGD by using the same worst-case reduction as in the work of Abadi et al.~\cite{abadi2016deeplearning}. 
Our experimental results suggest that also here, the truncated Gaussian distribution provides a strong utility-privacy trade-off.
\end{enumerate}

\section{Background}
    
This section covers the background required for our contributions in \cref{sec:approach,sec:results}. First, we introduce differential privacy with its variants and sequential composition results. Second, we discuss additive noise, truncated Gaussian distributions, and counting queries before concluding with a description of utility and optimal noise.

\subsection{Differential Privacy}

To quantify the privacy of a mechanism, Dwork et al.~\cite{dwork2006our} proposed a strong privacy notion, called $(\eps,\delta)$-approximate differential privacy (ADP). Given a privacy-preserving mechanism, this definition argues about the maximal output event probability deviation of that mechanism when comparing the result of any two \textit{neighboring} inputs, rendering the contribution of an individual plausibly deniable.
While the exact meaning of \textit{neighboring} depends on the specific task, it can be understood as two inputs that differ only by the contribution of a single individual. 

\begin{definition}[Approximate Differential Privacy~\cite{dwork2006our}]\label{definition:adp}
A randomized algorithm $M$ with domain $\mathcal D$ is ($\eps$, $\delta$)-approximate differentially private (($\eps$, $\delta$)-ADP) if for all $\mathcal S \in \text{Range}(M)$ and for all neighboring $D_0, D_1\in \mathcal D$:
\begin{equation*}
 \pr{M(D_0) \in S} \leq e^\eps \pr{M(D_1) \in S} + \delta.
\end{equation*}
where the probability space is over the coin flips of the mechanism $M$. If $\delta$ = 0, we say that $M$ is $\eps$-differentially private.  
\end{definition}
ADP guarantees that no post-processing can deteriorate its privacy guarantees~\cite{dwork2014algorithmic}, as long as no additional knowledge about the inputs is incorporated~\cite{narayanan2006break}.

For probabilistic differential privacy (PDP), a more intuitive formulation, this does not hold.

\begin{definition}[Probabilistic Differential Privacy~\cite{GoMaWaXiGe_09pdp}]\label{definition:pdp}
A randomized algorithm $M$ with domain $\mathcal D$ is ($\eps$, $\delta$)-probabilistic differentially private (($\eps$, $\delta$)-PDP) if for all neighboring $D_0, D_1\in \mathcal D$ we can divide the output space in two sets $O,S$ such that $\forall o\in O$
\begin{align*}
    \pr{M(D_0) = o\,} \leq e^\eps \cdot \pr{M(D_1) = o\,} \\
    \text{and\quad} \pr{M(D_0) \in S\,} \leq \delta
\end{align*}
\end{definition}



\subsubsection{Worst-case, Privacy Loss, and Distinguishing Events}

Classically, differential privacy argues about the outputs of a probabilistic mechanism $M$ that runs on similar (neighboring) inputs. In literature, however, a reduction to so-called \textit{worst-case output distributions}, i.e., a pair of mechanism output distributions $M(D_0)$ and $M(D_1)$ such that no pair of inputs induces more privacy leakage, is common to simplify the privacy analysis~\cite{dwork2014algorithmic,kairouz2017composition,MuVa_16:psharp,abadi2016deeplearning,mironov2017renyidp,CDPBun}.
While this formalization is unconventional and, at first glance, seems to restrict the applicability to particular queries, this approach leads to far more general results~\cite{sommer2019privacy}.
\fv{For example, analyzing the approximate randomized response (ARR) mechanism, i.e., analyzing two worst-case output distributions parametric solely in a ($\eps_j,\delta_j$) pair, exactly yields optimal mechanism-oblivious bounds~\cite{kairouz2017composition,MuVa_16:psharp}.}
In particular, we are often interested in quantifying the privacy of a particular mechanism under composition instead of the privacy of adversarially chosen mechanisms. 
Recent results show that better fitting worst-case distributions can lead to significantly tighter privacy bounds under composition~\cite{CDPBun,DwRo_16:concentrated,mironov2017renyidp,abadi2016deeplearning,meisermohammadi2018PB}. 
These methods started to more intensely use the \emph{privacy loss} of a mechanism%
\fv{~that has been proposed by a seminal work by Dinur and Nissim}%
~\cite{DiNi_03:confidence-gain}. 
Sommer et al.~\cite{sommer2019privacy} extended this definition, allowing to include mechanism output events with zero occurrence probability.
\begin{definition}[Privacy Loss Random Variable~\cite{sommer2019privacy}]\label{def:privacyloss}
Given a probabilistic mechanism $M: \mathcal D \rightarrow \mathbb{R}$, let $o \in \mathbb{R}$ be any potential output of $M$ and let $D_0,\,D_1\in\mathcal D$ be two inputs. We define the \emph{privacy loss} random variable of an output $o$ of $M$ for $D_0,\,D_1$ as\renewcommand{\drawXfromY}[2]{\pr{#2=#1}}
 \begin{align*}
&\loss{M(D_0)}{M(D_1)}{o}\\
&\hspace*{0.5em}=\begin{cases}
 \infty & \text{if } \drawXfromY{o\,}{M(D_0)} \!\neq\! 0\text{ and } \\
 & \drawXfromY{o\,}{M(D_1)} \!=\! 0 \\
 \ln\left(\frac{\drawXfromY{o}{M(D_0)}}{\drawXfromY{o}{M(D_1)}}\right) & \text{if }\drawXfromY{o\,}{M(D_i)}\! \neq \!0\,\,\,\forall i\!\in\!\{0,\!1\}\\
 \minfty & \text{else}
 \end{cases}
 \end{align*}
 where we consider $\infty$ and $\minfty$ to be distinct symbols.
\end{definition}

For a specific worst-case mechanisms $M(D_0)$, the events with privacy loss $\loss{M(D_0)}{M(D_1)}{o} = \minfty$ do never occur. The events with infinite privacy loss $\loss{M(D_0)}{M(D_1)}{o} = \infty$, however, reveal immediately that input $D_0$ is used for the mechanism $M$, which differential privacy tries to hide. Therefore, we coin these events \textit{distinguishing events}.
\begin{definition}[Distinguishing Events]
Let $D_0$,$D_1$ be worst-case inputs for a mechanism $M:\mathcal D\rightarrow \mathbb{R}$. Then,
\textit{distinguishing events} are defined as all events $o$ in 
\begin{equation*}
    \{o\mid \loss{M(D_0)}{M(D_1)}{o} = \infty,\,o\in \mathbb{R}\}.
\end{equation*}
\end{definition}

\subsubsection{Composition}\label{sec:background:DP:composition}

Previous definitions consider an adversary who has observed one mechanism output only. However, when a mechanism is regularly used, an adversary may encounter multiple (sequential) outputs originating from the same input $D$. Intuitively, privacy degrades with increasing numbers of collected observations. Obtaining a tight bound for privacy leakage under sequential composition is of essence as loose estimations force to weaken the mechanism to still fulfill the overly conservative privacy requirements while having a devastating effect on the usefulness of the obtained output.
While there are bounds considering adaptive mechanisms~\cite{abadi2016deeplearning,kairouz2017composition,dwork2010boosting}, this work considers only independent composition because adaptive mechanisms are often reduced to non-adaptive ones for simplifying the privacy analysis.

Abadi et al.~\cite{abadi2016deeplearning} introduced \textit{Moments Accountant} that allows simple composition of adaptive mechanisms and an upper bound of the resulting ADP-guarantees. To achieve this, they have shown that it is sufficient to bound the moment-generating function of the privacy-loss random variable.


\begin{lemma}[Moments Accountant~\cite{abadi2016deeplearning, mironov2017renyidp}]\label{lemma:moments-accountant}
For $\lambda > 0$, let $\alpha_{M}(\lambda)$ be the logarithm of the maximal moment-generating function of the privacy-loss random variable generated by the mechanism $M$ for any auxiliary input $aux$ and any neighboring inputs $D,D'$.
    \begin{enumerate}
        \item \textbf{Composability:}
            Suppose that a mechanism $M$ consists of a sequence of adaptive mechanisms $M_1,\ldots,M_k$ where $M_i:\left(\prod_i^{j=1} R_j\right)\times \mathcal D \rightarrow R_i$. Then, for any $D_0, D_1\in \mathcal D$, $\lambda > 0$
            \begin{equation*}
                \alpha_M\left(\lambda\right) \leq \sum_{i=1}^k \alpha_{M_i}\left(\lambda \right)
            \end{equation*}
        \item \textbf{Tail-bound:}
            For any $\eps >0$, the mechanism is $(\eps, \delta)$-ADP for 
            \begin{equation*}
                \delta = \min_\lambda \exp\left(\alpha_M\left(\lambda\right)-\lambda \eps\right)
            \end{equation*}
    \end{enumerate}
\end{lemma}

Providing tighter bounds for ADP and PDP, Meiser et al. introduces PrivacyBuckets~\cite{meisermohammadi2018PB}, a numerical tool for computing tight ADP- and PDP-bounds for sequential and independent composition, which was later refined and analytically consolidated by Sommer et al.~\cite{sommer2019privacy}. Their approach utilizes the distribution of the privacy-loss random variable~\cite{sommer2019privacy} and the fact that independent sequential composition can be expressed as a convolution of its (privacy-loss) distribution, obtaining a new privacy-loss distribution representing the loss of two adversarial observations, and deriving upper bounds for ADP and PDP guarantees from it.
Their work assumes that the analyzed mechanism can be reduced to a worst-case distribution that captures the worst privacy leakage possible. For most mechanisms, this assumption is feasible.
We will use an implementation of PrivacyBuckets by Sommer\footnote{\label{footnote:privacybuckets}Available at \url{https://github.com/sommerda/privacybuckets}} and denote the corresponding upper-bound for a given $\eps$ by $\delta^{PB}$.

\subsection{Sensitivity and Additive Noise}

In this work, we search optimal noise for mechanisms with additive noise. To formally introduce the required notation, we first define a noise function which is in fact a probability density function (pdf):
\new{
\begin{definition}[Noise Function and Sampling] \label{def:noiseFunc}
Let $X=\{x_i\}_{i\in\mathbb{Z}}$ with $x_i\scriptstyle<\textstyle\! x_j$ for $i\scriptstyle<\textstyle\!j$ be a discretization of $\mathbb{R}$. 
 Then $p:X\longrightarrow [0,1]$ is a noise function if $\forall x\in X$ 
 \begin{gather*}
  0 \leq p(x) \text{~~~~and~~~~}
  \sum\nolimits_{x\in X}\, p(x) = 1
 \end{gather*}
We sample from $p$ by first sampling $x_i \in X$ according to $p(x)$ and then sampling uniformly  from $[\scriptstyle\frac{x_{i-1} + x_i}{2},\, \frac{x_{i} + x_{i+1}}{2}\textstyle)$.
\end{definition}
}
We introduce the later frequently referenced truncated Gaussian noise as an example.
\begin{example}[Truncated Gaussian Noise]\label{def:truncated-gaussian} A symmetric, zero-centered and at distance $r$ from zero truncated Gaussian with variance $\sigma^2$ is defined as 
\begin{equation*}
p(x)dx = \indicatorfunc{|x| \leq \rangebegin} \frac{1}{\sqrt{2\pi\sigma^2}N}  \exp\left(-\frac{x^2}{2\sigma^2}\right)\,dx 
\end{equation*}
where $N \!=\! \erf\!\left(\!\frac{r}{\sqrt{2}\sigma}\!\right)$ is the normalisation constant with $\erf$ as the error function and $\mathds{1}$ as the indicator function.
\end{example}
Note that our numerically generated noise functions are discrete, i.e., the integral is replaced with a sum. Furthermore, the generated noise functions are generally truncated, which is reflected by $p(x) = 0$ beyond a certain distance from the center. 
We now define additive noise mechanisms that add noise to a processed output $q(D)$ for a dataset $D\in\mathcal D$.
\begin{definition}[Additive Noise Mechanism]\hfill\\
Let $q:\mathcal D \rightarrow \mathbb{R}$ and $p$ be a noise function. Then, an additive noise mechanism is defined as
\begin{equation*}
    M_q(D;p) := q(D) + x,\quad\text{with}\quad x \sim p
\end{equation*}
In abuse of notation, we write $M_q(D; x) = q(D) + x, x \in \mathbb R$, enforcing a deterministic output of $M_q$ and coin it \textit{deterministic additive noise mechanism}.
\end{definition}
A simple example are counting queries that have been shown to be differentially private when noising them with Laplace noise (noise $\sim \frac{\eps}{2} e^{\eps|x|}$)~\cite{dwork2006calibrating}.
\begin{example}[Counting Queries]
Let $q:\mathcal D\rightarrow \mathbb N$ count the occurrences of an attribute in dataset $D\in\mathcal D$. 
\end{example}
By definition, counting queries from two neighboring datasets, i.e., two datasets that differ only in one record, can only differ by $0$ or $1$. 
In literature, privacy guarantees are often assigned under the constraint that the processing result of two neighboring inputs does not deviate more than a specific value and that less deviation does not incur more privacy loss~\cite{dwork2006calibrating,abadi2016deeplearning,dwork2014algorithmic}. This maximal deviation is called \textit{sensitivity}.
\new{
\begin{definition}[Sensitivity] \label{sensitivity}
Let $q: \mathcal{D} \rightarrow \mathbb{R}^d$ be a real-valued query function. Its sensitivity $s$ is defined as
\begin{align*}
 s = \max_{\forall~\text{neighboring}~D_1,D_2 \in \mathcal{D}} ||q(D_1)-q(D_2)||_2
\end{align*}
\end{definition}
}

\subsection{Utility and Optimal Noise}

Utility describes the closeness of the randomized response from a mechanism $M(D)$ to the true result $q(D)$. For any $D$, a large deviation from $q(D)$ degrades the usefulness of the mechanism. 
Previous work formalized utility by a cost-function penalizing such deviations~\cite{geng2017staircase}. 
\new{
\begin{definition}[Utility-loss]\label{definition:utility}
Let $M_q: \mathcal D \times \mathcal R \rightarrow \mathbb R$ be a deterministic additive noise mechanism with noise function $p$, let $u:\mathbb R\rightarrow \mathbb R$ be a cost function penalizing deviations from $q(D)$. Let $\mathcal D$ be the space of all possible databases. Then, a \textit{utility-loss function} is defined as 
\begin{equation*}
\utilityloss{}(p; u, M_q, \mathcal D) = \sup_{\mathcal D}\int_{\mathcal R} u(M_q(D, x), q(D)) \,p(x) dx
\end{equation*}
\end{definition}
}
Note that the utility-loss considers only the deviation from $q(D)$, and is equal for all $D\in\mathcal D$. The domain $\mathcal R$ can be seen as a randomness space and, thereby, allowing the definition to be applicable to a much wider class of mechanisms than just additive noise. In this work, however, we are applying only $L_1$ and $L_2$ losses, effectively reducing the definition to $U_{L_1}(p) = \expectation{p}\!\left[|x|\right]$ and $U_{L_2}(p) = \expectation{p}\!\left[|x|^2\right]$.




Utility and privacy maximization stand in direct conflict with each other. The mechanism $M_q$ delivering the highest possible utility is returning the result $q(D)$ directly without noise, neglecting any privacy concerns. On the other hand, a mechanism $M$ guaranteeing high privacy \new{might} need to distort the output close to complete uselessness. Thus, we are interested in the noise that maximizes utility while not violating given $(\eps,\delta$)-differential privacy constraints. As we focus on truncated noise distributions, we encounter necessarily distinguishing events and, thereby, a non-zero $\delta$. 




Previous work has formalized optimal noise as a minimization problem, optimizing the utility-loss as much as possible while fulfilling pure $\eps$-DP ($(\eps,0)$-ADP)~\cite{geng2015optimal}. As our first contribution, we extend this definition to include a non-zero $\delta$.
\begin{definition}[Optimal Noise]\label{def:optmal-noise} \label{definition:optimal-noise}
Let $U$ be a utility-loss function and $M_q:\mathcal D\times \mathcal P \rightarrow \mathbb R$ a randomised mechanism. Then, for any $\eps, \delta > 0$, optimal noise is defined as
\begin{gather*}
    \hat{p} = \min_p \text{~} \utilityloss(p; u, M_q, \mathcal D) \\
    \text{where $M_q(D;p)$ is $(\eps, \delta)$-ADP.}
\end{gather*}
\end{definition}
\fv{Note that our optimization algorithm that we introduce later does not take $\eps$ and $\delta$ as input but balances between utility-loss and $\delta$ for a fixed $\eps$.}

\section{Approach \& Theoretic Results}\label{sec:approach}
This section presents our approach to finding optimal noise numerically and elaborating on our theoretical results. First, we prove that it is sufficient for many cases to assume symmetric and monotonic noise distributions (\cref{sec:approach:monotonicity}).
For a fixed $\eps$, we aim to minimize a weighted sum of utility-loss and the resulting privacy parameter $\delta$ by searching a suitable noise $p$ with gradient descent (\cref{sec:approach:gradient-descent}). We provide three numerical and derivable upper bounds for the minimizer: our Extended Moments Accountant incorporating distinguishing events (\cref{sec:approach:RDP}) and an ADP and PDP bound based on PrivacyBuckets (\cref{sec:approach:adp-bound,sec:approach:pdp-bound}). 
Any proofs and the extensive formulation of the Extended Moments Accountant have been deferred to the appendix. 

\subsection{Monotonicity \& Symmetry Assumptions}\label{sec:approach:monotonicity}
 
 
 
         

Now, we introduce simplifying assumptions about the additive optimal noise-shape and show them to be feasible.
Under the reasonable assumption that the utility-loss $u$ is symmetric, we will show that it is sufficient to consider only symmetric noise distributions \new{and argue that we only need to consider monotonically falling distributions.}

Intuitively, the symmetry property of optimal noise originates from the fact that we do not know in what direction the difference in outputs of $q(D)$ and $q(D')$ will occur. In compliance with the information-theoretic guarantees of differential privacy, we need to protect both equally. Hence, the symmetry. 
For our case, it is sufficient to show that there exists a symmetric noise distribution producing the same utility-loss as an asymmetric one while fulfilling the same ADP-guarantees and, thereby, be equally considerable optimal noise. 
\begin{lemma}\label{lemma:symmetric-noise}
 Let $u$ be a cost function symmetric in $x$. Let $\eps,\delta>0$. Given a noise distribution $p$ satisfying ($\eps,\delta$)-ADP, there exists a symmetric noise distribution $\hat{p}$ satisfying ($\eps,\delta$)-ADP and 
 \begin{equation*}
     \utilityloss(p; u, M_p, \mathcal D) = \utilityloss(\hat{p}; u, M_p, \mathcal D)
 \end{equation*}
\end{lemma}

\new{
We argue now that it is sufficient to consider piece-wise continuous (c.f. \cref{def:noiseFunc}) and monotonic noise. Geng et al. have shown that, for a given $\delta$, the optimal noise for \textit{a single} mechanism invocation is monotonic on all measurable sets with a potential point mass at $0$~\cite{geng2019pmlr}, maximizing utility by potentially returning the true query result in certain settings. However, in our setting where we aim to minimize $\delta$ together with the utility-loss, such a centered point mass is detrimental as the impact of $\delta$ gets strongly amplified under composition, see \cref{sec:aproach:privacybounds} for mathematical descriptions. Furthermore, our monotonic noise without a centered point mass still allows for a large mass at the center if preferred by the optimizer. 
Besides the center, a monotonic noise distribution is definitely optimal because non-monotonicity implies that there exist two query output differences $o, o'$ with $|o|<|o'|<s$ where $o$ leads to higher $\delta$ for a fixed $\eps$ than $o'$ (compared to the center) while $o'$ is occupying more probability mass, effectively worsening utility for the same $\delta$.
}

The optimal noise $p$ we search for must provide differential privacy guarantees to $M_q$ for any deviations $q(D)$ - $q(D')$ considering neighboring $D,D'$. While a numerical approach could check the conditions for any (discrete) deviation, the computational complexity increases linearly with the number of discretization steps.
However, if we assume from the center monotonically falling noise and two inputs $D,D'$ with a deviation smaller than the maximal sensitivity $s$, then the two identically shaped but differently centered output distributions produced by $M_q$ cannot induce a higher privacy leakage than the one occurring by inputs that deviate by $s$ because for any output events $o$, the privacy-loss is smaller or equal due to monotonicity. Therefore, it is sufficient to check for DP-abidance only at the maximal deviation~$s$. In combination with the previously shown symmetry property, such noise needs to be centered at zero.

\new{
Before we prove this simplification, we introduce a technical lemma proving that considering only the discrete noise occurrence probabilities $p(x)$ is sufficient to guarantee ADP for any query output $q(D)$ and any drawn noise in $\mathbb R$. This effectively connects our discrete numerical search to the continuous domain. 
While we only show the claim for ADP, we conjecture its validity for PDP as well.
\begin{lemma}\label{lemma:discretization-reduction} 
Let $M_q$ be a additive mechanism with noise $p$ and sensitivity $s$. Let $X=\{x_i\}_{i\in \mathbb Z}$ be a equidistant and ordered (cf. \cref{def:optmal-noise}) discretization of $\mathbb R$ with $s= x_i - x_j$ for some $i,j$. Then, for any $r\in \{x_i - x_j | i,j\in\mathbb Z\},\, r\leq s$
and any $S\subseteq X$
\begin{equation*}
    \sum_{x\in S} p(x) \leq \delta + e^\eps \sum_{x\in S} p(x+r) 
\end{equation*}
implies that $M_q$ is $(\eps, \delta)$-ADP for any query output in $\mathbb R$.
\end{lemma}
}

Now, we prove the sensitivity-reduction property for a 1-dimensional noise distribution as one of our main theoretic contributions. 
Note that it is straightforward to generalize the claim to arbitrarily dimensional and spherically rotation-symmetric noise distributions and sensitivity conditions because these problems can be reduced to a 1-dimensional privacy analysis as, e.g., Abadi et al.~\cite{abadi2016deeplearning} have shown. While we only show the claim for ADP, we conjecture its validity for PDP as well.
\redcomment{Besserer Name als Shift Invariance}
\begin{theorem}[Shift Invariance]\label{thm:shift-invariance}
Let $p$ be a symmetric and from $0$ monotonically decreasing noise function, i.e., for $0<r'<r$,
\begin{align*}
    p(- r) = p(r) &\text{~~~(symmetry)}\\
    p(r') \geq p(r) &\text{~~~(monotonicity)}
\end{align*}
Let $M_q$ be a additive noise mechanism with sensitivity $s$. If $\forall S \subseteq X$
\begin{equation*}
\sum_{x\in S} p(x) \leq \delta + e^\eps \sum_{x\in S} p(x+s) 
\end{equation*}
Then $M_q(\!D\!)$ is $(\eps,\delta)$-\!ADP for any input $D\!\in\!\mathcal D$.
\end{theorem}

\subsection{Optimal Noise by Gradient Descent}\label{sec:approach:gradient-descent}
Gradient descent locally minimizes a differentiable function by stepping in the opposite direction of its gradient.
We search for optimal noise $p$ by minimizing the utility-loss $\utilityloss$, i.e., minimizing the dispersion of additive noise while simultaneously punishing the algorithm for privacy leakage. We fix the maximal privacy-loss we allow without contributing to $\delta$, represented by the parameter $\eps$, and then minimize the incurring utility-loss together with a numerical upper bound $l^X$ on the $\delta$ originating from the generated noise $p$. We provide three different methods for $l^X$ later in this section, namely $l^{MA}$, $l^{ADP}$, and $l^{PDP}$. 
Specifically, we apply gradient descent with the following loss function
\begin{equation}\label{sec:approach:lossfunction}
        \LossWeMinimize_{\mathcal D,M_q}^{\utilityweight_t,\eps,n}(p) = l^X(I_0^p, I_1^p;\eps, n) + \utilityweight_t\!\cdot\utilityloss(p; u, M_p, \mathcal D) 
\end{equation}
where the utility-weight $\utilityweight_t$ is a weighting coefficient which might depend on the training epoch $t$, and $n$ denotes the number of mechanism invocations we consider. The distributions $I_0^p$ and $I_1^p$ are worst-case distributions of $M_q$, resulting in maximal privacy leakage after adding noise. For many applications, such worst-case distributions are known as, e.g., for counting queries where $I_0^p=p$ and $I_1^p=p_{\text{shifted-by-1}}$.  
\Cref{thm:shift-invariance} has shown that for symmetric and monotonic noises $p$, such $I$ exist.

To achieve as much generality as possible, we do not choose a well-known probability distribution and optimize its hyper-parameters. Instead, we use trainable noise-models that provide a discrete probability density function $p$ at predefined discretization steps $X$ with sufficiently high granularity.
Including all assumptions, we obtain for $\lambda\in\{1,2\}$ referring to $L_1$ or $L_2$ utility-loss
\begin{equation}\label{eq:approach:minimizing-equation}
    \LossWeMinimize_{\mathcal D,M_q,\lambda}^{\utilityweight_t,\eps,n}(p) = l^X(I_0^p, I_1^p;\eps, n) +  \utilityweight_t\cdot \left(\sum_{x \in X} |x|^\lambda p(x) \right)^{\!\frac{1}{\lambda}}\!\!
\end{equation}
\new{In accordance with \cref{eq:approach:minimizing-equation}, we linearly interpolate the cost-function $u$ between discretization steps $x$ to reduce computational complexity.} 

\myparagraph{Utility-weight decay}
Fixing the utility-weight $\utilityweight_t$ often resulted in a large utility-loss, dominating the corresponding privacy-loss $\delta$ enough to let the optimizer prefer a centrally collapsed noise distribution with all mass in the center. Such a distribution (illustrated in \cref{fig:utility_domination}) has almost no utility-loss but a $\delta=1$ which goes against the spirit of differential privacy. We have found to achieve more reliable results when we let the utility weight decay exponentially from a starting value $\utilityweight_{start}$ with the rate $\utilityweightdecayrate$, especially for a higher number of compositions. A lower bound $\utilityweight_{min}$ limits the decay.
    \begin{equation}\label{eq:utility_weight_decay}
            \utilityweight_\currentEpoch = \max\left(\frac{\utilityweight_{start}}{2^{\currentEpoch/\utilityweightdecayrate}}, \utilityweight_{min}\right)
    \end{equation}

    


\myparagraph{Convergence} 
For from the center monotonically falling noise, we argue that the loss in \cref{eq:approach:minimizing-equation} has a convex minimum. The monotonic noise can be characterized solely by the steepness of the noise function gradients. On the one hand, the applied utility-loss is convex in that steepness. On the other hand, we argue that a tight $\delta$-bound is also convex in said steepness. There are two extreme cases: first, all probability mass is concentrated at the center ($\delta^{PB}\! =\! 1$), and second, a horizontal line induces maximal distinguishing events at the outermost regions ($\delta^{PB}\! =\! \bucketSymbol_\infty$). Starting from the latter, an increasing steepness first monotonically decreases $\delta^{PB}$ because distinguishing events shrink and the privacy-loss approaches $\eps$ (monotonically reducing $\delta$), before the noise is too steep, diverging from the minimum again towards the other extreme case. While $l^X$ are upper bounds for such a tight $\delta$, at least Meiser et al. have shown that the infinite limit tightens their bound.
Finally, the sum of two convex functions is convex again.


\subsection{Noise Model}\label{sec:approach:noise-model}
Our examinations have shown that the noise generating model requires strong dependence between neighboring discretization-steps $x_i$ on the x-axis. Therefore, we decided to generate the first, monotonically increasing half of the noise $p$ from a model of $K$ stacked Sigmoid functions $\sigma(x) = \left(1+e^{-x}\right)^{-1}$. Their output is then normalised by a $\SoftMax(r_i;\overrightarrow{r}) = e^{r_i} / \sum_j e^{r_j}$ step before being mirrored and concatenated to obtain symmetric noise. We evaluate this model on $2N$ discretised and equidistant steps $x_i$ on the x-axis, $x_i\in [-\rangebegin, \rangebegin], i\in \{1,\ldots, 2N\}$, $x_i = i\cdot\frac{r}{N}-\rangebegin$. Due to mirroring, we only need to consider $x_i <0$. Our model $p$ is composed as follows: for $i\in\{1,\ldots N\}$, 
    \begin{equation}\label{eq:model}
        \begin{aligned}
            r_i &= \ln \left[\bias^2 + \sum_{j=0}^{K} \scales_j^2\cdot \Sigmoid\left(\slopes (x_i-\sigmoidc_j)\right)\right]\\
            p_i &= \frac{1}{2}\SoftMax\left(r_i; \{r_0, \ldots, r_N\}\right)\\
            p_j &= p_{2N-j+1}\quad\text{for } j\in\{N\!+\!1,\ldots, 2N\}
        \end{aligned}
    \end{equation}
The parameters $\bias, \scales_j,$ and $\sigmoidc_j\in\mathbb{R}$ are learned by gradient descent while the slope $\slopes\in\mathbb{R}$ is a fixed hyper-parameter, usually set to $500$ to allow sudden jumps of the noise $p$. We have initialised $\bias$ with $10$, the $\scales_j$ are drawn uniformly from $[0,1]$, and all $\sigmoidc_j$ are initially equidistantly distributed in $[-\rangebegin, 0]$ where $\rangebegin$ is the half-width of the noise, more distant from the center we truncate. By squaring $\bias$ and $\scales_j$, we enforce a monotonic increasing output. By applying $\SoftMax$, we obtain $\forall i,\,\,p_i >0$  and $\sum_i p_i = 1$. 

This model cannot support perfectly vertical jumps, which poses an important limitation. However, we avoid these effects by choosing a high slope $\slopes$ such that the output can sufficiently change between neighboring discretization-steps $x_i$. 

\subsection{Privacy Bounds}\label{sec:aproach:privacybounds}
We now introduce three different numerical upper differential privacy bounds for $l^\delta(I_0^p, I_1^p;\eps, n)$, namely $l^{MA}$ through our Extended Moments Accountant, and $l^{ADP}$ and $l^{PDP}$ via PrivacyBuckets.


\subsubsection{Extended Moments Accountant}\label{sec:approach:RDP}

Frequently used when applying differential privacy to Deep Learning, Moments Accountant (MA) (see \cref{lemma:moments-accountant}) and its ADP bound is applied. However, the formulation by Abadi et al.~\cite{abadi2016deeplearning} does not incorporate the existence of distinguishing events. Their theorem assumes implicitly infinitely wide noise that nowhere has zero occurrence probability. To remedy this shortcoming and to allow the analysis of truncated noise, we introduce our second major theoretical contribution, the Extend Moments Accountant, that incorporates the existence of such distinguishing events. Subsequently, we derive our numerical ADP-bound $l^{MA}$.

The extension of MA requires a more complex reduction to worst-case output distributions of $M$, which simultaneously fulfill the maximization of the distinguishing events $\bucketSymbol_\infty$ and the maximization of the corresponding moments $\alpha(\lambda)$. However, such worst-case output distributions $M(D_0)$ and $M(D_1)$ usually exist and are applied often to simplify the privacy analysis. The use of such worst-case distributions renders the theorem itself adaptive (the mechanism is allowed to use previous outputs) as no other inputs can achieve a higher privacy leakage by definition. Abadi et al. themselves reduce the privacy analysis of their differentially private stochastic gradient descent algorithm (DP-SGD) by comparing a Gaussian to a Gaussian-mixture distribution. For brevity, we show a simplified version of our extended MA theorem here and refer to \cref{sec:full-extended-MA} for the fully defined version and a discussion of its formulations.

\newcommand{\modifiedminus}{\text{-}}
\begin{theorem}[Extended MA {[informal]}]\label{thm:extended_ma:informal} 
For a mechanism M and $\forall \lambda > 0$, let $k \in \{0, 1\}$ (the $l$ might change depending on $\lambda$), s.~t. $\alpha_{M}^{D_k,D_{1\modifiedminus k}}(\lambda)$ and $\bucketSymbol_{\infty,M}^{D_k,D_{1\modifiedminus k}}$ dominate all $\alpha_{M}^{D,D'}(\lambda)$ and $\bucketSymbol_{\infty,M}^{D,D'}$ for all neighboring $D,D'\in\mathcal D$. Then, 
\begin{enumerate}
    \item \textbf{[Composability]} Suppose that a mechanism M consists of a sequence of adaptive mechanisms $M_1,\ldots,M_n$ where $M_i=\prod_{j=1}^{i-j} R_j\times \mathcal D\rightarrow R_i$. Then, for any $\lambda > 0$ 
    \begin{align*}
        \alpha_M^{D_k,D_{1\modifiedminus k}}(\lambda) &\leq \sum_i^n \alpha_{M_i}^{D_k,D_{1\modifiedminus k}}(\lambda)\\
        \bucketSymbol_{\infty,M}^{D_k,D_{1\modifiedminus k}} &= 1 - \prod_i^n \left[1-\bucketSymbol_{\infty,M_i}^{D_k,D_{1\modifiedminus k}}\right]
    \end{align*}
    \item \textbf{[Tail Bound]} For any $\eps>0$, $M$ is $(\eps, \delta)$-differentially private for $\delta = \max(\delta_{D_0, D_1}, \delta_{D_1, D_0})$ with
    \begin{equation*}\hspace{-1.8em}
        \delta_{D_k, D_{1\modifiedminus k}} \!=\!\bucketSymbol_{\infty,M}^{D_k,D_{1\modifiedminus k}}\! + \min_\lambda \left(1\!-\!\bucketSymbol_{\infty,M}^{D_k,D_{1\modifiedminus k}}\right)\cdot e^{ \left(\alpha_M^{D_k,D_{1\modifiedminus k}}(\lambda) - \lambda\cdot\eps\right)}
    \end{equation*}
\end{enumerate}
\end{theorem}
While this theorem formally proves ADP-guarantees, the tail estimation of the privacy leakage is obtained by the Markov inequality, actually producing a PDP-guarantee (which itself implies ADP) identically to the original theorem by Abadi et al.~\cite{abadi2016deeplearning}. For certain optimization scenarios shown later, the noise obtained will resemble a result from the PrivacyBuckets PDP-bound more than the PrivacyBuckets ADP-bound. 

From his theorem, we derive our fully differentiable and gradient descent suitable upper bound $l^{MA}$ for two worst-case mechanism output distributions $M(D_0)$ and $M(D_1)$. 

\newcommand{\Alphadiv}[3]{\Gamma_{\!\!#1}\left(#2||#3\right)}
\newcommand{\mypr}[1]{\pr{o\leftarrow {#1}}}

\begin{proposition}[Numerical Extended MA]\label{lemma:delta-from-rdp}\hfill\\
    Let $M(D_0)$ and $M(D_1)$ be worst-case mechanism output distributions. With the number of compositions $n\geq 1$, $\eps > 0$, and 
    \begin{align*}
        \bucketSymbol_\infty^{A,B} &= \pr{\loss{A}{B}{o} = \infty}\\
        \Alphadiv{\lambda}{A}{B} &= \log\hspace{-2.5em} \sum_{\{o\, | \mypr{A}, \mypr{B} \neq 0\}}\hspace{-2.5em} \mypr{A} \left(\frac{\mypr{A}}{\mypr{B}}\right)^\lambda\\
        \delta_{A,B}^{MA}(\eps) &= \min_\lambda 1 - \left[1- \bucketSymbol_{\infty}^{A,B}\right]^n +e^{ n\cdot\Alphadiv{\lambda}{A}{B} - \lambda\cdot\eps}, 
    \end{align*}
    the mechanism $M$ is $(\eps, \delta)$-ADP with 
    \begin{align*}
        \delta = \max(\delta_{M(D_0), M(D_1)}^{MA}, \delta_{M(D_1), M(D_0)}^{MA})
    \end{align*}
\end{proposition}


\myparagraph{Implementation of $l^{MA}$} We compiled $l^{MA}$ to be equal to $\delta$ in \cref{lemma:delta-from-rdp}, optimizing $\lambda$ and the noise $p$ simultaneously. We fix the truncation range of the learned noise $p$ to avoid problems with non-differentiable discrete inclusion-indices of distinguishing events $\bucketSymbol_\infty$. To ensure numerical stability, we executed all computations in log-space\footnote{We used PyTorch's built-in function \textit{logsumexp} (\url{https://pytorch.org/docs/stable/generated/torch.logsumexp.html})}. Additionally, we set $\lambda = \lambda_{sq}^2 + 10^{-4}$ and trained $\lambda_{sq}$ to enforce positivity.

\subsubsection{PrivacyBuckets ADP-bound}\label{sec:approach:adp-bound}

PrivacyBuckets~\cite{meisermohammadi2018PB} is another method to provide upper bounds for privacy leakage.
It utilizes that the privacy-loss random variable of two independent invocations of a mechanism is equal to a \textit{convolution} of the privacy-loss random variables produced by single invocations. 
\new{Based on this observation, PrivacyBuckets distributes probability mass in a finite amount of equidistant and discrete buckets spanning a specific privacy-loss range and convolves them, in contrast to evaluating each combination of privacy-losses individually which grows exponentially with the number of compositions. The estimation accuracy depends on the number of buckets and their width described by the discretization factor $f > 1$ (while $f-1$ stays small), capturing the maximal difference between two privacy losses that obtain the same bucket. This method is used in practice in Google's Differential Privacy Library.\!\footnote{\url{https://github.com/google/differential-privacy/tree/main/python/dp_accounting}} Due to computational complexity, we did not implement error-correction~\cite{meisermohammadi2018PB}.
}

Formally, given $D,D'$ and utilizing the discrete output distribution $\{p^{D}_i\}_i$ originating from $M(D;p)$, we combine the probability mass of similar privacy-loss occurrences in discrete buckets referenced by their index $j\in \mathbb Z$
\begin{equation*}
    \bucketd[1]{j} = \sum_{f^{j\text{-}1}<e^\losstwo{M(D;p)}{M(D';p)}\leq f^j} p^{D}_i
\end{equation*}
for a suitable discretization-factor \new{$f$}. 
The subscript denotes the number of mechanism invocations, only one in this case. The events with privacy-loss $\Losssymbol(o) = \infty$ are collected separately in their own bucket $\bucketd[1]{\infty}$, \new{a corner case with special treatment}. 
For discretization purposes, the number of buckets are limited to $2h$ such that both ends span a series of indices from $-h$ to $h$. We denote these confined buckets by $\bucketdl{j}$, $j\in\{-h,\ldots,h\}$. We define the bucket $\bucketdl[1]{-h}=\sum_{e^\losstwo{M(D;p)}{M(D';p)}\leq f^{-h}} p^{D}_i$ to contain all masses with privacy loss smaller or equal than $\log(f^{-h})$ and incorporate the resulting events larger than $\log f^h$ in $\bucketdl[1]{\infty}=\bucketd{\infty} + \sum_{f^{h} < e^\losstwo{M(D;p)}{M(D';p)}} p^{D}_i$. \new{All other buckets with $i\! \in\! [\text{-}h\!+\!1, h]_\mathbb N$ stay same: $\bucketdl[1]{i} = \bucketd[1]{i}$.}
Meiser et al.~\cite{meisermohammadi2018PB} have shown that independent sequential composition can be expressed as
\begin{align*}
    &\bucketdl[2]{j} = \sum_{\hspace{-3em}\mathrlap{i\in\{|j|-h,h\}}}
    \bucketdl[1]{i} \cdot \bucketdl[1]{j\!-\!i}  \quad \text{for}\quad j \in \{-h+1,\ldots,h\}\\
    &\bucketdl[2]{-h} = \sum_{\hspace{-3em}\mathrlap{j\in\{-2h,-h\}}}\hspace{2.6em}\sum_{\hspace{-1em}\mathrlap{i\in\{-h,h+j\}}} \bucketdl[1]{i} \cdot \bucketdl[1]{j\!-\!i} \quad\text{and}\\
    &\bucketdl[2]{\infty} = 2\cdot\bucketdl[1]{\infty}\!-\!\left(\bucketdl[1]{\infty}\right)^2\!+\!  \sum_{\hspace{-0.5em}j=h+1}^{2h}\sum_{\hspace{-1em}\mathrlap{i=j-h}}^h \bucketdl[1]{i} \cdot \bucketdl[1]{j\!-\!i}
\end{align*}
This composition is equivalent to a convolution of the buckets~\cite{sommer2019privacy} and it can be repeated as long as numerically feasible, even with other $n$ than exponents of two~\cite{meisermohammadi2018PB}. We obtain the upper bound after $n$ compositions by evaluating
\begin{equation*}
    \delta_{M(D_0),M(D_1)}^{ADP} = \bucketdl[n]{\infty} + \!\sum_{\mathclap{j\geq \frac{\eps}{\log f}}} \left(1 - e^{j\cdot \log (f) - \eps}\right)\cdot \bucketdl[n]{j}
\end{equation*}
with the final $\delta=\max\!\left(\!\delta_{M(D_0),M(D_1)}^{ADP}, \delta_{M(D_1),M(D_0)}^{ADP}\!\right)$. 

\myparagraph{Implementation}
We implemented $l^{ADP}$ according to the $\delta$ introduced above, computing values in log-space where possible. For non-symmetric problems, we computed $l^{ADP}_{A,B}$ and $l^{ADP}_{B,A}$ separately and chose the noise-producing the smaller $\delta^{PB}$. The distinguishing events are included in $\bucketdl[n]{\infty}$. For gradient-descent algorithms, the gradients need to be derivable, which is problematic for \textit{equal} and \textit{smaller} relations or for the \textit{ceiling} function when computing indices. As a remedy, we replaced the derivative of such Boolean functions with a sharp function while relying on the built-in function for the forward-pass. Only for $\smaller$, we replaced the forward-pass function with a Sigmoid and used its derivative for the backward-pass. 
\begin{align*}
      \smaller(x,y) :=& \Sigmoid(-5\cdot(x-y+0.5))\\
      \frac{\partial}{\partial x} \equal(x,y) :=& \frac{1}{1+\frac{(x-y)^2}{0.01}}\\
      \frac{\partial}{\partial x} \ceil(x) :=& 1
\end{align*}



\subsubsection{PrivacyBuckets PDP-bound}\label{sec:approach:pdp-bound}

For this bound, we used the same composition technique as for the previously introduced PrivacyBuckets ADP-bound. In contrast to $\delta^{ADP}$\!\!, events with a privacy-loss exceeding $\eps$ are not weighted correspondingly but fully included%
~\cite{sommer2019privacy}. 
\begin{equation*}
    \delta_{M(D_0),M(D_1)}^{PDP} = \bucketdl[n]{\infty} + \sum_{\mathclap{j\geq \frac{\eps}{\log f}}} \bucketdl[n]{j}
\end{equation*}
The numerically differentiable bound $l^{PDP}$ is obtained analogically to $l^{ADP}$\!\! with the same constraints applying.


\section{Evaluation}\label{sec:results}

\begin{table}[t]
    \setlength{\extrarowheight}{0.5mm}
    \ifpopets
\else
        \footnotesize
    \fi
    \centering
    \caption{Nomenclature}
    \label{tab:nomenclature}
    \begin{tabular}{p{16mm}|l}
       $p$ & \textmd{generated noise distribution}\\
       $n$ &\textmd{number of mechanism invocations (compositions)}\\
       $\utilityloss_{L_1}$, $\utilityloss_{L_2}$  & \textmd{utility-loss with $L_1$ or $L_2$ norm}\\
       $l^{MA}$  &  \textmd{optimizer loss for Extended MA}\\
       $l^{ADP}$  &  \textmd{optimizer loss for PrivacyBuckets ADP}\\
       $l^{PDP}$  &  \textmd{optimizer loss for PrivacyBuckets PDP}\\
       $\delta,\eps$  & \textmd{ADP or PDP guarantees}\\
       $\delta^{PB}$ & \textmd{numerical exact upper-bound\textsuperscript {\ref{footnote:privacybuckets}}}\\
       $\delta^{PB}_n$ & \textmd{$\delta^{PB}$ evaluated for n compositions}\\
       $\delta^{PB}_{A/B},\, \delta^{PB}_{B/A}$ & \textmd{one-sided evaluation of $\delta^{PB}$} \\ 
       $l^X_{A,B},\, l^X_{B,A}$ &\textmd{one-sided evaluation, $X$ in \{MA, ADP, PDP\}}\\
    \end{tabular}
\end{table}

In this section, we present the results of our implementation. \Cref{sec:results:implementation} introduces the methodology and implementation. Then, we reproduce the analytical results from Geng et al.~\cite{geng2017staircase} in \cref{sec:results:reproduction}.  \Cref{sec:results:utility_delta} points out the trade-offs between utility and $\delta^{PB}$ and partly outperforms truncated Gaussian noise. In addition, we generated noise for DP-SGD in \cref{sec:results:dp_sgd} and illustrate truncation effects for finite additive noise events in \cref{sec:evaluation:truncation}. Finally, \cref{sec:results:numerical_stability} presents the numerical stability of our approach. \new{We have made  the source code for our approach available on GitHub\footnote{\url{https://github.com/teuron/optimal_truncated_noise}}.}

\subsection{Evaluation Details}\label{sec:results:implementation}
While we have minimized the noise distributions with different losses $l^X$, any shown value for $(\eps,\delta)$-ADP (or PDP) guarantees are computed using the GitHub PrivacyBuckets implementation\textsuperscript{\ref{footnote:privacybuckets}} with the generated worst-case noise distribution(s) as input. This implementation supports error correction for its upper-bounds $\delta^{PB}$, which outperforms our complexity-reduced losses $l^{ADP}$ and $l^{PDP}$. \new{Please note that such provided $\delta^{PB}$ for the noises generated are accurate and numerical errors are negligible.} For comparative reasons, we produce for $l^{PDP}$ an ADP-$\delta^{PB}$ as well, despite the PDP-formulation used during optimization. We used $250.000$ buckets and an adaptively chosen factor $f$. For worst-case output distributions, we evaluated both combinations $\delta^{PB}_{A/B}$ and $\delta^{PB}_{B/A}$ for the indicated $\eps$ and showed their maximum. 

Contrarily to the broad understanding of ($\eps,\delta$)-differential privacy guarantees as a property of mechanisms, we aim to find the \textit{minimal} $\delta$ for a given $\eps$. Accordingly, we say that two worst-case noise distributions \textit{produce} a (minimal) $\delta$. If not otherwise indicated, we restricted ourselves to optimizations for $\eps=0.3$.

\myparagraph{Implementation Details}
We use \textit{PyTorch} (v1.7.1)~\cite{pytorch} as optimization framework performing gradient descent in double-precision with Adam~\cite{kingma2017adam} while applying an exponential learning-rate and utility-weight decay. 
We use the following hyper-parameters if not otherwise indicated:
The noise model (\cref{sec:approach:noise-model}) contains $K=10000$ Sigmoid functions and the corresponding parameters are initialized uniformly at random.
The model is evaluated at $60000$ equidistant points on the x-axis within a symmetric range $[-r+a, r+a]$ for $500$.
For numerical stability, we bias the range by $a=10^{-5}$ as numerical instabilities can occur at $x$=$0$.
Also, using more discretization steps and more Sigmoid functions did not result in considerably better noise due to numerical instability and internal rounding errors.
We minimize \cref{eq:approach:minimizing-equation} with a learning rate $\learningRate=0.001$ (exponential decay factor $0.99995$) and a utility-weight $\utilityweight=0.5$ (halving period $\utilityweightdecayrate=2.500$ epochs, $\utilityweight_{min}=10^{-7}$). The privacy bound $l^{MA}$ is evaluated in $100.000$ epochs and $l^{ADP}$ and $l^{PDP}$ in a coarseness factor $f=1.000001$ and $15.000$ epochs using $2h$=$1000$ buckets.  

\myparagraph{Hardware Details} We run our experiments on a cluster which consists of 8 Nvidia A100 $40GB$ GPUs, an AMD EPYC 7742 CPU, and 1TB RAM. A single run of $l^{ADP}$ or $l^{PDP}$ takes $1$h $21$min, $l^{MA}$ is faster with $1.02$h. Extending the number of buckets increases the computational complexity exponentially; a single run with $10000$ buckets requires $1.41$ days.
Re-running all experiments once requires approximately $25$ days. 

\subsection{Reproduction of Analytic Optimal Noise}\label{sec:results:reproduction}

    \begin{figure}[t]
	    \centering
	    \begin{subfigure}[t]{0.48\textwidth}
	        \centering
	        \includegraphics[width=1.0\textwidth]{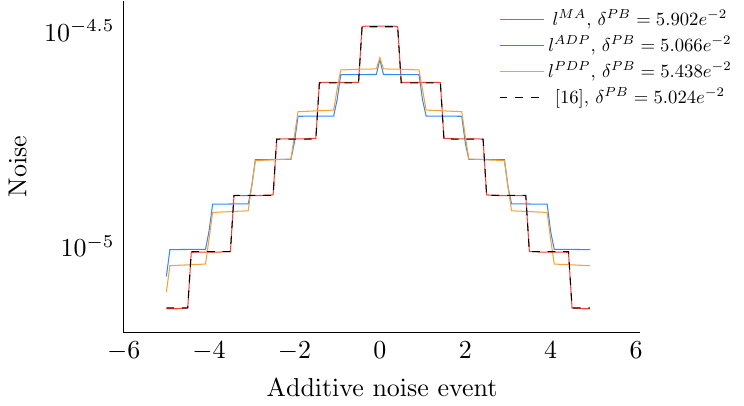}
	        \caption{with monotonicity enforcement.}
	        \label{fig:gengrepoduction:correct}
	    \end{subfigure}
	    \begin{subfigure}[t]{0.48\textwidth}
    	    \includegraphics[width=1.0\textwidth]{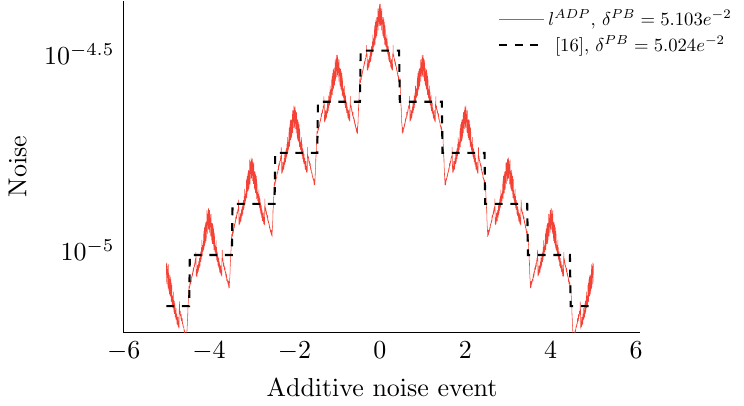}
    	    \caption{without monotonicity enforcement.}
	    	\label{fig:gengrepoduction:churchtowers}
	    \end{subfigure}
	    \caption{
	    Log-noise of numerically approximated staircase mechanism compared to optimal noise (Geng et al.~\cite{geng2017staircase}). (a) $l^{ADP}$\!\! covers the analytic noise with a relative difference in $\delta^{PB}$ of $10^{-2}$. The half-step shift for $l_{MA}$ coincides with $l^{PDP}$\!\! as $l^{MA}$ is conceptionally PDP. (b) without monotonicity enforcement, our method produces "pointy towers" guaranteeing ADP with lower utility-loss only for output differences exactly $0$ or $1$. $\eps=0.3$,  $\learningRate=0.01$.}
	    \label{fig:gengrepoduction}
	\end{figure} 
Geng et al.~\cite{geng2017staircase} proposed to replace the Laplace mechanism (additive Laplace noise to sensitivity bounded mechanisms) with their staircase mechanism. They showed that the staircase mechanism provides the same $\eps$-ADP guarantees as the Laplace mechanisms but minimizes the $L_1$-utility-loss \expectation{p}{|x|} for sensitivity $1$.  

First, we show in \cref{fig:gengrepoduction:correct} that our numerical approach $l^{ADP}$ reproduces the staircase-mechanism for $L_1$ utility. Contrarily to our numerical approach, the analytic staircase noise is stretched infinitely. For comparison, we truncated and re-normalized the proposed analytic noise, resulting in distinguishing events at the outermost stair-step. Consequently, the truncated analytic mechanism is not pure $\eps$-ADP anymore but contains a small $\delta$. The relative difference for $\delta^{PB}$ between the generated and analytic noise is smaller than $10^{-2}$. For completeness, we show the optimization results for our other mechanisms $l^{MA}$ and $l^{PDP}$ as well. As $l^{MA}$ internally uses a PDP-bound (discussed in \cref{sec:approach:RDP}), its generated noise is similar to $l^{PDP}$. With this figure, we want to illustrate that our implementations are sound: $l^{ADP}$ reproduces the analytic noise, and $l^{MA}$ and $l^{PDP}$ produce similar results despite being conceptionally different. We show full $\eps$-$\delta$-graphs and closeness to the analytic noise in \cref{fig:truncation_effects-epsilon-delta}. 


Without the monotonicity enforcement (see \cref{sec:approach:monotonicity}) \new{on a CNN model (see \cref{appendix:cnn_model}), the gradient-descent based tool} does not consider un-noised mechanism outputs $q(D)$ other than $0$ or $1$; the model optimizes for privacy only at the edges of the steps. There, the optimizer shifts probability mass towards the center of the steps, creating thereby shapes resembling \textit{pointy towers}, illustrated in \cref{fig:gengrepoduction:churchtowers}.  Note that the pointy-tower noise still guarantees $(\eps,\delta)$-ADP with lower utility cost ($l^{ADP}$: $\utilityloss_{L_1}=1.879$, Geng et al. $\utilityloss_{L_1}=1.884$), but only for cases where the differences $|q(D) - q(D')|$ of neighboring $D, D'$ results exactly in $0$ or $1$ as, e.g., for counting queries. 

\begin{figure}[t]
    \centering
    \begin{subfigure}[t]{0.48\textwidth}
        \includegraphics[width=1.0\textwidth]{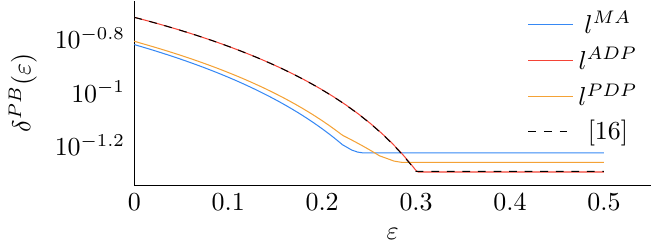}
    \end{subfigure}
    \begin{subfigure}[t]{0.48\textwidth}
	    \includegraphics[width=1.0\textwidth]{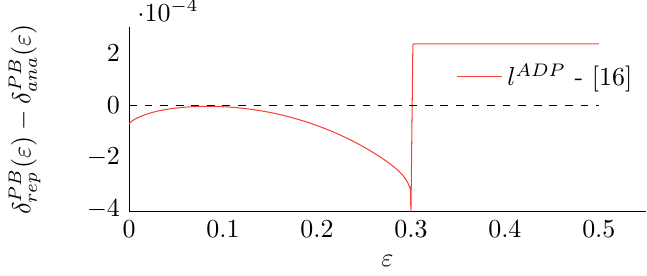}
    \end{subfigure}
    \hfill\phantom{*}\vspace{-1.5em}
    \caption{The $\delta^{PB}\!(\eps)$ graphs for truncated staircase mechanism~\cite{geng2017staircase} (\textit{ana}) and the reproduced noise distribution (\textit{rep}) compared, both generated with $\eps$=$0.3$. Note that this $\eps$ and the privacy-loss $\eps$ in the graphs differ. The  difference between $l^{ADP}$\!\! and \cite{geng2017staircase} is negligible, mostly dominated by distinguished events occurring after $\eps=0.3$}
\label{fig:truncation_effects-epsilon-delta}
\end{figure} 








\subsection{Utility vs. \texorpdfstring{$\delta(\eps)$}{delta(epsilon)} under Composition}\label{sec:results:utility_delta}

    \begin{figure*}[hbt!]
	    \centering
	    \begin{subfigure}[t]{0.6\textwidth}	
	    \phantom{******}
	    \includegraphics[width=0.8\textwidth]{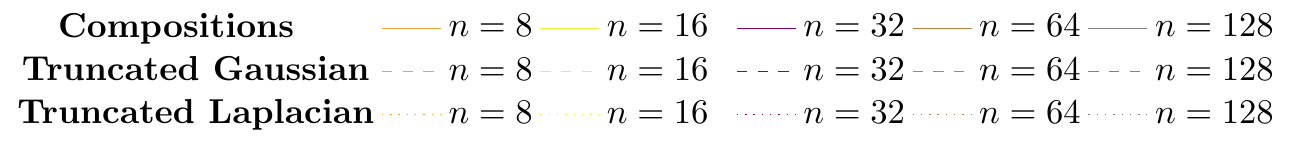}
	    \end{subfigure}\hfill\phantom{*}
	    
	    \begin{subfigure}[t]{0.28\textwidth}
	        \includegraphics[width=1.0\textwidth]{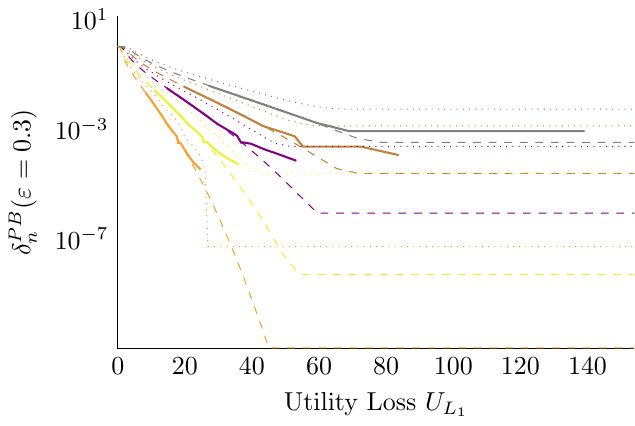}
	        \caption{PrivacyBuckets ADP $l^{ADP}$ with $L_1$ utility-loss}
            \label{fig:optimality_graph:pb_adp_l1_uwd_false}
	    \end{subfigure}
	    \begin{subfigure}[t]{0.28\textwidth}
    	    \includegraphics[width=1.0\textwidth]{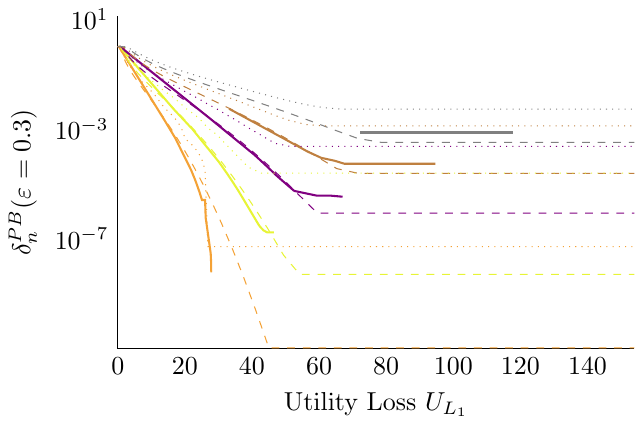}
	        \caption{Extended MA $l^{MA}$ with $L_1$ utility-loss}
            \label{fig:optimality_graph:pb_renyi_markov_l1_uwd_false}
	    \end{subfigure}
	    \begin{subfigure}[t]{0.28\textwidth}
    	    \includegraphics[width=1.0\textwidth]{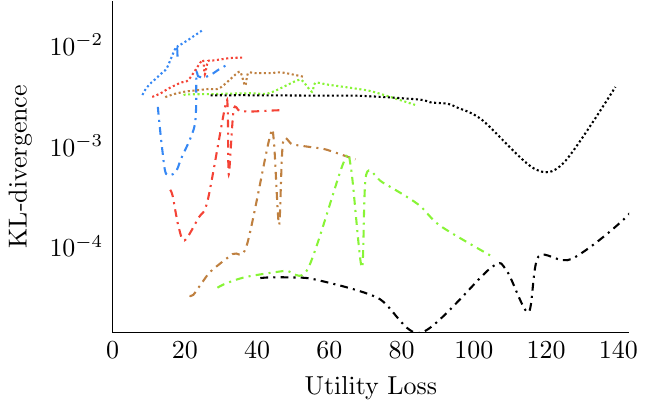}
	        \caption{For $L_2$-loss, $D_{KL}(\text{\setstackgap{S}{1pt}\tiny\stackanchor{truncated}{Gaussian}}||\text{generated})$ decreases with $n$ but stagnates for $L_1$-loss.
	        }
            \label{fig:optimality_graph:difference}
	    \end{subfigure}
	    \begin{subfigure}[t]{0.12\textwidth}\vspace{-9.3em}
        \includegraphics[width=0.92\textwidth]{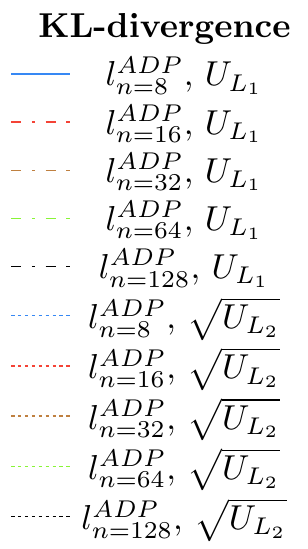}
	    \end{subfigure}
	    \hfill\phantom{*}
	    \begin{subfigure}[t]{0.28\textwidth}
	        \includegraphics[width=1.0\textwidth]{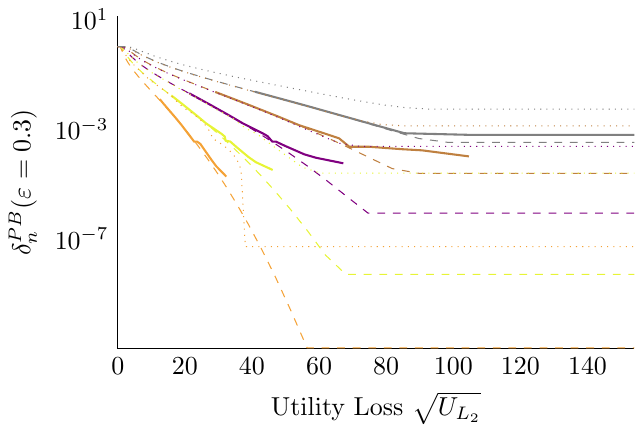}
	        \caption{PrivacyBuckets ADP $l^{ADP}$ with $L_2$ utility-loss}
            \label{fig:optimality_graph:pb_adp_l2_uwd_false}
	    \end{subfigure}
	    \begin{subfigure}[t]{0.28\textwidth}
    	    \includegraphics[width=1.0\textwidth]{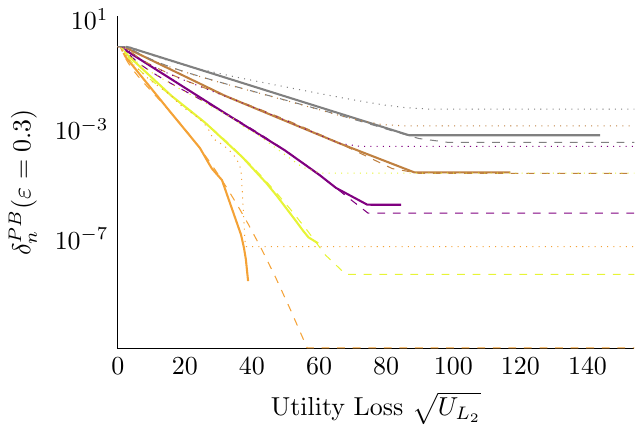}
	        \caption{Extended MA $l^{MA}$ with $L_2$ utility-loss}
            \label{fig:optimality_graph:pb_renyi_markov_l2_uwd_false}
	    \end{subfigure}
	    \begin{subfigure}[t]{0.28\textwidth}
	        \includegraphics[width=1.0\textwidth]{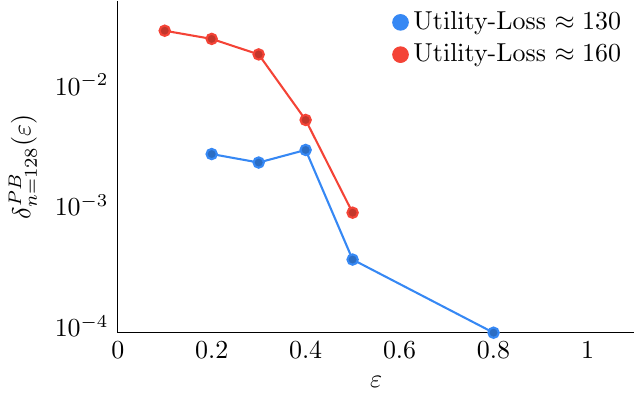}
            \caption{$\eps$-sweep for $l^{ADP}_{n=128} U_{L_2}$}
            \label{fig:different_eps:eps_delta}
	    \end{subfigure}
	    \hfill\phantom{*}
        \caption{Utility-loss vs. privacy leakage for sensitivity-bounded queries (sensitivity $s\!=\!1$) and multiple number of compositions $n$. PrivacyBuckets-ADP $l^{ADP}$\! and Extended Moments Accountant (MA) $l^{MA}$\! are compared to the truncated Gaussian mechanism. 
        \textbf{(c)} Despite the clear tendency for $L_2$-loss, the KL-divergence shows erratic behavior for a fixed $n$ due to numerical challenges. $\eps=0.3$}
	    \label{fig:optimality_graph}
    \end{figure*} 


Utility, the usefulness of the perturbed output $M(q;p)$, degrades with a larger width of noise. Adding wider noise increases the expected deviation from the real value, while a lower width comes with lower privacy guarantees. We examine the trade-off between utility and the corresponding $\delta^{PB}$ for a fixed $\eps=0.3$. Varying utility-outcomes are achieved by weighting the penalty for low utility via $\utilityweight$ in our loss-function differently (see \cref{sec:approach:lossfunction}). 
The effective utility-weights $\utilityweight$ span a range from $10^{-7}$ to $2$.

\Cref{fig:optimality_graph} illustrates our results for real-valued sensitivity-bounded query functions for $L_1$ and $L_2$ utility losses. Next to a single invocation of the mechanism $M$, we extend our results to independent sequential composition, i.e., the case where an adversary receives multiple independent outputs from a mechanism, all based on the same input $D$ and the same noise function $p$. Moreover, we trained the model with and without \textit{utility weight decay}. The upper-bounds $\delta^{PB}_n$ are obtained by the GitHub implementation\textsuperscript{\ref{footnote:privacybuckets}} where the privacy leakage of a single invocation was composed with the required number of compositions $n$. \new{A selection of generated noises are illustrated in appendix \cref{fig:optimalit:showcase}.} Noise that needs to achieve a certain $\delta$ after multiple compositions is wider than noise aiming for the same $\delta$ with fewer compositions, resulting in lower utility for higher $n$.
We compare our results to the privacy leakage of truncated Gaussian noise generated by clipping a zero-centered Gaussian distribution to the same range in which we optimize our numerical optimal noise. See \cref{def:truncated-gaussian} for a formal definition. 

There exists a minimal $\delta^{PB}$ for truncated Gaussian noise. With a small standard deviation $\sigma$, the resulting $\delta^{PB}$ is large but reduces with increasing $\sigma$. However, with growing $\sigma$, probability mass is shifted to the outermost regions of the noise just before the truncation-barrier and increasingly contributes distinguishing events which not only dominate $\delta^{PB}$ for large $\sigma$ but also increase it again. This effect is shown in appendix \cref{fig:std-dev-struncated-gaussian} and applies in a broader sense to our generated noise as well. As the standard deviation corresponds to the utility-loss monotonically, the optimality-graphs shown in \cref{fig:optimality_graph} would increase again for larger $\sigma$. There is no meaning in having noise with larger $\delta^{PB}$ \textit{and} larger utility-loss. Therefore, we replace these increasing values with the more optimal minimal, producing the horizontal lines. We show our raw results in appendix \cref{fig:optimality_graph:old}. There, we also included $l^{PDP}$ which we omitted in the main body because we have not shown \redemph{Shift Invariance} (\cref{thm:shift-invariance}) for this case. 


Our numerical setup produced inappropriate results for certain hyper-parameter configurations, primarily due to collapsing noise distributions with low utility-loss and $\delta \approx 1$ originating from a bestriding penalty for the utility-loss. 
See \cref{fig:utility_domination} for an example of such noise. For illustration clarity, we removed these samples from our plots. 

In conclusion, we are remarkably close to the results for truncated Gaussians. While we can outperform the Gaussian in certain cases for $l^{MA}$, we need to admit that these differences are minor (<4\%), except for $l^{MA}$ for $n=8$ and $n=16$ compositions where they are larger. With $l^{ADP}$, our optimality graphs were always above the truncated Gaussians. With $l^{ADP}$ for $L_2$-utility-loss, however, we reproduced the shape of a Gaussian with a KL-divergence of $<10^{-4}$ for $128$ compositions, illustrated in \cref{fig:optimality_graph:difference}. Our results suggest that this difference shrinks with a larger number of compositions, posing (truncated) Gaussians as a near-optimal shape considering $L_2$-loss.
\new{The performance of other privacy budgets parameters than $\eps=0.3$ are illustrated in \cref{fig:different_eps:eps_delta}.}

\subsection{DP-SGD}\label{sec:results:dp_sgd}

    \begin{figure*}[htb!]
	    \centering
	    \begin{subfigure}[t]{0.42\textwidth}
    	    \includegraphics[width=1.0\textwidth]{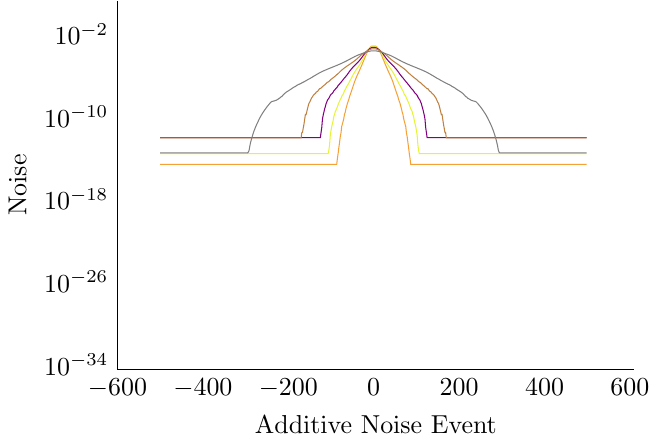}
	        \caption{$l_{B,A}^{ADP}$ for multiple number of compositions}
            \label{fig:dp_sgd:pb_adp_b_a}
	    \end{subfigure}
	    \begin{subfigure}[t]{0.42\textwidth}
    	    \includegraphics[width=1.0\textwidth]{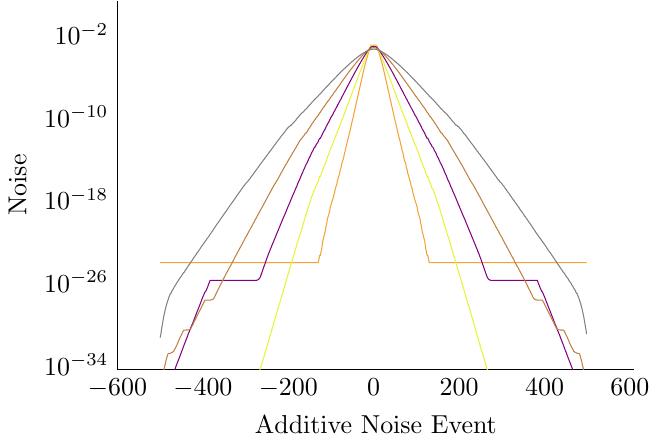}
	        \caption{$l_{B,A}^{MA}$ for multiple number of compositions}
            \label{fig:dp_sgd:ma_b_a}
	    \end{subfigure}
	    \begin{subfigure}[t]{0.12\textwidth}\vspace{-14em}
    	    \includegraphics[width=0.9\textwidth]{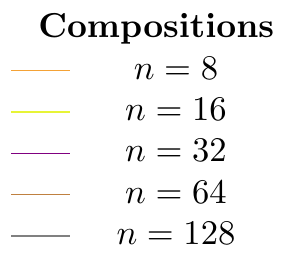}
	    \end{subfigure}
	    \hfill\phantom{*}
	    \begin{subfigure}[t]{0.42\textwidth}
    	    \includegraphics[width=1.0\textwidth]{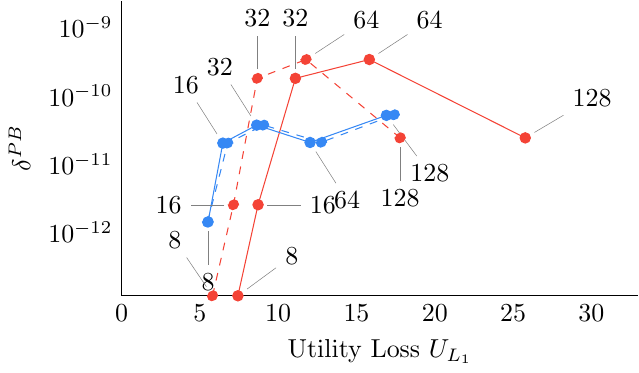}
	        \caption{Generated noise distributions compared to a truncated Gaussian with same $\delta^{PB}$\!\!. Numbers indicate compositions.}
            \label{fig:dp_sgd:optimality}
	    \end{subfigure}
	    \begin{subfigure}[t]{0.42\textwidth}
    	    \includegraphics[width=1.0\textwidth]{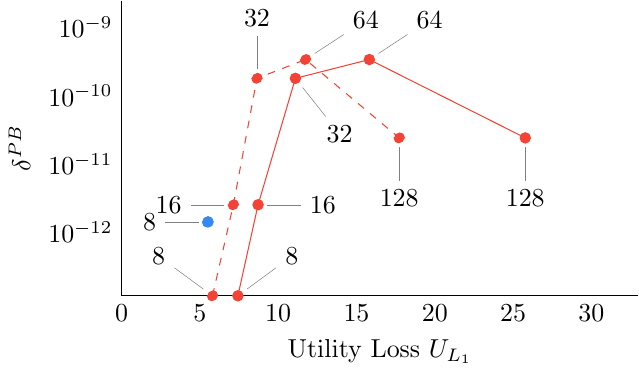}
	        \caption{Same as (c) but generated noise distributions and the corresponding truncated Gaussian are equally range-reduced.}
            \label{fig:dp_sgd:optimality_cut}
	    \end{subfigure}
	    \begin{subfigure}[t]{0.12\textwidth}\vspace{-12em}
    	    \includegraphics[width=0.9\textwidth]{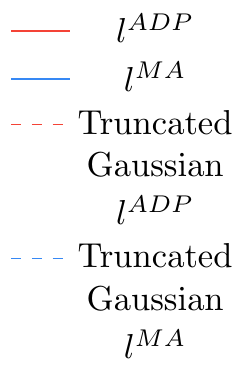}
	    \end{subfigure}
	    \hfill\phantom{*}
	    \caption{Illustration of DP-SGD (sub-sampling probability $q=0.1$) for multiple number of compositions $n$ and clipping distance $C=1$. PrivacyBuckets-ADP and Extended Moments Accountant (MA) are compared to the truncated Gaussian mechanism with same $\delta^{PB}$\!. $\eps=0.3$}
        \label{fig:dp_sgd}
    \end{figure*}

Differentially private stochastic gradient descent (DP-SGD) is a technique to train parameters with an SGD algorithm differentially private, protecting individual training samples. The highly influential work by Abadi et al.~\cite{abadi2016deeplearning} showed that the privacy leakage is sufficiently bounded by comparing the following two noise distributions:
\begin{align*}
    A &:= N(0, \sigma)\\
    B &:=(1-q) \cdot N(0, \sigma) + q\cdot N(C, \sigma)
\end{align*}
for a clipping constant $C$, a sampling probability $q$, and a noise-multiplier $\sigma$. 
As the added noise is monotonically falling from the mean, any gradient-update resulting in less than $C$ deviation has less privacy leakage than when maximally distanced by $C$ due to the larger overlap of the noising Gaussians.
We adapted our algorithm to optimize $p$ by constructing the following intermediate worst-case distributions 
\begin{equation}
\begin{aligned}
    A &:= p\\
    B &:= (1-q)\cdot p + q\cdot p_{\text{shifted-by-$C$}}
\end{aligned}\label{eq:DP-SGD-optimzer-worst-case-distribtutions}
\end{equation}
and then minimizing $l^X_{A,B}$ and $l^X_{B,A}$ (instead of $l^X_{p,p+s}$ for sensitivity $s$). While we have not proven the \redemph{Shift Invariance} \cref{thm:shift-invariance} for this case, we argue that this assumption is applicable as well because the same reasoning applies: With sensitivity $C$, all outputs with smaller deviation than $C$ induce less privacy leakage. 

For each of the two cases $X\in\{MA,ADP\}$, we run the minimizer twice (once for $l^X_{A,B}$ and once for $l^X_{B,A}$), each run producing two graphs $\delta^{PB}_{A/B}$ and $\delta^{PB}_{B/A}$, eight graphs in total. For a given $\eps$, the optimal noise $p$ originates from the run where $\delta = \max(\delta^{PB}_{A/B}, \delta^{PB}_{B/A})$ is smaller. 
We find that for our setup, it is sufficient to consider $l^X_{B,A}$ only as it dominates $l^X_{A,B}$ everywhere.  
For clarification, the four corresponding lines for each case are shown in \cref{fig:delta-eps-graph-DP-SGD}.

\new{For $L_1$ utility-loss, \Cref{fig:dp_sgd} illustrates our findings for a fixed $\eps=0.3$ and $C=1$ while the $L_2$ results are shown in appendix $\cref{fig:dp_sgd_appendix}$.} We chose a high $q=0.1$ to make potential effects visible. The first row shows actual noise distributions for several number of compositions. We conjecture that the appearing horizontal plateaus originate from the pursuit of the optimizer to achieve zero there, but the bias $A$ in our model prevents that. The combinations $l^X_{A,B}$ and noise by $l^{PDP}$ are shown in appendix \cref{fig:dp_sgd_appendix}. For numerical reasons revisited later in the discussion, we did not allow our model to learn the truncation range adaptively. 

In \cref{fig:dp_sgd:optimality}, we compare our generated noise to truncated Gaussians spanning the full range $[-500,500]$ and with adapted standard deviation $\sigma$ to produce the same $\delta^{PB}_n$. A single plot-line is produced by the different number of compositions. Unfortunately, all generated noise distributions with obtained horizontal plateaus have a higher utility-loss, most likely due to these plateaus. In contrast, noises without plateaus ($l^{MA}$ for $n$=$16,32,64,128$) show better utility. 
For remedy, \cref{fig:dp_sgd:optimality_cut} illustrates the same procedure, but the range of the already trained generated noise and the Gaussian is reduced to the point where the horizontal plateau begins, minus $1$ range-unit to minimize the distinguishing events when the curve is shifted by sensitivity $1$. For $l^{MA}$, this is only meaningful for $8$ compositions. It did not bring the desired effect, most likely because the reduction required a noise-renormalization the optimizer did not anticipate. \Cref{fig:dp_sgd:comparison:pb_adp,fig:dp_sgd:comparision:ma} in the appendix illustrate truncated Gaussians for both cases, full-range and reduced. They almost overlap with a maximal relative difference in their standard deviation of $0.33\%$. In contrast to the noise discussed in the previous section, however, DP-SGD generated noise does not resemble a Gaussian shape. 


\myparagraph{Sampling from our numerical distribution}
Rotation-symmetric noise can be reduced to a 1-dimensional privacy-analysis~\cite{abadi2016deeplearning} for which we provide an optimization method. To sample from the generated noise, we can extend the 1-dimensional numerical optimal noise $p$ to a $d$-dimensional \textit{rotation-symmetric} distribution, expressed by a point on the unit-sphere, which is then projected (multiplied) by the radius $r$ from the origin. The point on the unit-sphere can be sampled by drawing $d$ normally distributed  ($\mu=0,\sigma=1$) random numbers $Y_i$ and normalize them by their norm $\sqrt{\sum_i Y_i^2}$, see~\cite{donald1969art}.  
The radius $r$ is sampled from the generated noise $p$ directly according to \cref{def:noiseFunc}. 
We leave a proof and a detailed evaluation of the proposed sampling method for future work.





\subsection{Truncation Effects}\label{sec:evaluation:truncation}

    \begin{figure*}[t]
	    \centering
	    \begin{subfigure}[t]{0.32\textwidth}
	        \includegraphics[width=1.0\textwidth]{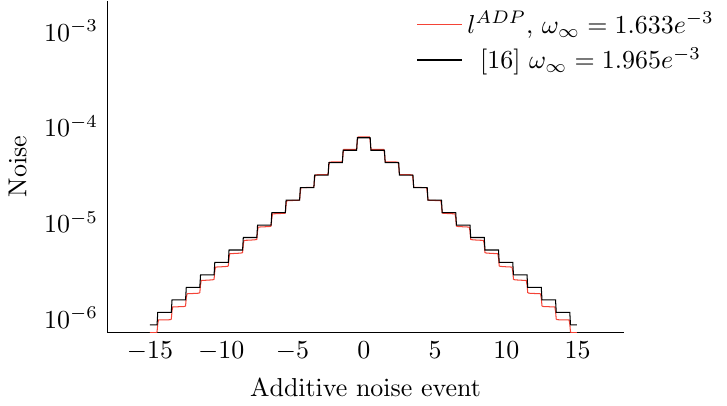}
            \label{fig:truncation_effects:range15}
	    \end{subfigure}
	    \begin{subfigure}[t]{0.32\textwidth}
    	    \includegraphics[width=1.0\textwidth]{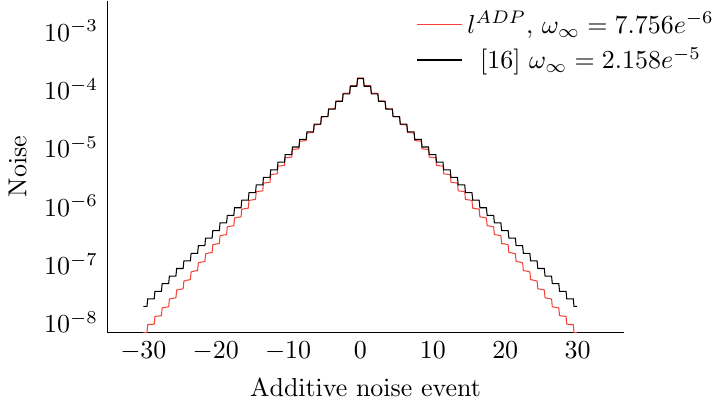}
            \label{fig:truncation_effects:range30}
	    \end{subfigure}
	    \begin{subfigure}[t]{0.32\textwidth}
    	    \includegraphics[width=1.0\textwidth]{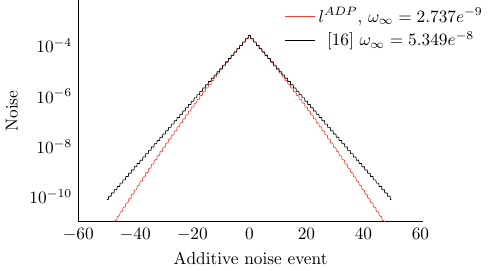}
            \label{fig:truncation_effects:range50}
	    \end{subfigure}\vspace{-1.5em}
	    \caption{Truncation effects illustrated on $l^{ADP}$ for different ranges by comparing to truncated and re-normalized analytic staircase noise (Geng et al.~\cite{geng2017staircase}) with sensitivity $s=1$ for a single mechanism invocation $n=1$. The model attempts to minimize $l^{ADP}$ and consequently shifts mass away from the distinguishing events $\bucketSymbol_\infty$. Learning rate $\learningRate=0.01$, $\eps=0.3$}
        \label{fig:truncation_effects}
	\end{figure*}

Analytical noise distributions in differential privacy usually spread to infinity, e.g., the staircase mechanism by Geng et al.~\cite{geng2017staircase}. Our numerical approximation considers only bounded widths of additive noise, resulting in distinguishing events at the outermost areas of the noise pattern as such events can only originate from one of the two possible inputs $D$ or $D'$. Our numerical approach seeks to minimize such distinguishing events, as their probability mass directly contributes to $l^X$. Consequently, the center of the noise distribution is elevated due to the normalization requirement that all 
probabilities sum up to $1$.

This particular effect is visible in the noise distributions shown in \cref{fig:truncation_effects} where we compare generated noise patterns to the truncated and re-normalized staircase mechanism: the wider the truncated noise, the more visible this effect.



%
%
%

\subsection{Numerical Reproducibility}\label{sec:results:numerical_stability}
    Our approach provides reproducible and stable results for different random parameter initialization.
    For the analytic noise reproduction experiment with range $[-5, 5]$ and $n=1$, we have run each privacy accountant method $10$ consecutive times with different initialization seeds. The gained insights apply to the other setups as well.
    \Cref{fig:training} illustrates the convergence of the resulting $\delta^{PB}$ and its variance, obtained from $10$ runs. The privacy-bound methods $l^{ADP}$ and $l^{PDP}$ converge between $5\,000$ and $8\,000$ and $l^{MA}$ roughly after $80\,000$. Consequently, we have set the training duration to $15\,000$, and $100\,000$ respectively.
    
    As demonstrated in \cref{fig:training:std_all}, all three methods show a significantly smaller standard deviation of the corresponding $l^X$ compared to the absolute value $|l^X|$ (<1\%), indicating stable convergence. The alternating pattern of $l^{ADP}$ and $l^{PDP}$ is an artifact of the logarithmic y-axis. The method $l^{MA}$ alternates in a similar absolute range. 
    
    

    \begin{figure}[t]
	    \centering
	    \begin{subfigure}[t]{0.43\textwidth}
	        \includegraphics[width=1.0\textwidth]{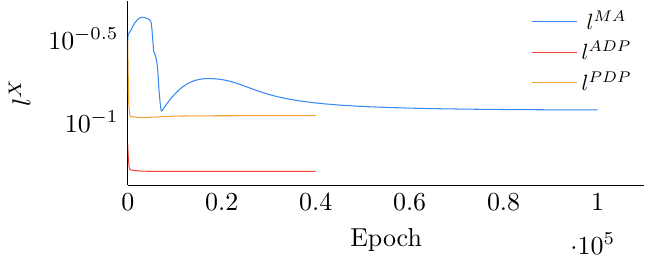}
	        \caption{$l^X$ mean  over epochs}
            \label{fig:training:delta_over_epochs}
	    \end{subfigure}
	    \hfill\phantom{*}
	    \begin{subfigure}[t]{0.43\textwidth}
    	    \includegraphics[width=1.0\textwidth]{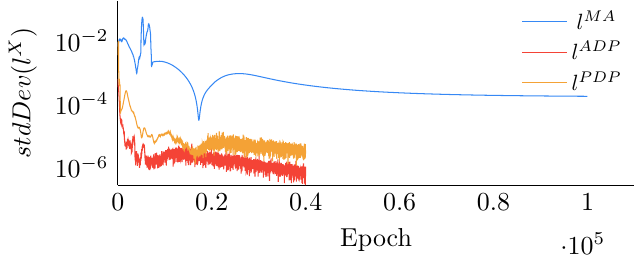}
	        \caption{$stdDev(l^X)$ over epochs}
            \label{fig:training:std_all}
	    \end{subfigure}
	    \hfill\phantom{*}
	    \label{fig:training}
        \caption{Numerical stability of our model over $10$ runs with differently seeded initialization. Note that $l^{ADP}$ and $l^{PDP}$ already converged within $10$k epochs and that $l^{MA}$\! approaches $l^{PDP}$\!, not $l^{ADP}$\!.}
    \end{figure}

\section{Discussion}

Our proposed algorithm outputs noise distributions for which we can guarantee an upper-bound for ($\eps,\delta$)-ADP. 
Simultaneously, we cannot provide a method to estimate how much our result deviates from the target optimum that minimizes our optimal noise requirements.
However, when reproducing previous analytical results (\cref{sec:results:reproduction}), we have shown in \cref{fig:truncation_effects-epsilon-delta} that the relative difference between the analytic and numerical optimal noise is smaller than $10^{-2}$.

\myparagraph{Closeness to truncated Gaussians}
Another argument for our soundness:
our algorithm produces noise with ($\eps,\delta$)-ADP guarantees similar to truncated Gaussians while also creating similar utility-loss as shown in \cref{fig:optimality_graph}. We are often below the ones of a truncated Gaussian, also for DP-SGD. For $L_2$-utility-loss and sensitivity-bounded queries, we generate almost perfect Gaussian-shaped numerical optimal noise.
These effects might indicate that (truncated) Gaussians are close to optimal in a simple setting, especially under independent sequential composition and $L_2$-utility-loss.
\redcomment{Is Gaussian optimal for l2-loss? we know that a Gaussian prior of a optimization parameter leads to l2-regularizer. Moreover, our utility-losses look suspiciously like regularizer. 

\url{https://stats.stackexchange.com/questions/163388/why-is-the-l2-regularization-equivalent-to-gaussian-prior}

David: Ok, wir wissen: Gaussian prior für beta fürt zum l2 loss. Die frage ist noch, wie sehr unser $delta_log_loss$ das beta verzerrt. Wir wissen auch, dass die Gaussian eine "invariante" ist unter Composition. Das heisst, bei Gaussian noise macht der l2 regularizer nichts kaputt. Die frage, die IMHO noch überigbleibt, ist, ob diese invariante ein maximu oder ein minimum der otpimirung ist. Unsere impirie suggestiert ein Minimum. Sonst würden wir nicht da hin konvergieren. 

die frage: was passiert mit p unter composiotn? Ich nehme an, das ist ein multinomiales polynom, welches aufgrund des zentralen grenzwertsatzes zu einer Gaussian konvergiert, mit einem linearen fehlerterm. FALSCH: l2 PB PDP PRODUZIERT AUCH GAUSSIANS: Und ich verumute, dass die bucketgewichtung mit dem  $-e^log_f*j - \eps$ diesen linearen fehler ausgleicht. das ergebnis ist die Gaussian. Und darum klappts nicht so gut bei PDP / Markov. 

AUf jeden fall sieht es so aus, als würde die Compositon irgendwas Gaussianartiges machen (da convolution $-->$ zentraler Grenzwertsatz) und daher den noise für l1 loss, der eigentlich eine Laplace Prior hat, "verbiegen". Und das optimum für l2 loss ist eine invariante der composition, das wissen wir von privacy loss classes paper. Wieso der unterschied zwischen AP und PDP zustande kommt, da kann ich atm nur spekulieren. 
}


\myparagraph{Why ADP} ADP is broadly accepted within the community because it is closed under \textit{post-processing}~\cite{dwork2014algorithmic}, i.e., other algorithms that take the output of an ADP mechanism cannot worsen the privacy guarantees, which is not valid for its more intuitive alternative formulation PDP. In addition, Moments Accountant~\cite{abadi2016deeplearning} analyzes its privacy leakage for ADP as well. 

\myparagraph{Why gradient descent}
We decided to search for optimal noise by applying gradient descent because it allows us to use complex (deep) noise-generating models for searching a full (discrete) output distribution, not only hyper-parameters. Next to its great tooling support and available documentation (PyTorch), it also allows us to try different models in an easy plug-and-play fashion. 

\myparagraph{Why Sigmoid-model} 
We tried several different models starting from a multi-layer-perceptron over convolutional neural networks to more constrained methods like our final Sigmoid-approach. We found that for producing intermediate horizontal lines like in \cref{fig:gengrepoduction:correct}, we require a strong dependence between neighboring x-axis points that neural networks were not able to deliver easily. In the end, we achieved the best results with the used Sigmoid-model, while higher values for the number of Sigmoids $K$ or x-axis discretization steps did not yield better results. 

\myparagraph{Numerical and Complexity Issues}
Due to the complexity and RAM requirement on our hardware, we were not able to produce more than $128$ compositions with $l^{ADP}$ and $l^{PDP}$. The simpler $l^{MA}$ can achieve up to $2^{10}$ compositions. We also tried to adapt the truncation range of the generated noise automatically during training. While easy to adapt $l^{ADP}$ and $l^{PDP}$, the loss $l^{MA}$ treats the exclusion of distinguishing events in a discrete non-differentiable manner. However, this induced for all methods unanimous additional instabilities leading to higher privacy leakage.
While analytic distributions (and sampling therefrom) might deliver better results in certain scenarios, our solution can adapt to more complex cases and provides a general method to find noise independent of the specific underlying problem.   

\new{
\myparagraph{Application scenarios} Truncated optimal noise is useful where a minimal and limited utility-loss is preferred. Examples are noised real-time systems with privacy requirements where a delay imposes utility-loss, e. g., privacy-aware operating system modules or anonymous communication networks (ACNs). Specifically, the ACN \textit{Stadium}~\cite{tyagi2017stadium} delays package forwarding by adding Poisson noise due to its closed analytical expression under composition.
Another interesting scenario is DP-SGD~\cite{abadi2016deeplearning} itself: In the unlikely case that a huge Gaussian noise is drawn that demolishes the model in training completely, the entire training will be repeated, restricting the width of the applied noise distribution more broadly. 


}

\myparagraph{Extension to non-binary worst-case distributions} In the first sections of our evaluation, we focus on the case \new{where the \textit{worst-case} query output distribution for neighboring inputs $|q(D) - q(D')|$, the \textit{input} to our optimization, are single events (either $0$ or $1$). For DP-SGD, we extended our approach to a mixture input distributions that are then convolved with the learned noise distribution.}
If the \new{worst-case} output-frequency distribution of the functional $|q(D) - q(D')|$ for neighboring inputs are known, a noise distribution with lower utility-loss than considering only the maximal deviation $s$ might be generated while still guaranteeing (average) ($\eps,\delta$)-ADP. Single events from certain input pairs might still have a higher privacy-loss than the average ADP-guarantee, but their occurrence probability is by design sufficiently small such that their weighted contribution to $\delta$ does not break the average ADP-guarantees.

We conjecture that for output-frequency distributions that are symmetric around $0$, our monotonicity and symmetric of noise result (\cref{thm:shift-invariance}) applies as well. If this is not the case, for example, with asymmetric worst-case frequency distributions of $q$, then our approach is still applicable, deploying a different non-monotonic and non-symmetric model. 
The same argumentation applied to privacy guarantees under limited background knowledge where the uncertainty of an attacker due to incomplete knowledge of the input datasets is modeled by a distribution of possible outputs of $q$~\cite{desfontaines2019differential}. 

\new{
\myparagraph{Beyond additive mechanisms} 
The exponential mechanism, a widely used and very flexible technique in differential privacy, has been shown to equal to adding Gumbal noise to a mechanism output and then picking the noisy maximum~\cite{dwork2014algorithmic}. Our method might be extended to search optimal noise for such post-processing as long as there are two resulting worst-case distributions, as long as the delta-bound and the utility-loss are derivable or can be sufficiently approximated by derivable functions and that the optimization converges. 
This might encompass a wide variety of DP mechanisms. 
}

\section{Related Work}
Differential privacy has many variants. The ones considered in this work are Rényi-differential Privacy/Moments-Accountant~\cite{abadi2016deeplearning,mironov2017renyidp}, approximate differential privacy~\cite{dwork2006our}, and probabilistic differential privacy~\cite{GoMaWaXiGe_09pdp}. Desfontaines et al. compiled a comprehensive list of other variants~\cite{desfontaines2020sok}.

\myparagraph{Composition.}
%
Vadhan et al.~\cite{dwork2010boosting} have shown that privacy will deteriorate as $\sqrt{n}\eps + n\eps^2$ under sequential composition, rather than the (trivial) worst-case $n\eps$ from previous work.
Meiser and Mohammadi~\cite{meisermohammadi2018PB} have recently introduced a numerical method for computing nearly tight ADP bounds (including lower bounds). Koskela et al.~\cite{koskela2021computing} extend PrivacyBuckets by convolving the privacy-loss random variables via FFT. While relevant to our approach, their error estimation is difficult to control with gradient descent at the time of writing.
Sommer et al.~\cite{sommer2019privacy} provided an analytical foundation for PrivacyBuckets and delivered many additional insights, such as the PDP bound used in this work. 
Kairouz et al.~\cite{kairouz2017composition} derived ADP bounds to prove upper ADP bounds for any mechanism, but their bounds are less tight than the Moments Accountant or PrivacyBuckets. 
Dong et al.~\cite{dong2020optimal} provide a nearly optimal composition theorem for the exponential mechanism. 

\myparagraph{Optimal noise and truncation.}
Ghosh et al.~\cite{ghosh2012universally} studied a very general utility-maximization framework for a single count query with a sensitivity one under $\eps$-differential privacy. 
Gupte et al.~\cite{gupte2010universally} derived the optimal noise probability distributions for a single count query with a sensitivity one for minimax (risk-averse) users.
Geng examined optimal noise-adding mechanisms with varying co-authors~\cite{geng2015optimal2,geng2017staircase,geng2015optimal,geng2019pmlr}, \new{focusing on optimal noise for $(\eps,0)$-ADP and $(0, \delta)$-ADP settings. In contrast to our work, they only considered one query (n=1).}
Soria-Comas et al.~\cite{soria2013optimal} also independently derived the staircase-shaped noise probability distribution under a different optimization framework.
Kumar et al.~\cite{kumar2019deriving} provided an optimal noise-adding mechanism for matrices in machine learning. 
Balle et al.~\cite{balle2018improving} optimized the Gaussian mechanism by bounding the privacy loss directly without using a tail-bound as used in Moments Accountant.
Bun et al.~\cite{bun2019composable} extended concentrated differential privacy to accept a sub-exponential instead of a sub-Gaussian tail, based on the Rényi-divergence. This kind of truncation, however, still does not incorporate distinguishing events. 
Geng et al.~\cite{geng2020tight} and Holohan et al.~\cite{holohan2018bounded} worked on the truncated (or bounded) Laplace mechanism. In contrast to our work, they searched only for optimal hyper-parameters for Laplace noise. 
Cesar et al.~\cite{cesar2021bounding} showed the superiority of Gaussian over Laplacian mechanism under composition when noising histograms. 
To our best knowledge, there is no other work aiming to search for general optimal noise numerically without optimizing only hyper-parameters of a predefined probability distribution. 


\section{Conclusion}
    We have introduced a novel tool for learning truncated additive DP mechanisms with strong (and sometimes near-optimal) utility-privacy trade-offs for sensitivity-bounded queries and DP-SGD. 
    We have proven that learning such mechanisms can be reduced to optimizing symmetric and from the mean decreasing noise distribution with gradient descent. Potentially of independent interest, we extended the Moments Accountant (MA) to incorporate distinguishing events. As that turned out to be non-trivial in general, we identified sufficient conditions for which MA can be feasibly computed, even if distinguishing events occur.
    For sensitivity-bounded queries, our learning method reproduces the proven optimal mechanism of Geng et al.~\cite{geng2017staircase} and (for a high number of compositions) generates -- using MA as privacy estimator -- a noise distribution that is comparable to and sometimes better than truncated Gaussian noise. For the worst-case output distribution of DP-SGD, we achieve similar privacy-utility trade-offs to the truncated Gaussian. 
    An interesting direction for future work is utilizing our gradient-descent approach to numerically find and validate noise distributions for complex, potentially randomized scenarios.
    \fv{We consider it interesting for future work to extend our model to non-binary, non-monotonous or asymmetric worst-case distributions.}
\bibliographystyle{plain}
\bibliography{main}

\begin{thebibliography}{10}

\bibitem{abadi2016deeplearning}
Martin Abadi, Andy Chu, Ian Goodfellow, H.~Brendan McMahan, Ilya Mironov, Kunal
  Talwar, and Li~Zhang.
\newblock Deep learning with differential privacy.
\newblock {\em Proceedings of the 2016 ACM SIGSAC Conference on Computer and
  Communications Security}, Oct 2016.

\bibitem{balle2018improving}
Borja Balle and Yu-Xiang Wang.
\newblock Improving the gaussian mechanism for differential privacy: Analytical
  calibration and optimal denoising.
\newblock In {\em International Conference on Machine Learning}, pages
  394--403. PMLR, 2018.

\bibitem{bun2019composable}
Mark Bun, Cynthia Dwork, Guy~N. Rothblum, and Thomas Steinke.
\newblock Composable and versatile privacy via truncated cdp.
\newblock In {\em Proceedings of the 50th Annual ACM SIGACT Symposium on Theory
  of Computing}, STOC 2018, page 74–86, New York, NY, USA, 2018. Association
  for Computing Machinery.

\bibitem{CDPBun}
Mark Bun and Thomas Steinke.
\newblock Concentrated differential privacy: Simplifications, extensions, and
  lower bounds.
\newblock In Martin Hirt and Adam Smith, editors, {\em Theory of Cryptography},
  pages 635--658, Berlin, Heidelberg, 2016. Springer Berlin Heidelberg.

\bibitem{cesar2021bounding}
Mark Cesar and Ryan Rogers.
\newblock Bounding, concentrating, and truncating: Unifying privacy loss
  composition for data analytics.
\newblock In {\em Algorithmic Learning Theory}, pages 421--457. PMLR, 2021.

\bibitem{desfontaines2019differential}
Damien Desfontaines, Esfandiar Mohammadi, Elisabeth Krahmer, and David Basin.
\newblock Differential privacy with partial knowledge.
\newblock {\em arXiv preprint arXiv:1905.00650}, 2019.

\bibitem{desfontaines2020sok}
Damien Desfontaines and Bal{\'a}zs Pej{\'o}.
\newblock Sok: Differential privacies.
\newblock {\em Proceedings on Privacy Enhancing Technologies}, 2:288--313,
  2020.

\bibitem{DiNi_03:confidence-gain}
Irit Dinur and Kobbi Nissim.
\newblock {Revealing Information While Preserving Privacy}.
\newblock In {\em {Proceedings of the Twenty-second ACM SIGMOD-SIGACT-SIGART
  Symposium on Principles of Database Systems (PODS)}}, pages 202--210. ACM,
  2003.

\bibitem{dong2020optimal}
Jinshuo Dong, David Durfee, and Ryan Rogers.
\newblock Optimal differential privacy composition for exponential mechanisms.
\newblock In {\em International Conference on Machine Learning}, pages
  2597--2606. PMLR, 2020.

\bibitem{dwork2006our}
Cynthia Dwork, Krishnaram Kenthapadi, Frank McSherry, Ilya Mironov, and Moni
  Naor.
\newblock Our data, ourselves: Privacy via distributed noise generation.
\newblock In {\em Annual International Conference on the Theory and
  Applications of Cryptographic Techniques}, pages 486--503. Springer, 2006.

\bibitem{dwork2006calibrating}
Cynthia Dwork, Frank McSherry, Kobbi Nissim, and Adam Smith.
\newblock Calibrating noise to sensitivity in private data analysis.
\newblock In Shai Halevi and Tal Rabin, editors, {\em Theory of Cryptography},
  pages 265--284, Berlin, Heidelberg, 2006. Springer Berlin Heidelberg.

\bibitem{dwork2014algorithmic}
Cynthia Dwork and Aaron Roth.
\newblock The algorithmic foundations of differential privacy.
\newblock {\em Foundations and Trends in Theoretical Computer Science},
  9(3-4):211--407, 2014.

\bibitem{DwRo_16:concentrated}
Cynthia Dwork and Guy~N. Rothblum.
\newblock {Concentrated Differential Privacy}.
\newblock {\em {CoRR}}, abs/1603.01887, 2016.

\bibitem{dwork2010boosting}
Cynthia Dwork, Guy~N. Rothblum, and Salil Vadhan.
\newblock Boosting and differential privacy.
\newblock In {\em 2010 IEEE 51st Annual Symposium on Foundations of Computer
  Science}, pages 51--60. IEEE, 2010.

\bibitem{geng2019pmlr}
Quan Geng, Wei Ding, Ruiqi Guo, and Sanjiv Kumar.
\newblock Optimal noise-adding mechanism in additive differential privacy.
\newblock In Kamalika Chaudhuri and Masashi Sugiyama, editors, {\em Proceedings
  of the Twenty-Second International Conference on Artificial Intelligence and
  Statistics}, volume~89 of {\em Proceedings of Machine Learning Research},
  pages 11--20. PMLR, 16--18 Apr 2019.

\bibitem{geng2020tight}
Quan Geng, Wei Ding, Ruiqi Guo, and Sanjiv Kumar.
\newblock Tight analysis of privacy and utility tradeoff in approximate
  differential privacy.
\newblock In {\em International Conference on Artificial Intelligence and
  Statistics}, pages 89--99. PMLR, 2020.

\bibitem{geng2017staircase}
Quan {Geng}, Peter {Kairouz}, Sewong {Oh}, and Pramod {Viswanath}.
\newblock The staircase mechanism in differential privacy.
\newblock {\em IEEE Journal of Selected Topics in Signal Processing},
  9(7):1176--1184, 2015.

\bibitem{geng2015optimal2}
Quan Geng and Pramod Viswanath.
\newblock The optimal noise-adding mechanism in differential privacy.
\newblock {\em IEEE Transactions on Information Theory}, 62(2):925--951, 2015.

\bibitem{geng2015optimal}
Quan Geng and Pramod Viswanath.
\newblock Optimal noise adding mechanisms for approximate differential privacy.
\newblock {\em IEEE Transactions on Information Theory}, 62(2):952--969, 2015.

\bibitem{ghosh2012universally}
Arpita Ghosh, Tim Roughgarden, and Mukund Sundararajan.
\newblock Universally utility-maximizing privacy mechanisms.
\newblock {\em SIAM Journal on Computing}, 41(6):1673--1693, 2012.

\bibitem{GoMaWaXiGe_09pdp}
Michaela G{\"{o}}tz, Ashwin Machanavajjhala, Guozhang Wang, Xiaokui Xiao, and
  Johannes Gehrke.
\newblock Privacy in search logs.
\newblock {\em CoRR}, abs/0904.0682, 2009.

\bibitem{gupte2010universally}
Mangesh Gupte and Mukund Sundararajan.
\newblock Universally optimal privacy mechanisms for minimax agents.
\newblock In {\em Proceedings of the twenty-ninth ACM SIGMOD-SIGACT-SIGART
  symposium on Principles of database systems}, pages 135--146, 2010.

\bibitem{holohan2018bounded}
Naoise Holohan, Spiros Antonatos, Stefano Braghin, and P{\'o}l Mac~Aonghusa.
\newblock The bounded laplace mechanism in differential privacy.
\newblock {\em arXiv preprint arXiv:1808.10410}, 2018.

\bibitem{kairouz2017composition}
Peter Kairouz, Sewoong Oh, and Pramod Viswanath.
\newblock {The composition theorem for differential privacy}.
\newblock {\em IEEE Transactions on Information Theory}, 63(6):4037--4049,
  2017.

\bibitem{kingma2017adam}
Diederik~P. Kingma and Jimmy Ba.
\newblock Adam: A method for stochastic optimization, 2017.

\bibitem{donald1969art}
Donald~E. Knuth.
\newblock The art of computer programming, vol. 2: Seminumerical algorithms,
  1969.

\bibitem{koskela2021computing}
Antti Koskela and Antti Honkela.
\newblock Computing differential privacy guarantees for heterogeneous
  compositions using fft, 2021.

\bibitem{kumar2019deriving}
Mohit Kumar, Michael Rossbory, Bernhard~A Moser, and Bernhard Freudenthaler.
\newblock Deriving an optimal noise adding mechanism for privacy-preserving
  machine learning.
\newblock In {\em International Conference on Database and Expert Systems
  Applications}, pages 108--118. Springer, 2019.

\bibitem{karaoke}
David Lazar, Yossi Gilad, and Nickolai Zeldovich.
\newblock Karaoke: Distributed private messaging immune to passive traffic
  analysis.
\newblock In {\em 13th {USENIX} Symposium on Operating Systems Design and
  Implementation ({OSDI} 18)}, pages 711--725, Carlsbad, CA, October 2018.
  {USENIX} Association.

\bibitem{meisermohammadi2018PB}
Sebastian Meiser and Esfandiar Mohammadi.
\newblock Tight on budget? tight bounds for r-fold approximate differential
  privacy.
\newblock In {\em Proceedings of the 2018 ACM SIGSAC Conference on Computer and
  Communications Security}, CCS '18, page 247–264, New York, NY, USA, 2018.
  Association for Computing Machinery.

\bibitem{mironov2017renyidp}
Ilya {Mironov}.
\newblock Rényi differential privacy.
\newblock In {\em 2017 IEEE 30th Computer Security Foundations Symposium
  (CSF)}, pages 263--275, 2017.

\bibitem{MuVa_16:psharp}
Jack Murtagh and Salil Vadhan.
\newblock The complexity of computing the optimal composition of differential
  privacy.
\newblock In {\em Proceedings, Part I, of the 13th International Conference on
  Theory of Cryptography (TCC)}, pages 157--175. Springer Berlin Heidelberg,
  2016.

\bibitem{narayanan2006break}
Arvind Narayanan and Vitaly Shmatikov.
\newblock How to break anonymity of the netflix prize dataset.
\newblock {\em arXiv preprint cs/0610105}, 2006.

\bibitem{pytorch}
Adam Paszke, Sam Gross, Francisco Massa, Adam Lerer, James Bradbury, Gregory
  Chanan, Trevor Killeen, Zeming Lin, Natalia Gimelshein, Luca Antiga, Alban
  Desmaison, Andreas Kopf, Edward Yang, Zachary DeVito, Martin Raison, Alykhan
  Tejani, Sasank Chilamkurthy, Benoit Steiner, Lu~Fang, Junjie Bai, and Soumith
  Chintala.
\newblock Pytorch: An imperative style, high-performance deep learning library.
\newblock In H.~Wallach, H.~Larochelle, A.~Beygelzimer, F.~d\textquotesingle
  Alch\'{e}-Buc, E.~Fox, and R.~Garnett, editors, {\em Advances in Neural
  Information Processing Systems 32}, pages 8024--8035. Curran Associates,
  Inc., 2019.

\bibitem{sommer2019privacy}
David~M. Sommer, Sebastian Meiser, and Esfandiar Mohammadi.
\newblock Privacy loss classes: The central limit theorem in differential
  privacy.
\newblock {\em Proceedings on privacy enhancing technologies},
  2019(2):245--269, 2019.

\bibitem{soria2013optimal}
Jordi Soria-Comas and Josep Domingo-Ferrer.
\newblock Optimal data-independent noise for differential privacy.
\newblock {\em Information Sciences}, 250:200--214, 2013.

\bibitem{stadium}
Nirvan Tyagi, Yossi Gilad, Derek Leung, Matei Zaharia, and Nickolai Zeldovich.
\newblock Stadium: A distributed metadata-private messaging system.
\newblock In {\em Proceedings of the 26th Symposium on Operating Systems
  Principles}, SOSP '17, page 423–440, New York, NY, USA, 2017. Association
  for Computing Machinery.

\bibitem{tyagi2017stadium}
Nirvan Tyagi, Yossi Gilad, Derek Leung, Matei Zaharia, and Nickolai Zeldovich.
\newblock Stadium: A distributed metadata-private messaging system.
\newblock In {\em Proceedings of the 26th Symposium on Operating Systems
  Principles}, SOSP '17, page 423–440, New York, NY, USA, 2017. Association
  for Computing Machinery.

\bibitem{vuvuzela}
Jelle van~den Hooff, David Lazar, Matei Zaharia, and Nickolai Zeldovich.
\newblock {\em Vuvuzela: Scalable Private Messaging Resistant to Traffic
  Analysis}, page 137–152.
\newblock Association for Computing Machinery, New York, NY, USA, 2015.

\end{thebibliography}
\appendix
\section{Utility vs. \texorpdfstring{$\delta^{PB}$}{delta} under Composition }  
    In addition to the Utility vs. $\delta^{PB}$ plots (\cref{fig:optimality_graph}), we present the results trained using the privacy accountant methods broken up in two sets, where one has \textit{utility weight decay} (UWD) activated while the other has not. When we applied \textit{utility weight decay} (UWD), we set the decay rate $\utilityweightdecayrate$ as such, that the resulting utility weight in the last epoch is approximately $\utilityweight_{start} \cdot 10^{-2}$, respectively $\utilityweightdecayrate=2500$ for PrivacyBuckets ADP/PDP and $\utilityweightdecayrate=16666$ for Moments Accountant. As seen in \cref{fig:optimality_graph:old:pb_adp_l1_uwd_true,fig:optimality_graph:old:pb_adp_l2_uwd_true,fig:optimality_graph:pb_pdp_l1_uwd_true,fig:optimality_graph:pb_pdp_l2_uwd_true,fig:optimality_graph:old:pb_renyi_markov_l1_uwd_true,fig:optimality_graph:old:pb_renyi_markov_l1_uwd_true}, were able to produce better noise with less $\delta^{PB}$ using this technique.

\begin{figure}[htb!]
    \centering
    \includegraphics[width=0.48\textwidth]{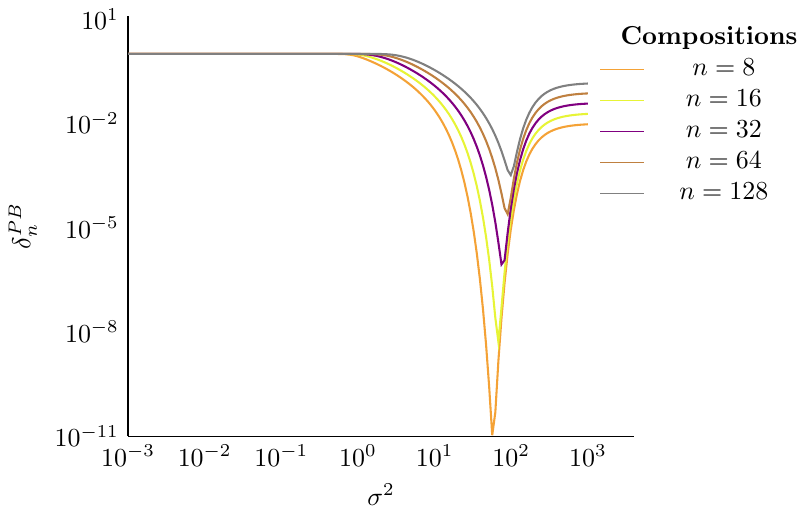}  
    \caption{Illustrating the (non-trivial) continuous mapping between the $\delta_{n}^{PB}$ and the standard deviation for truncated Gaussians.}
    \label{fig:std-dev-struncated-gaussian}
\end{figure}
\begin{figure}[hbt!]
    \centering
    \includegraphics[width=0.48\textwidth]{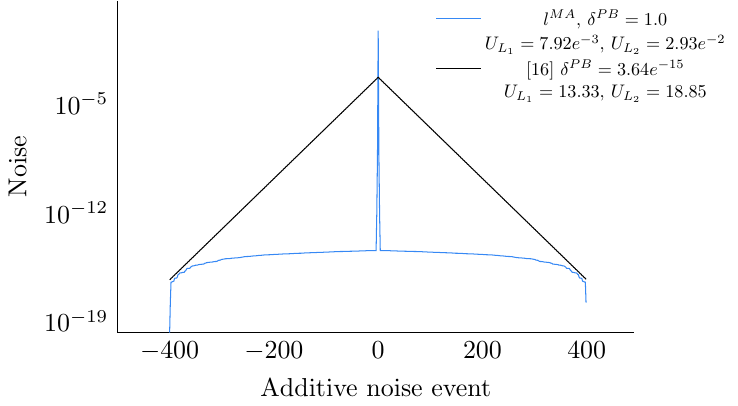}
    \caption{Illustration of a collapsed model trained using Extended Moments Accountant with $\eps=0.3$, $n=4$ compositions, and no utility weight decay, compared to Geng et al.~\cite{geng2017staircase} with $\eps=\frac{0.3}{4}$. While the utility-loss is small for a noise with a utility-loss that considerably greater than the $\delta$ (factor 1000+), the resulting $\delta^{PB}$ is almost equal to $1$, going against the purpose of finding suitable noise. \redcomment{Learning rate $\theta=0.01$.}}
    \label{fig:utility_domination}
\end{figure}
\begin{figure*}[htb!]
	    \centering
	    \begin{subfigure}[c]{0.29\textwidth}
	        \includegraphics[width=1.0\textwidth]{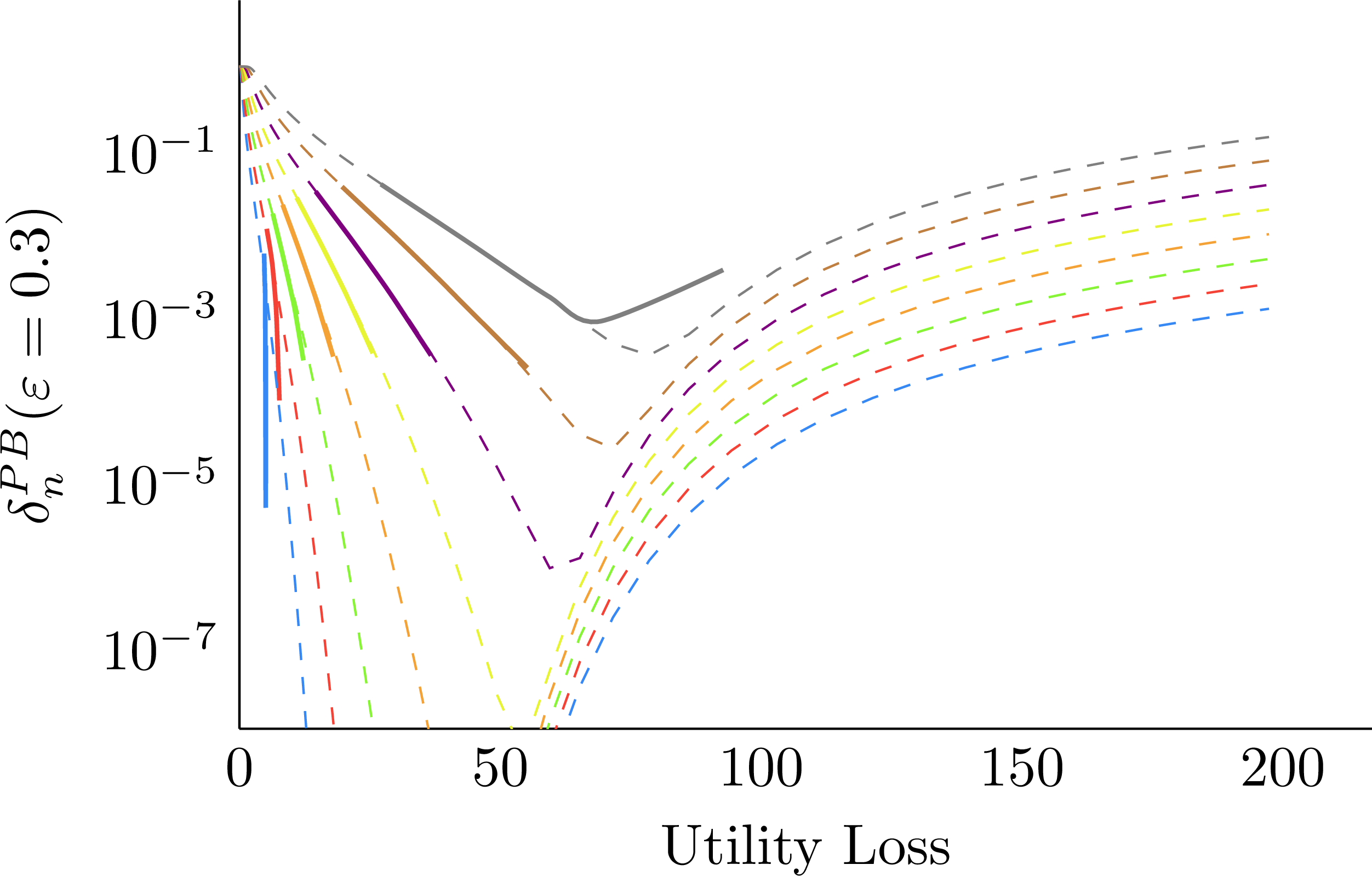}
	        \caption{PrivacyBuckets ADP with $L_1$ utility-loss}
            \label{fig:optimality_graph:old:pb_adp_l1_uwd_false}
	    \end{subfigure}
	    \begin{subfigure}[c]{0.29\textwidth}
    	    \includegraphics[width=1.0\textwidth]{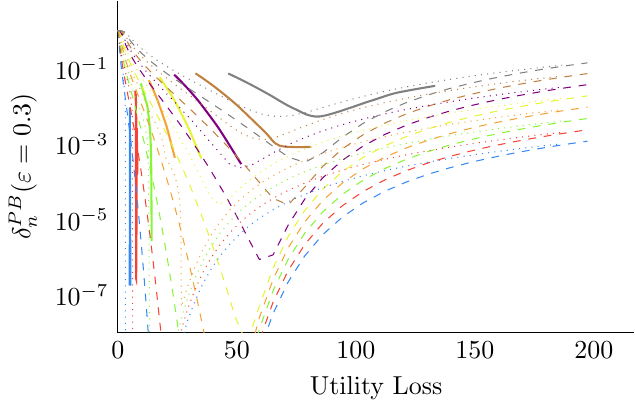}
	        \caption{PrivacyBuckets PDP  with $L_1$ utility-loss}
            \label{fig:optimality_graph:pb_pdp_l1_uwd_false}
	    \end{subfigure}
	    \begin{subfigure}[c]{0.29\textwidth}
    	    \includegraphics[width=1.0\textwidth]{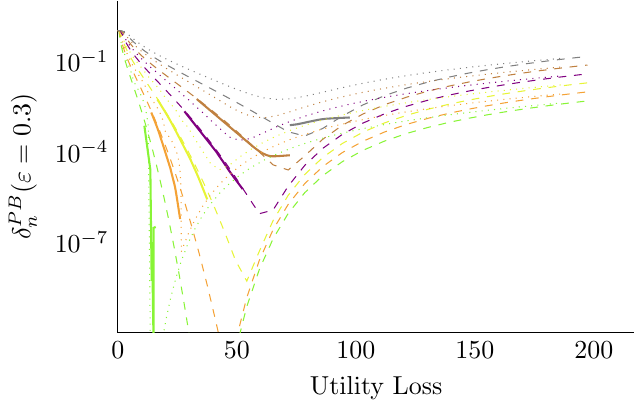}
	        \caption{Extended MA with $L_1$ utility-loss}
            \label{fig:optimality_graph:old:pb_renyi_markov_l1_uwd_false}
	    \end{subfigure}
	    \hfill\phantom{*}
	    \begin{subfigure}[c]{0.29\textwidth}
	        \includegraphics[width=1.0\textwidth]{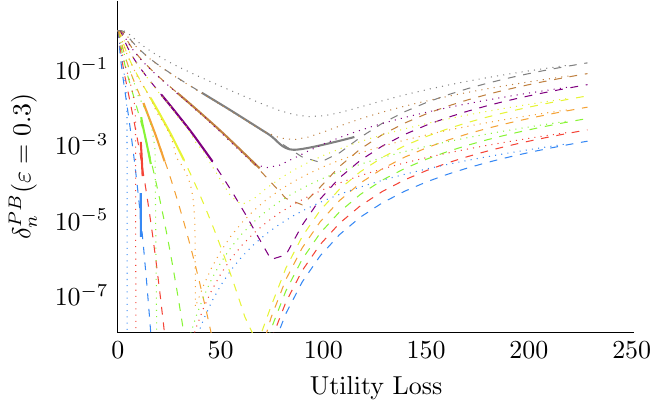}
	        \caption{PrivacyBuckets ADP with $L_2$ utility-loss}
            \label{fig:optimality_graph:old:pb_adp_l2_uwd_false}
	    \end{subfigure}
	    \begin{subfigure}[c]{0.29\textwidth}
    	    \includegraphics[width=1.0\textwidth]{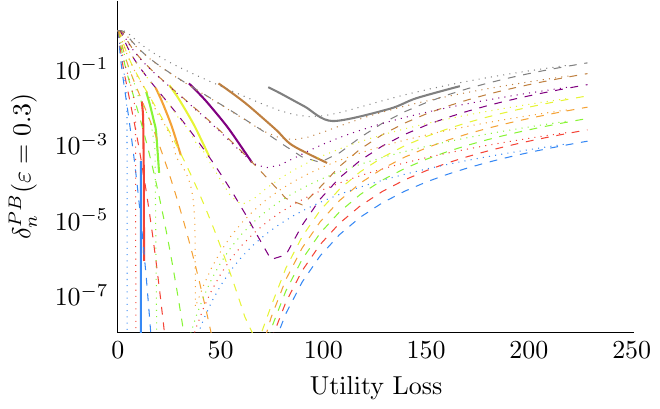}
	        \caption{PrivacyBuckets PDP with $L_2$ utility-loss}
            \label{fig:optimality_graph:pb_pdp_l2_uwd_false}
	    \end{subfigure}
	    \begin{subfigure}[c]{0.29\textwidth}
    	    \includegraphics[width=1.0\textwidth]{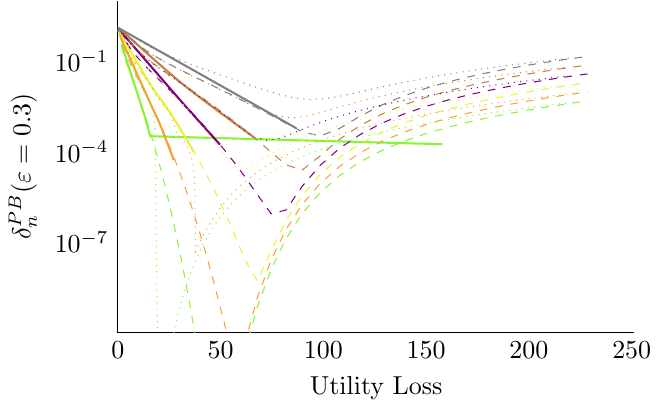}
	        \caption{Extended MA with $L_2$ utility-loss}
            \label{fig:optimality_graph:old:pb_renyi_markov_l2_uwd_false}
	    \end{subfigure}
	    \begin{subfigure}[c]{0.1\textwidth}\vspace{-10em}
    	    \includegraphics[width=1.0\textwidth]{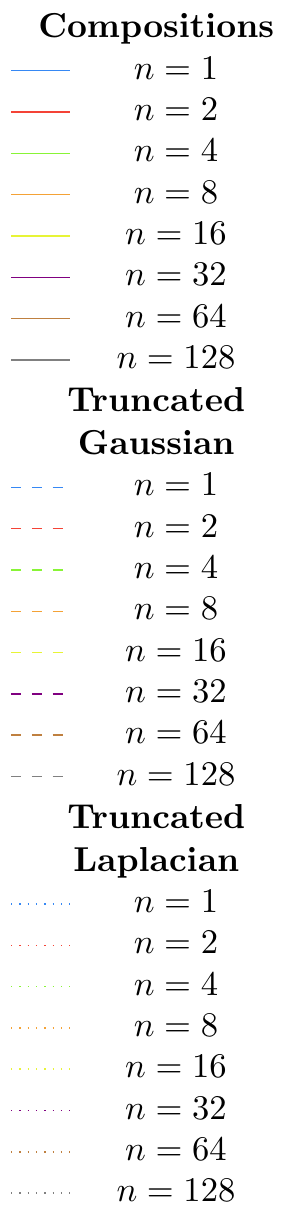}
	    \end{subfigure}
	    \begin{subfigure}[t]{0.29\textwidth}
	        \includegraphics[width=1.0\textwidth]{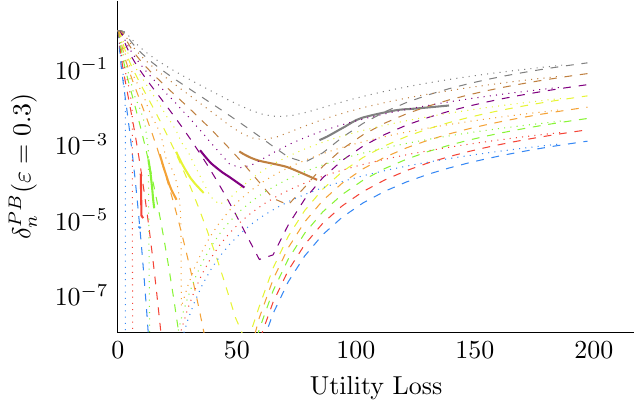}
	        \caption{PrivacyBuckets ADP with $L_1$ utility-loss and UWD}
            \label{fig:optimality_graph:old:pb_adp_l1_uwd_true}
	    \end{subfigure}
	    \begin{subfigure}[t]{0.29\textwidth}
            \includegraphics[width=1.0\textwidth]{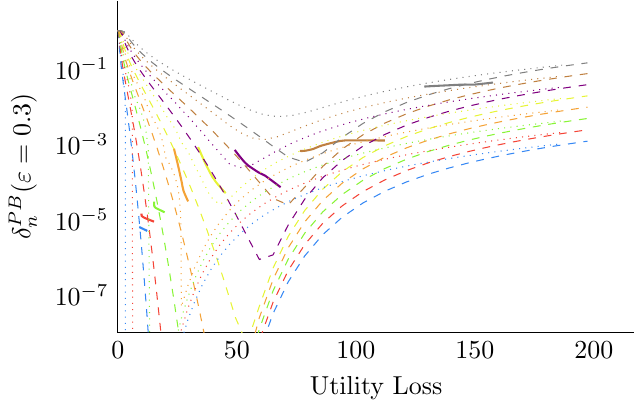}
	        \caption{PrivacyBuckets PDP with $L_1$ utility-loss and UWD}
            \label{fig:optimality_graph:pb_pdp_l1_uwd_true}
	    \end{subfigure}
	    \begin{subfigure}[t]{0.29\textwidth}
     	    \includegraphics[width=1.0\textwidth]{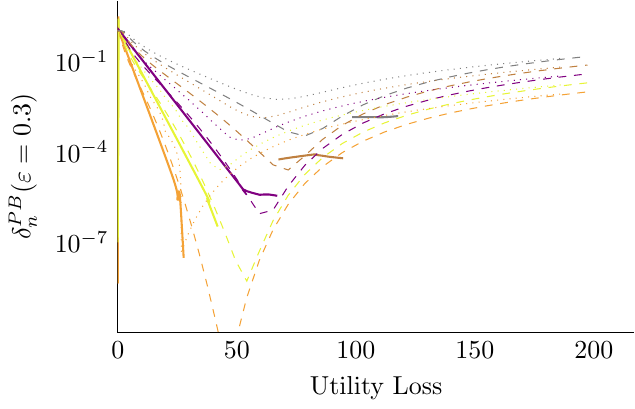}
	       \caption{Extended MA with $L_1$ utility-loss and UWD}
           \label{fig:optimality_graph:old:pb_renyi_markov_l1_uwd_true}
	   \end{subfigure}
	    \hfill \phantom{*}
	    \begin{subfigure}[t]{0.29\textwidth}
	        \includegraphics[width=1.0\textwidth]{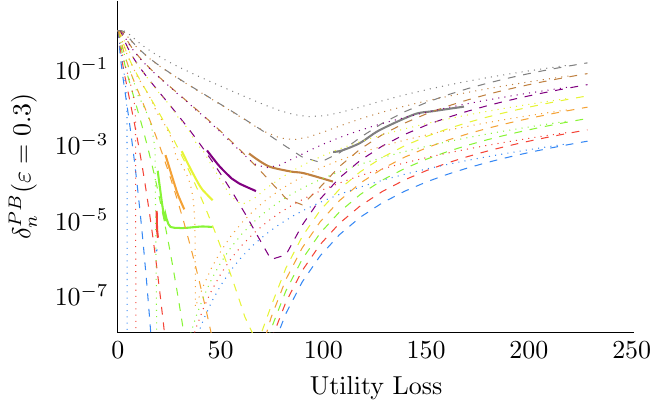}
	        \caption{PrivacyBuckets ADP with $L_2$ utility-loss and UWD}
            \label{fig:optimality_graph:old:pb_adp_l2_uwd_true}
	    \end{subfigure}
	    \begin{subfigure}[t]{0.29\textwidth}
     	    \includegraphics[width=1.0\textwidth]{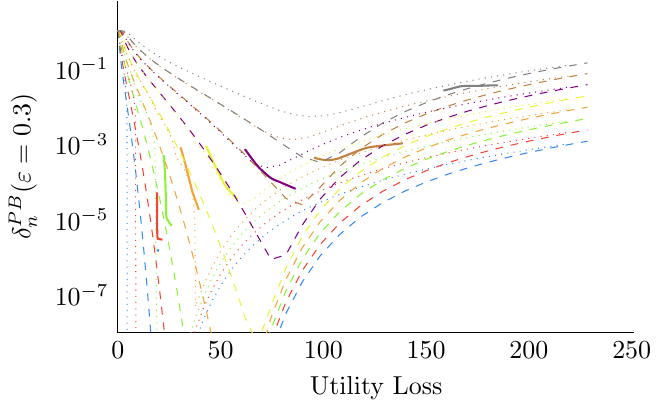}
	        \caption{PrivacyBuckets PDP with $L_2$ utility-loss and UWD}
            \label{fig:optimality_graph:pb_pdp_l2_uwd_true}
	    \end{subfigure}
	    \begin{subfigure}[t]{0.29\textwidth}
     	    \includegraphics[width=1.0\textwidth]{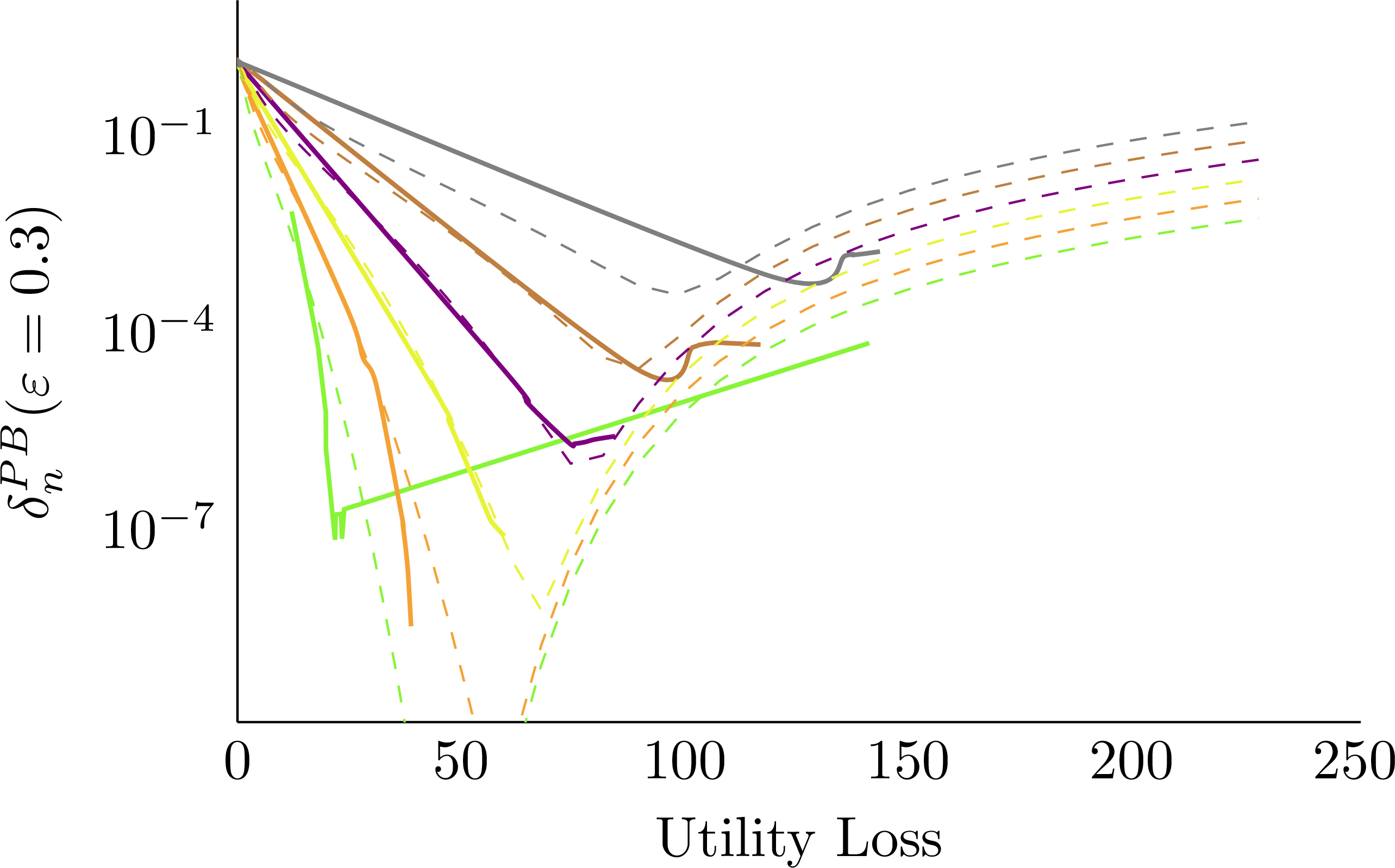}
	        \caption{Extended MA with $L_2$ utility-loss and UWD}
            \label{fig:optimality_graph:old:pb_renyi_markov_l2_uwd_true}
	    \end{subfigure}
	    \hfill \phantom{*}
        \caption{Utility-loss vs. privacy leakage for sensitivity-bounded queries (sensitivity $s\!=\!1$) and multiple number of compositions $n$. PrivacyBuckets-ADP $l^{ADP}$\!, PrivacyBuckets-PDP $l^{PDP}$\!, and Extended Moments Accountant (MA) $l^{MA}$\!, with and without utility-weight decay (UWD), are compared to the truncated Gaussian mechanism. Note that we have not proven \redemph{Shift Invariance} (\cref{thm:shift-invariance}) for the PDP-case but still apply it here. For comparative reasons, we show for $l^{PDP}$ an ADP-bound despite the PDP-formulation used during optimization. We removed Extended Moments Accountant (MA) with $n\in \{1,2,4\}$ due to numerical instability. 
        We also excluded data-Tuples with collapsing noise functions (producing $\delta \approx 1$). $\eps=0.3$.
        }
	    \label{fig:optimality_graph:old}
    \end{figure*} 
    \begin{figure*}[htb!]
	    \centering
	    \begin{subfigure}[t]{0.42\textwidth}
	        \includegraphics[width=1.0\textwidth]{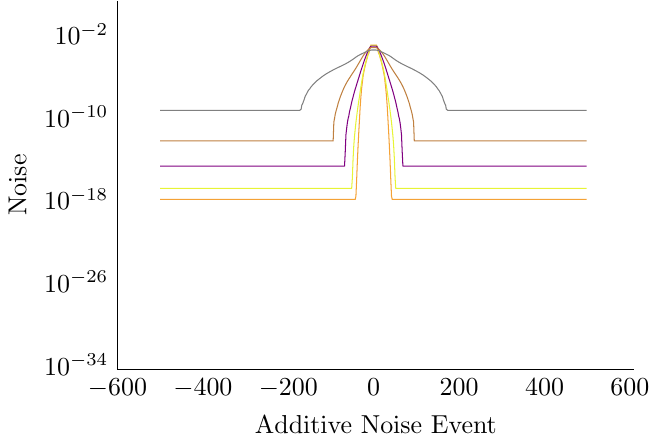}
	        \caption{PrivacyBuckets ADP trained on $l_{A,B}^{ADP}$ for multiple number of compositions.}
            \label{fig:dp_sgd:pb_adp_a_b}
	    \end{subfigure}
	    \begin{subfigure}[t]{0.42\textwidth}
	        \includegraphics[width=1.0\textwidth]{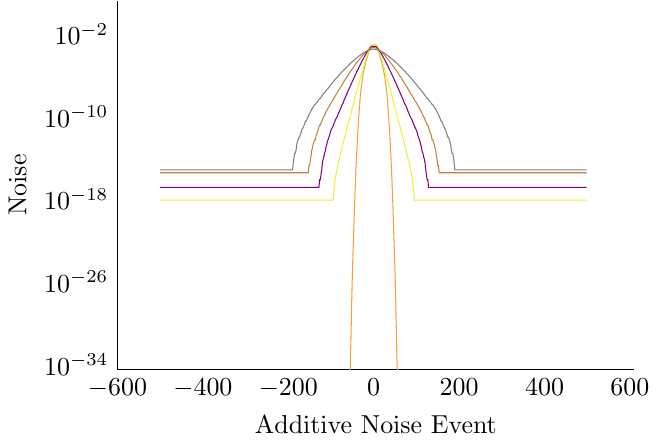}
	        \caption{MA trained on $l_{A,B}^{MA}$ for multiple number of compositions.}
            \label{fig:dp_sgd:ma_a_b}
	    \end{subfigure}
	    \begin{subfigure}[t]{0.12\textwidth}\vspace{-13em}
    	    \includegraphics[width=0.95\textwidth]{images/texfiles/legend/legend_dp_sgd}
	    \end{subfigure}
	    \hfill\phantom{*}
	    \begin{subfigure}[t]{0.42\textwidth}
	        \includegraphics[width=1.0\textwidth]{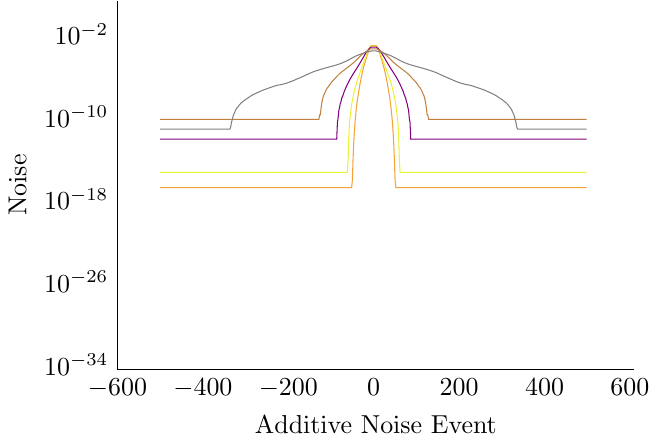}
	        \caption{PrivacyBuckets PDP trained on $l_{A,B}^{PDP}$ for multiple number of compositions.}
            \label{fig:dp_sgd:pb_pdp_a_b}
	    \end{subfigure}
	    \begin{subfigure}[t]{0.42\textwidth}
	        \includegraphics[width=1.0\textwidth]{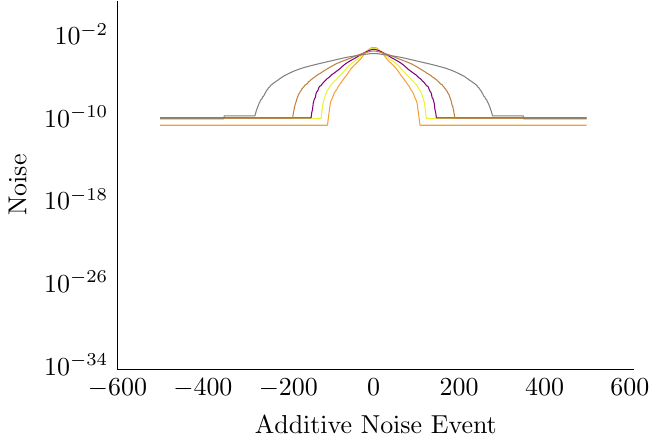}
	        \caption{PrivacyBuckets PDP trained on $l_{B,A}^{PDP}$ for multiple number of compositions. $q=0.1$, $C=1$, $\eps=0.3$}
            \label{fig:dp_sgd:pb_pdp_b_a}
	    \end{subfigure}
	    \hfill\phantom{*}
	    \caption{Additional generated noise $p$ for multiple number of compositions $n$ for completeness. Note that we have not proven \redemph{Shift Invariance} (\cref{thm:shift-invariance}) for the PDP-case but still apply it here.}
        \label{fig:dp_sgd_appendix}
    \end{figure*}

\begin{figure}[htb!]
    \centering
    \begin{subfigure}[t]{0.48\textwidth}
        \phantom{**************}
        \includegraphics[width=0.6\textwidth]{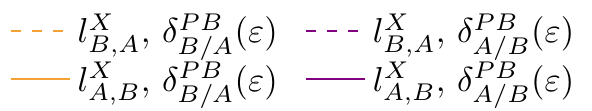}
    \end{subfigure}
    \begin{subfigure}[t]{0.48\textwidth}
        \includegraphics[width=1.0\textwidth]{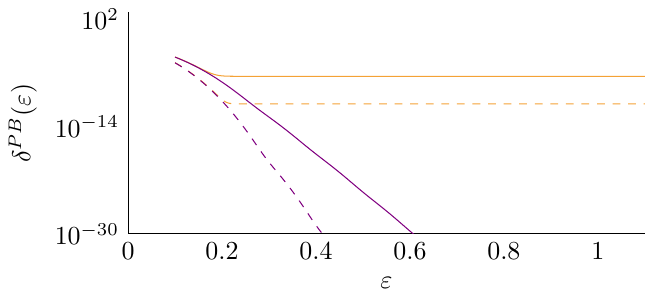}
        \caption{$l^{ADP}$}
    \end{subfigure}
    \begin{subfigure}[t]{0.48\textwidth}
        \includegraphics[width=1.0\textwidth]{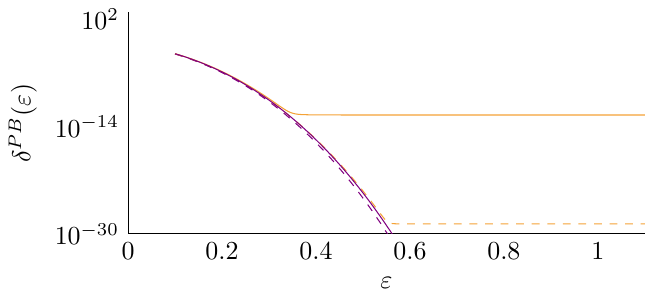}
        \caption{$l^{MA}$}
    \end{subfigure}
    \caption{$\delta^{PB}(\eps)$-graphs for generated DP-SGD worst-case distributions (\cref{eq:DP-SGD-optimzer-worst-case-distribtutions}) and $n=128$ compositions. For $X\in\{MA,ADP\}$, the preferable distribution is generated by $l_{B,A}^X$, resulting in a lower $\delta^{PB}$ than $l_{A,B}^X$. The final $\delta^{PB}$ is $\max(\delta_{A,B}^{PB}, \delta_{B,A}^{PB})$ with $p$ found according to $l_{B,A}$. The noise $p$ was generated by $\eps=0.3$, $q=0.1$, $C=1$.}
    \label{fig:delta-eps-graph-DP-SGD}
\end{figure}

\begin{figure}
	  \begin{subfigure}[t]{0.69\columnwidth}
	       \includegraphics[width=1.0\textwidth]{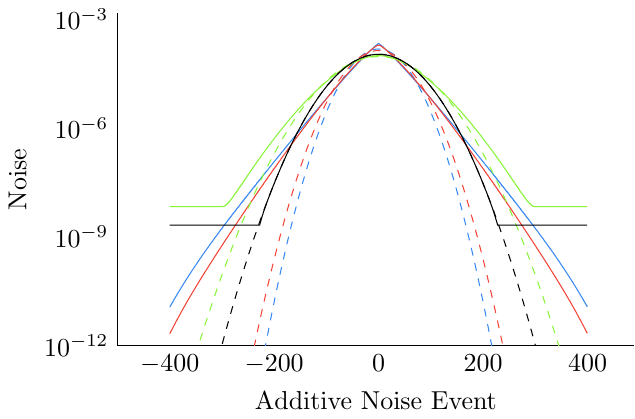}
	    \end{subfigure}
	    \begin{subfigure}[t]{0.3\columnwidth}
        \includegraphics[width=0.8\columnwidth]{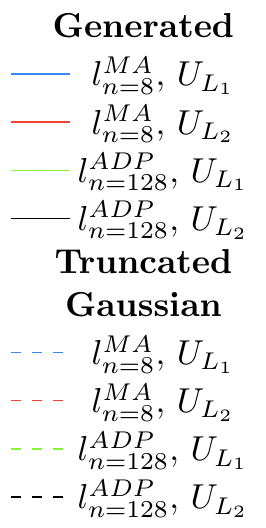}
	    \end{subfigure}
	    \caption{Selected generated noise distributions for $n$=$8,128$ with corresponding truncated Gaussians producing equal utility-loss. Illustrates the closeness for $l^{ADP}_{n\!=\!128}$.}\label{fig:optimalit:showcase}
\end{figure}

\new{
\section{CNN Model}\label{appendix:cnn_model}
    We also evaluated a noise model based on neural network design elements, namely linear-layers and 1D-convolutions, commonly referred to as \textit{CNN} model. However, we discovered significant numerical instabilities leading to overly high $\delta^{PB}$\!, especially but not limited to when we tried to enforce monotonicity. Therefore, we use this model to illustrate non-monotonic noise without the requirement to produce a low $\delta^{PB}$\!. Similarly to \cref{eq:model}, we evaluate the model by  considering only the points on the negative x-axis. For weight matrices $G$ and $H$ with bias vectors $I$ and $J$, filters $K_t$ ($t \in [0,10)$), $\circledast$ denoting convolution, and $i\in\{1,\ldots N\}$ and $j\in\{1,\ldots 100\}$ describing layer-widths, we obtain $p$ as follows:
    \newcommand{\CNNDenseONE}{\ensuremath{G}}
    \newcommand{\CNNDenseTWO}{\ensuremath{H}}
    \newcommand{\CNNBiasONE}{\ensuremath{I}}
    \newcommand{\CNNBiasTWO}{\ensuremath{J}}
    \newcommand{\CNNConvolution}{\ensuremath{K}}
        \begin{alignat*}{3}
            r_{j,0} &= x\cdot \CNNDenseONE^T + \CNNBiasONE_j && \quad\text{// linear layer}\\
            r_{j,t+1} &= r_{j,t} \circledast \CNNConvolution_t \text{ for } t\in [0,10) && \quad\text{// 10 conv layers}\\
            v_{i} &= r_{10}\cdot \CNNDenseTWO^T 
            + \CNNBiasTWO_i && \quad\text{// linear layer}\\
            p_i &= \frac{1}{2}\SoftMax\left(v_{i}; \{v_{0}, \ldots, v_{N}\}\right)&&\quad\text{// $1^{\text{st}}$\! noise half}\\
            p_k &= p_{2N-k+1}\quad\text{for } k\in\{N\!+\!1,\ldots,&& 2N\}  \,\,\,\text{// $2^{\text{nd}}$\! half}
        \end{alignat*}
}

\section{Extended Moments Accountant (Full Proof)}\label{sec:full-extended-MA}
The following theorem generalizes Moments Accountant by Abadi et al.~\cite{abadi2016deeplearning} to mechanisms producing distinguishing events. This extension, however, comes at a formal cost: we cannot simply include adaptive mechanisms by taking the largest moment $\alpha(\lambda)$ for a given $\lambda$ to capture the maximal privacy-loss, as we need to ensure a second value, the probability mass of distinguishing events ($\bucketSymbol_\infty$), to be bounded simultaneously. Moreover, these two values, the maximal $\alpha(\lambda)$ and the maximal $\bucketSymbol_\infty$ need to be produced by a single pair of worst-case output distributions $M(D_0)$ and $M(D_1)$ such that no inputs $D, D'$, and $aux$ to the mechanism lead to more privacy leakage. These output distributions may come from a worst-case input triple $(D_0, D_1, aux)$ but can also be defined freely, as long as they are worst-case.
In literature, many mechanisms implicitly assume such worst-case inputs, for example, by defining a sensitivity. The work by Abadi et al. reduces its privacy analysis to worst-case output distribution and compares Gaussian noise vs. a Gaussian-mixture. We now introduce extended definitions, which are required to prove our generalized Moments Accountant. 

\begin{definition}[Worst-case inputs for MA]\label{def:worst-case-MA}
For any neighboring $D, D'\in \mathcal D$, mechanism $M$, auxiliary input $aux$
\begin{align*}
    c(o_i;M_i,aux,D,D')&\phantom{**********************}
\end{align*}
\begin{equation*}
\hspace{-0.9pt}\!=\!\begin{cases}
\infty &\hspace{-1.0em}\text{if }  \Pr[M_{1:i}(D;aux)\!=\!o_{1:i}]\! \neq\! 0\\
       &\hspace{-1.0em}\text{and }\Pr[M_{1:i}(D';aux)\!=\!o_{1:i}]\! =\! 0\\
\ln{\!\frac{\Pr[M_{1:i}(D;aux) = o_{1:k}]}{\Pr[M_{1:i}(D';aux)\!=\! o_{1:i}]}} & \hspace{-1.0em}\text{if }  \Pr[M_{1:i}(D;aux)\!=\! o_{1:i}]\! \neq\! 0\\
&\hspace{-1.0em}\text{and }\Pr[M_{1:i}(D';aux)\!=\!o_{1:i}]\! \neq\!0\\
-\infty & \hspace{-1.0em}\text{else }
\end{cases}
\end{equation*}
\begin{gather*}
\alpha_{M_i}(\lambda; aux, D,\! D') \!=\! \log \expectation{\mathrlap{o_i \sim M_i(D)}}[e^{\lambda c(o_i;M_i,aux,D,D')}\vert o_i\!\neq\!\infty, \text{-}\infty]\\
\bucket[\infty]{M_i,aux,D,D'}\! =\! \Pr_{\mathrlap{o\sim M_{i}(D)}}[c(o;M_i,aux,D,D') = \infty] 
\end{gather*}
Let $D_0,D_1$ be worst-case inputs such that the following is true for either $k=0$ or $k=1$:
\begin{align*}
    aux_k &\!=\! \argmax_{aux'} \alpha_{M_i}(\lambda; aux', D_k, D_{1\!-\!k})\\
    aux_k &\!=\! \argmax_{aux'} \bucketSymbol_{\hspace{-0.9pt}\infty}\hspace{-0.9pt}(M_i,aux',D_k,D_{1\!-\!k})\\
    \alpha_{M_i}(\lambda; aux_k, D_k, D_{1\!-\!k}) &\!\geq\! \max_{\hspace{1.8em}\mathclap{D,D',aux'}} \alpha_{M_i}(\lambda; aux', D, D')\\
    \bucket[\infty]{M_i,aux_k,D_k,D_{1\!-\!k}} &\!\geq\! \max_{\hspace{1.8em}\mathclap{D,D',aux'}} \bucket[\infty]{M_i,aux',D,D'}
\end{align*}
Further let 
\begin{align*}
    \alpha_{M_i}(\lambda; D_k, D_{1-k})\! =\! \alpha_{M_i}(\lambda; aux_k, D_k, D_{1-k})\\
    \bucket[\infty]{M_i, D_k,D_{1-k}}\! =\! \bucket[\infty]{M_i,aux_k,D_k,D_{1-k}}
\end{align*}
\end{definition}
Contrarily to the theorem by Abadi et al., we defined $\alpha_M(\lambda)$ to be the expectation over a normalised output distribution generated by $M(D_0)$: $\Pr[M(D_0;aux)\!=\!o\mid o\neq\infty,\minfty] =  \frac{\Pr[M(D;aux)\!=\!o]}{1-\bucket[\infty]{M,aux,D_0,D_1}}$ (the events with $o=\minfty$ do not occur). The privacy loss $c(o)$, however, is generated by non-normalised distributions. Also note that only one pair ($\alpha_{M_i}(\lambda; D_l, D_{1-l}),\bucket[\infty]{M_i, D_l,D_{1-l}}$)$_{i\in\{0,1\}}$ needs to dominate all other values. Which one, however, might depend on the required $\eps$. See \cref{fig:delta-eps-graph-DP-SGD} for an illustrative example.

Under the assumptions that the algorithm we consider allows such worst-case inputs or a reduction to output distributions that lead to an always higher $\alpha (\lambda)$ and $\bucketSymbol_\infty$, we can prove ADP-guarantees for such an algorithm. 
\begin{theorem}[Generalized Moments Accountant]\label{thm:apprixmate-rdp-composition}
Let the variables be defined as in \cref{def:worst-case-MA}. Then for $k\in\{0,1\}$, 
\begin{enumerate}
    \item \textbf{[Composability]} Suppose that a mechanism M consists of a sequence of adaptive mechanisms $M_1,\ldots,M_n$ where $M_i=\prod_{j=1}^{i-j} R_j\times \mathcal D\rightarrow R_i$. Then, for any $\lambda > 0$ 
    \begin{align*}
        \alpha_M(\lambda; D_k, D_{1-k}) &\!\leq\! \sum_i^n \alpha_{M_i}(\lambda; D_k, D_{1-k})\\[-6pt]
        \bucket[\infty]{M,D_k, D_{1-k}} &\!=\! 1 - \prod_i^n \left[1-\bucket[\infty]{M_i; D_k, D_{1-k}}\right]
    \end{align*}
    \item \textbf{[Tail Bound]} For any $\eps>0$, $M$ is $(\eps, \delta)$-differentially private for $\delta = \max(\delta_{D_0, D_1}, \delta_{D_1, D_0})$ with 
    \begin{align*}
        \delta_{D_k, D_{1-k}} =&\, \bucket[\infty]{M,D_k, D_{1-k}} \\
        &+ \min_\lambda \left(1\!-\!\bucket[\infty]{M,D_k, D_{1-k}}\right)\\[-6pt]
        &\hphantom{\,\,+ \min_\lambda}\,\cdot \exp{ \left(\alpha_M(\lambda;D_k, D_{1-k} )\!-\! \lambda\cdot\eps\right)}
    \end{align*}
\end{enumerate}
\end{theorem}
\begin{proof}[\Cref{thm:apprixmate-rdp-composition}]
This proof is inspired by the proof of Theorem 2 in \cite{abadi2016deeplearning} with considerations from Theorem 1 in~\cite{sommer2019privacy}.
\textbf{Composability:} 
For $k\in\{0,1\}$, there are three sets an output of $M_i(D_k)$ can occur:
\begin{align*}
    U_{M_i}^\infty &= \left\{o_i\mid c(o_i;M_i,aux_k,D_k,D_{1-k}) = \infty \right\}\\
    U_{M_i}^+ &= \left\{o_i\mid c(o_i;M_i,aux_k,D_k,D_{1-k}) \notin \{\infty,\infty\} \right\}\\
    U_{M_i}^\minfty &= \left\{o_i\mid c(o_i;M_i,aux_k,D_k,D_{1-k}) = \minfty \right\}
\end{align*}
These sets might differ depending on $i$. We now consider only the composition of $M_i$ and $M_{i+1}$. There, outputs $o_i$ and $o_{i+1}$ each originating from one of their individual three sets, resulting in 9 possible combinations: 
$U_{M_i}^\infty \times U_{M_{i+1}}^\infty$, $U_{M_i}^\infty \times U_{M_{i+1}}^+$, $U_{M_i}^\infty \times U_{M_{i+1}}^\minfty$,
$U_{M_i}^+ \times U_{M_{i+1}}^\infty$,  $U_{M_i}^+ \times U_{M_{i+1}}^+$, $U_{M_i}^+ \times U_{M_{i+1}}^\minfty$, 
$U_{M_i}^\minfty \times U_{M_{i+1}}^\infty$, $U_{M_i}^\minfty \times U_{M_{i+1}}^+$, $U_{M_i}^\minfty \times U_{M_{i+1}}^\minfty$.
Sommer et al.~\cite{sommer2019privacy} have shown in Theorem 1 that all sets including a privacy loss of $\minfty$ have zero occurrence probability and can be neglected~\cite{sommer2019privacy}. Moreover, Abadi et al. have already shown the case $U_{M_i}^+ \times U_{M_{i+1}}^+$:
\begin{equation*}
    \alpha_{M_{i+1} \circ M_i}(\lambda) \leq \alpha_{M_i}(\lambda) + \alpha_{M_{i+1}}(\lambda)
\end{equation*}
Their proof requires only independence of output events in the analysed worst-case setting, which is valid for us as well. 

We show now that the combinations $U_{M_i}^\infty \times U_{M_{i+1}}^\infty$, $U_{M_i}^\infty \times U_{M_{i+1}}^+$, and $U_{M_i}^+ \times U_{M_{i+1}}^\infty$ all contribute solely to $\bucketSymbol_\infty$. The case $U_{M_i}^\infty \times U_{M_{i+1}}^\infty$ was already shown by~\cite{sommer2019privacy} (proof Theorem 1).
In the case of $U_{M_i}^\infty \times U_{M_{i+1}}^+$, we have by definition: $c(o_i;M_i,aux_l,D_l,D_{1-l}) = \infty$ and $c(o_{i+1};M_{i+1},aux_k,D_k,D_{1-k}) \neq \minfty$. It follows
\begin{align*}
&c(o_i;aux_k,D_k,D_{1-k}) = \infty \\
&\Longrightarrow \hphantom{\text{an}}\quad \Pr[M_{i}(D_k)\!=\!o_{i}|aux_k] \!\neq \!0\\
&\hphantom{\Longrightarrow}\text{and}\quad \Pr[M_{i}(D_{k-1})\!=\!o_{i}|aux_k] \!= \!0\\
\hspace{-4.3em}\makebox[0pt][l]{\text{and}}\hspace{4.3em}
&c(o_{i+1};M_{i+1},aux_k,D_k,D_{1-k}) \neq \minfty \\
&\Longrightarrow \pr{M_{i+1}(D_k)=o_{i+1}|aux_k} \neq 0
\end{align*}

We show now $c(o_{i:i+1};M_{i:i+1},aux_k,D_k,D_{1-k}) = \infty$:
\begin{align*}
&c(o_i;M_i,aux_k,D_k,D_{1-k}) = \infty \\
&\Longrightarrow\hphantom{\text{an}}\quad \Pr[M_{i}(D_k)\!=\!o_{i}|aux_k]\!\neq\! 0 \\
&\hphantom{\Longrightarrow}\text{and}\quad \Pr[M_{i}(D_{k\!-\!1})\!=\!o_{i}|aux_k]\!= \!0\\
&\Longrightarrow\hphantom{\text{an}}\quad  \pr{M_{i}(D_k)=o_{i}|aux_k} \neq 0 \\
&\hphantom{\Longrightarrow}\text{and}\quad  \prod_{j=i}^{i+1}\pr{M_{j}(D_{k-1})=o_{k}|aux_k} = 0\\
&\stackrel{(i)}{\Longrightarrow}\hphantom{\text{an}}\quad \prod_{j=i}^{i+1}\pr{M_{j}(D_k)=o_{k}|aux_k} \neq 0 \\
&\hphantom{\Longrightarrow}\text{and}\quad  \prod_{j=i}^{i+1}\pr{M_{j}(D_{k-1})=o_{k}|aux_k} = 0\\
&\Longrightarrow\hphantom{\text{an}}\quad  \Pr[M_{i:i+1}(D_k;aux_k)\!=\!o_{i:i+1}] \!\neq \!0 \!\\
&\hphantom{\Longrightarrow}\text{and}\quad \Pr[M_{i:i+1}(D_{k\!-\!1};aux_k)\!=\!o_{i:i+1}]\! =\! 0\\
&\Longrightarrow\hphantom{\text{an}}\quad  c(o_{i:i+1};M_{i:i+1},aux_k,D_k,D_{1-k}) = \infty
\end{align*}
where we used (i) $\pr{M_{i+1}(D_k)=o_{i+1}|aux_k} \neq 0$. The claim for $U_{M_i}^+ \times U_{M_{i+1}}^\infty$ follows analogously to $U_{M_i}^\infty \times U_{M_{i+1}}^+$ by switching indices. 
The statement of composability follows by induction.

\textbf{Tail-Bound:} For a fixed $k\in\{0,1\}$, we show 
$\delta_{D_k,D_{k-1}}= \min_{\lambda} (\bucketSymbol_\infty) + (1-\bucketSymbol_\infty))\cdot\exp(\alpha_M(\lambda) - \lambda \eps))$.
In accordance with \cite{sommer2019privacy}, we consider a privacy loss $c(o) = \infty$ to be larger than any $\epsilon\in \mathbb{R}$. Then we have, 
\begin{align*}
&\Pr_{o \sim M(D_k)}[c(o) \geq \eps] \\
&= (1-\Pr_{\mathclap{o \sim M(D_k)}}[c(o)= \infty])\cdot \Pr_{\mathclap{o \sim M(D_k)}}[c(o) \geq \eps|c(o) \neq \infty]\\
&\quad + \Pr_{\mathclap{o \sim M(D_k)}}[c(o)= \infty]\cdot \Pr_{\mathclap{o \sim M(D_k)}}[c(o) \geq \eps|c(o) = \infty] \\
&= (1-\bucketSymbol_\infty)\cdot \Pr_{\mathclap{o \sim M(D_k)}}[c(o) \geq \eps|c(o) \neq \infty] + \bucketSymbol_\infty \\
&= (1\!-\!\bucketSymbol_\infty)\cdot \Pr_{\mathclap{o \sim M(D_k)}}[\exp (\lambda c(o)\! \geq\! \exp (\lambda \eps)|c(o)\! \neq\! \infty])\! +\! \bucketSymbol_\infty\\
&\stackrel{MI}{\leq} (1-\bucketSymbol_\infty)\frac{E_{o \sim M(D_k)}[\exp (\lambda c(o)))|c(o) \neq \infty]}{\exp(\lambda \eps)} + \bucketSymbol_\infty\\
&\stackrel{\substack{\alpha_M(\lambda)\\= \max \alpha_M(\lambda; aux, D_k, D_{1-k})}}{\leq} (1\!-\!\bucketSymbol_\infty)\exp(\alpha_M(\lambda) - \lambda \eps) + \bucketSymbol_\infty
\end{align*}
Where we used Markov's inequality (MI). 
After this we can use a similar ADP reduction as seen in~\cite{abadi2016deeplearning}.
\begin{align*}
    \intertext{Let $B=\{o: c(o) \geq \epsilon\}$. Then for all $S$.}
    &\Pr[M(D)\! \in \!S]\! = \!\Pr[M(D')\!\in\! S \cap B^{C}]\! + \!\Pr[M(D)\! \in S \!\cap B]\\
    &\leq e^\epsilon \cdot\Pr[M(D')\!\in\! S \cap B^{C}]\! +\! \Pr[M(D)\!\in\! B]\\
    &\leq e^\epsilon\cdot \Pr[M(D')\!\in\! S]\! +\! (1\!-\!\bucketSymbol_\infty)\exp(\alpha_M(\lambda)\! - \!\lambda \eps)\! +\! \bucketSymbol_\infty
\end{align*}
The claim follows as the result is valid for all $\lambda$. 
Finally, we take the maximum $\delta = \max(\delta_{D_0, D_1}, \delta_{D_1, D_0})$, fulfilling ADP-guarantees for $aux_k, D_0$, and $D_1$. For a given $\epsilon$, we argue that for any inputs $D,D',aux$ to $M$, the resulting privacy-loss is smaller or equal to $\delta$ because by definition we have chosen the worst-case inputs such that $\forall \lambda$ no other inputs lead to a larger $\bucketSymbol_\infty$ or $\alpha(\lambda)$. As we have only used upper bounds or equality to express the impact on the privacy-loss, the ADP bound is guaranteed for any inputs $D,D',aux$ to $M$. 
\end{proof}

\begin{proof}[\Cref{lemma:delta-from-rdp}]
This is a special case of the more general \cref{thm:apprixmate-rdp-composition} for worst-case inputs $D_0, D_1$. Let $A(o)$=$\pr{o\!\leftarrow\! M(D_0)}$ and $B(o)$=$\pr{o\!\leftarrow\! M(D_1)}$.
\begin{align*}
\alpha_{M_i}(\lambda; aux, D_0, D_1) &= \log \sum_{o} \frac{A(o)}{1-\bucketSymbol_\infty} \exp\left( \lambda \log \frac{A(o)}{B(o)} \right) \\
&= \log \left[ \frac{1}{1-\bucketSymbol_\infty} \sum_{o} A(o) \left(\frac{A(o)}{B(o)}\right)^\lambda\right] \\
&=  \log \left[\frac{1}{1-\bucketSymbol_\infty}\right] + \Alphadiv{\lambda}{A}{B} 
\end{align*}
Applying $n$ compositions and the tail-bound of thm. \ref{thm:apprixmate-rdp-composition},
\begin{align*}
\delta_{D_k,D}&\vphantom{s}_{_{1\modifiedminus k}} \!\!=\!\bucketSymbol_{\infty,M}^{D_k,D_{1\modifiedminus k}}\! + \! \min_\lambda \left(1\!-\!\bucketSymbol_{\infty,M}^{D_k,D_{1\modifiedminus k}}\right)\!\cdot\! e^{ \left(\alpha_M^{D_k,D_{1\modifiedminus k}}(\lambda) \!- \!\lambda\cdot\eps\right)}\\
&=\!\bucketSymbol_{\infty,M}^{D_k,D_{1\modifiedminus k}}\! +\! \min_\lambda \left(1\!-\!\bucketSymbol_{\infty,M}^{D_k,D_{1\modifiedminus k}}\right)\!\cdot\! e^{ \left(n\cdot\alpha_{M_i}^{D_k,D_{1\modifiedminus k}}(\lambda) \!- \!\lambda\cdot\eps\right)}\\
&=\! \bucketSymbol_{\infty,M}^{D_k,D_{1\modifiedminus k}}\! + \min_\lambda \left(1\!-\!\bucketSymbol_{\infty,M}^{D_k,D_{1\modifiedminus k}}\right)\\
&\quad\cdot\! e^{ n\cdot\left( \log \left[\frac{1}{1-\bucketSymbol_\infty}\right] + \Alphadiv{\lambda}{A}{B}\right)  \!- \!\lambda\cdot\eps} \\
&=\! \bucketSymbol_{\infty,M}^{D_k,D_{1\modifiedminus k}}\! + \min_\lambda e^{ \left( n\cdot\Alphadiv{\lambda}{A}{B}  \!-\! \lambda\cdot\eps\right)} 
\end{align*}
because $\bucketSymbol_{\infty,M}^{D_k,D_{1\modifiedminus k}} = 1 - [1- \bucketSymbol_{\infty}]^n $. The claim follows.
\end{proof}

\section{Other Proofs}
 \begin{proof}[\Cref{lemma:symmetric-noise}]
 Define $\hat{p}(S)$ as follows: 
 $\forall$ measurable sets $S\subseteq R$: $\hat{p}(S) = \frac{p(S)+P(-S)}{2}$
This $\hat{p}$ is symmetric and as the Loss function $u$ is symmetric $p(S)$ and $\hat{p}$ have the same utility loss. Next we show that $\hat{p}$ also satisfies ($\epsilon$, $\delta$)-differential privacy:
\begin{align*}
    &|\hat{p}(S) - e^{\epsilon} \hat{p}(S+d)|\\
    =&|(\frac{p(S) - p(-S)}{2} - e^{\epsilon} \frac{p(S+d) + p(-S-d)}{2}|\\
    =&|\frac{p(S) - e^{\epsilon}p(S+d)}{2} - \frac{p(-S) - e^{\epsilon}p(-S-d)}{2}|\\
    \leq&|\frac{p(S) - e^{\epsilon}p(S+d)}{2}| + |\frac{p(-S) - e^{\epsilon}p(-S-d)}{2}|\\
    \leq& \frac{\delta}{2} + \frac{\delta}{2} =  \delta
\end{align*}
 similar to Geng et al.~\cite{geng2019pmlr}.
 \end{proof}
 
 \new{
\begin{proof}[\Cref{lemma:discretization-reduction}]
\newcommand{\vol}{\text{vol}}
\newcommand{\intd}{\text{d}}
\newcommand{\pdfp}[1]{\text{pdf}_p(#1)}
Let $\eps$ be fixed. Let $\nu = \min \{x_i - x_j | i,j\in\mathbb Z\}$ be the equidistant discretization step. We use $p(x)$ to denote the probability to sample a discrete $x\in X$ from $p$, i.e., the probability mass function.    
By definition, the continuous probability density function of $p$ has the shape of horizontal segments (plateaus), one for each $x\in X$ where that $x$ is in its center, and with width $\nu$:
\begin{equation*}
    \pdfp{q}\,dx = \sum_{x\in X} \indicatorfunc{q\in [x - \frac{\nu}{2}, x+\frac{\nu}{2})} \frac{p(x)}{\nu}\, dx
\end{equation*}
 It is sufficient to show that $\forall r, |r| \leq s$, and $\forall Q\subseteq R$ 
\begin{equation*}
     \pr[p]{Q} - e^\eps \pr[p]{Q+r} \leq \delta  
\end{equation*}
where we define $Q + r = \{q + r\,|\,q\in Q\}$. 

Let the the next lower and next higher integer multiple of $\nu$ of $r$ be defined by $r_l = \nu \cdot \lfloor \frac{r}{\nu} \rfloor$ and $ r_r = \nu \cdot \lfloor \frac{r}{\nu} + 1\rfloor$. 
For each $x\in X$, we define two subsets, depending on what segment (either $r_l$ or $r_r$ distant from $x$) the points $q$ residing on segment $x$ are shifted. 
\begin{align*}
    Q_l^x &= \{q\, |\, q -x  \in [-\frac{\nu}{2}, -\frac{\nu}{2} + r_r - r),\,q \in Q\}\\
    Q_r^x &= \{q\, |\, q-x  \in [-\frac{\nu}{2} + r_r - r, \frac{\nu}{2}),\,q \in Q\}
\end{align*}
Note that $\forall x\in X,\, \forall q,q' \in Q_l^x,\, \pdfp{q}$ = $\pdfp{q'}$ as all events on the same segment are distributed uniformly. Same for $Q_r^x$. We denote $Q_l = \bigcup_{x \in X} Q_l^x$ and $Q_r = \bigcup_{x \in X} Q_r^x$ and note that $Q = Q_l\cup Q_r$, and that all $Q_l^x$ and $Q_r^x$ are pairwise distinct. 

Let $\xi_l = \max_{x\in X} \vol(Q_l^x) / \nu \leq \vol([-\frac{\nu}{2}, -\frac{\nu}{2} + r_r - r)) / \nu$ and $\xi_r = \max_{x\in X} \vol(Q_r^x) / \nu \leq \vol([-\frac{\nu}{2} + r_r - r, \frac{\nu}{2})) / \nu$. By construction, $\xi_l + \xi_r \leq \vol([-\frac{\nu}{2}, \frac{\nu}{2})\,) / \nu = \frac{\nu}{\nu} = 1$.
Let $S_l = \{x | Q_l^x \neq \{\}, x\in X\}$ containing all relevant segment centers for $Q_l$. Similarly, let $S_r = \{x | Q_r^x \neq \{\}, x\in X\} = S_l + \nu$. Then,
\begin{align*}
    \pr{Q}& - e^\eps\, \pr{Q+r} \\
    =& \pr{Q_l} - e^\eps\, \pr{Q_l+r} + \pr{Q_r} - e^\eps\, \pr{Q_r+r}\\
    \stackrel{(i)}{=}& \sum_{x\in S_l} \int_{\mathrlap{Q_l^x}} \pdfp{q} - e^\eps\, \pdfp{q +r_l} \,\intd x\\
    &+ \sum_{x\in S_r} \int_{\mathrlap{Q_r^x}} \pdfp{q} - e^\eps\, \pdfp{q+r_r} \,\intd x\\
    \stackrel{(ii)}{\leq}&  \sum_{x\in S_l} \xi_l \cdot\int_{\mathrlap{[x - \frac{\nu}{2}, x+\frac{\nu}{2})}} \pdfp{q} - e^\eps\, \pdfp{q+r_l} \,\intd x\\
    &+ \sum_{x\in S_r} \xi_r \cdot\int_{\mathrlap{[x - \frac{\nu}{2}, x+\frac{\nu}{2})}} \pdfp{q} - e^\eps\, \pdfp{q+r_r} \,\intd x\\
    \stackrel{(iii)}{=}&  \xi_l \cdot \sum_{x\in S_l} p(x) - e^\eps\, p(x + r_l) \\
    &+ \xi_r \cdot \sum_{x\in S_r} p(x) - e^\eps\, p(x + r_r) \\
    \stackrel{(iv)}{\leq}& \xi_l\cdot \delta +  \xi_r \cdot \delta \,\,\,\leq \delta 
\end{align*}
with $(i)$ countability and absolute convergence for infinite sums is given as we consider well-defined probability densities always greater than $0$, $(ii)$ by construction, $(iii)$ applying definition of $p$, and $(iv)$ by initial assumption for any $S$. 
\end{proof}
}

    \begin{figure*}[t]
	    \centering
	    \begin{subfigure}[t]{0.41\textwidth}
	        \includegraphics[width=1.0\textwidth]{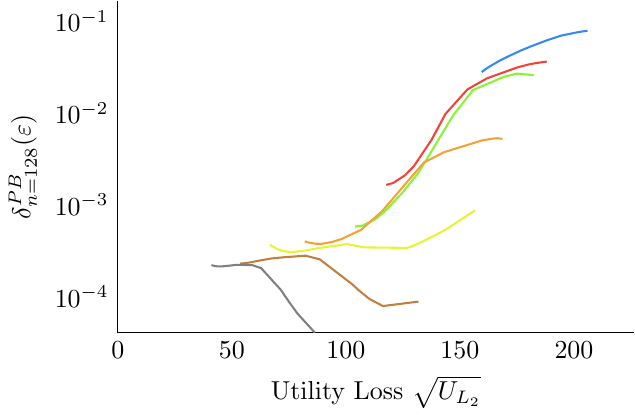}
           \caption{$L_2$ utility-loss vs. $\delta^{PB}_{n=128}$}
           \label{fig:different_eps:utility_delta}
	    \end{subfigure}
	    \begin{subfigure}[t]{0.41\textwidth}
    	    \includegraphics[width=1.0\textwidth]{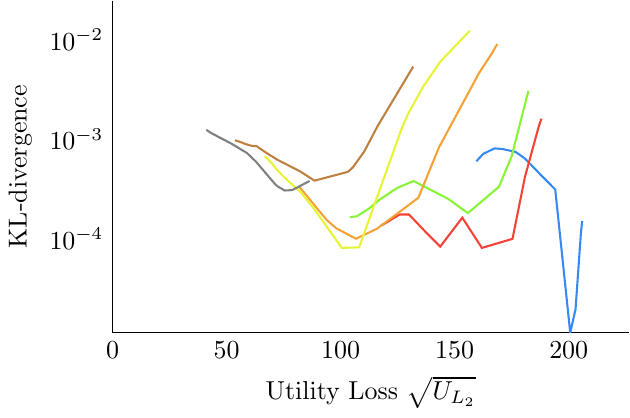}
            \caption{$D_{KL}(\text{\setstackgap{S}{1pt}\tiny\stackanchor{truncated}{Gaussian}}||\text{generated})$, decreasing with $\eps$.} 
            \label{fig:different_eps:kl}
	    \end{subfigure}
	    \begin{subfigure}[t]{0.15\textwidth}\vspace{-13em}
	        \includegraphics[width=1.0\textwidth]{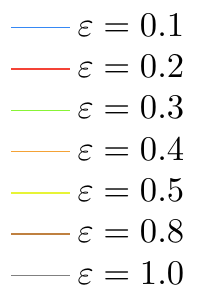}
	    \end{subfigure}
	    \caption{More detailed results for different $\eps$ $l^{ADP}_{n=128}\,U_{L_2}$, utility-weight decay applied.}
        \label{fig:different_eps}
	\end{figure*}

\begin{proof}[\Cref{thm:shift-invariance}]
We need to show that $\forall s'.\:\! -s \!\leq \!s' \!\leq\! s$ we have $\sum_{x\in S} p(x)\! \leq\! \delta \!+ e^\eps \sum_{x\in S} p(x+s)\!\Longrightarrow\!\sum_{x\in S} \!p(x) \!\leq\! \delta \!+\! e^\eps\! \sum_{x\in S} p(x\!+\!s')$. 
Since p is symmetric and the statement is trivially true for $s'=0$ we can limit ourselves to showing the statement $\forall s'.\: 0<s'\leq s$. We define the set $S' \subseteq S$ as the set of $x$ with $p(x) > e^\eps p(x+s)$ and rewrite our assumption to $\sum_{x\in S'} p(x) + \sum_{\substack{x \in S;\\x\notin S'}} p(x)  \! \leq\! \delta \!+ e^\eps \sum_{x\in S'} p(x+s) + e^\eps \sum_{\substack{x \in S;\\x\notin S'}} p(x+s)$.

\noindent As the statement must hold for any set $S$, specifically also for sets containing only one $x$, this can be split into two separate assumptions: $\forall x\in S\setminus S', \: p(x) \leq p(x+s)$ and $\sum_{x\in S'} p(x) \leq \delta \!+ e^\eps \sum_{x\in S'} p(x+s)$. 

\noindent Therefore, it is sufficient to show $\forall x\in S\setminus S', \: p(x) \leq p(x+s) \Longrightarrow p(x) \leq p(x+s')$ and $\sum_{x\in S'} (p(x) - e^\eps p(x+s)) \leq \delta \Longrightarrow \sum_{x \in S'} (p(x) - e^\eps p(x+s'))$.

\noindent First, we define the following statements for clarity

$A_1$: $p(x) \leq e^\eps p(x+s)$

$A_2$:  $x \notin S' \land 0<s' \leq s$ $\land$ p is symmetric  $\land$ p is a monotonously decreasing function from $p(0)$

B: $p(x) \leq e^\eps p(x+s')$

\noindent The first part of the proof shows that $\forall s, s', p, x. \: A_1 \land A_2 \Rightarrow B$. We prove this by contraposition, i.e., prooving that $\forall s, s', p, x. \: A_1 \land \neg B \Rightarrow \neg A_2$. We first show that the contraposition statement is equivalent to the required statement:
\begin{align*}
A_1 \land A_2 \Rightarrow B &\Leftrightarrow \neg (A_1 \land A_2) \Leftarrow \neg B& &|\text{De Morgan}\\
&\Leftrightarrow (\neg A_1 \lor \neg A_2) \lor B& &|\text{implication rule}\\
&\Leftrightarrow \neg A_2 \lor (\neg A_1 \lor B)\\
&\Leftrightarrow \neg A_2 \Leftarrow \neg(\neg A_1 \lor B)& &|\text{implication rule}\\
&\Leftrightarrow A_1 \land \neg B \Rightarrow \neg A_2
\end{align*}

\noindent Let $x\in \mathbb{R}$ and $x\notin S'$. We use a case distinction on $x$ to prove $p(x) \leq p(x+s) \Longrightarrow p(x) \leq p(x+s')$. Because of the symmetry assumption, we can always consider shifts in the negative direction as well.

\noindent \textbf{Case 1: $x+s' > 0$}.\\~\\
Since $x+s' > 0$ and $s \geq s'$, we have $x+s > 0$. 
Since $x+s \geq x+s'$, by monotonicity of $p$, we know that $p(x+s')\geq p(x+s)$. Furthermore, by design, we have $x-s\leq x-s'\leq x+s$. 
If $p(x-s')$ would be smaller than both $p(x-s)$ and $p(x+s)$, then $x-s'$ has to be either smaller than $x-s$ or larger than $x+s$ both of which is not possible. Thus, formulating this we get:
\begin{align*}
A_2 &\Rightarrow p(x+s') \geq p(x+s) \land\\
&\quad\left[p(x-s') \geq p(x+s) \lor p(x-s') \geq p(x-s)\right]
\end{align*}

\noindent Furthermore, we get the following expressions for $A_1$ and $\neg B$:\\
$A_1 \Leftrightarrow p(x) \leq e^\epsilon p(x+s) \land p(x) \leq e^\epsilon p(x-s)$\\
$\neg B \Leftrightarrow p(x) > e^\epsilon p(x+s') \lor p(x) > e^\epsilon p(x-s')$\\
We use the above defined statements to prove the contraposition:
\begin{alignat*}{3}
A_1 \land \neg B &\Rightarrow &&A_1 \land \left[p(x) > e^\epsilon p(x\!+\!s') \lor p(x) > e^\epsilon p(x\!-\!s')\right]&&\\
&\Rightarrow &&[\left\{p(x) \leq e^\epsilon p(x\!+\!s) \land p(x) \leq e^\epsilon p(x\!-\!s)\right\}\\
& &&\land p(x) > e^\epsilon p(x\!+\!s')]&&\\ 
& &&\lor [\left\{p(x) \leq e^\epsilon p(x\!+\!s) \land p(x) \leq e^\epsilon p(x\!-\!s)\right\}\\
& &&\land p(x) > e^\epsilon p(x\!-\!s')] &&\\
& \Rightarrow &&\left[p(x) \leq e^\epsilon p(x\!+\!s) \land p(x) > e^\epsilon p(x\!+\!s')\right] &&\\
& &&\lor [\left\{p(x) \leq e^\epsilon p(x\!+\!s) \land p(x) > e^\epsilon p(x\!-\!s')\right\}&&\\
& &&\land \left\{p(x) \leq e^\epsilon p(x\!-\!s) \land p(x) > e^\epsilon p(x\!-\!s')\right\}]&&\\
& \Rightarrow && p(x\!+\!s') < p(x\!+\!s)&&\\
& &&\lor \left[p(x\!-\!s') < p(x\!+\!s) \land p(x\!-\!s') < p(x\!-\!s)\right]\\
& \Rightarrow && \neg A_2
\end{alignat*}
\noindent \textbf{Case 2: $x+s' \leq 0$.}\\~\\
Since $x+s' \leq 0$ and $s, s' > 0$, we have $x, x-s \leq 0$. 
Since $x-s \leq x+s'$ as well as $x-s \leq x-s'$, by monotonicity of $p$, we know that $p(x+s')\geq p(x-s)$ as well as $p(x-s') \geq p(x-s)$.Thus, formulating this we get:
\begin{align*}
A_2 &\Rightarrow p(x+s') \geq p(x-s) \land p(x-s') \geq p(x-s)
\end{align*}

\noindent Furthermore, we get the following expressions for $A_1$ and $\neg B$:\\
$A_1 \Leftrightarrow p(x) \leq e^\epsilon p(x+s) \land p(x) \leq e^\epsilon p(x-s)$\\
$\neg B \Leftrightarrow p(x) > e^\epsilon p(x+s') \lor p(x) > e^\epsilon p(x-s')$\\
We use the above defined statements to prove the contraposition:
\begin{alignat*}{3}
A_1 \land \neg B &\Rightarrow &&\left[p(x) \leq e^\epsilon p(x+s) \land p(x) \leq e^\epsilon p(x-s)\right] \land \neg B&&\\
&\Rightarrow && p(x) \leq e^\epsilon p(x-s)\\
& &&\land \left[p(x) > e^\epsilon p(x+s') \lor p(x) > e^\epsilon p(x-s')\right]&&\\
&\Rightarrow && \left[p(x) \leq e^\epsilon p(x-s) \land p(x) > e^\epsilon p(x+s')\right]&&\\
& &&\lor \left[p(x) \leq e^\epsilon p(x-s) \land p(x) > e^\epsilon p(x-s')\right]&&\\
& \Rightarrow && p(x+s') < p(x-s) \lor p(x-s') < p(x-s)&&\\
& \Rightarrow && \neg A_2
\end{alignat*}

\noindent These two cases prove the sub-statement $\forall x\in S\setminus S'. \: p(x) \leq p(x+s) \Longrightarrow p(x) \leq p(x+s')$. Following this we show $\sum_{x\in S'} (p(x) - e^\eps p(x+s)) \leq \delta \Longrightarrow \sum_{x \in S'} (p(x) - e^\eps p(x+s'))\leq \delta$

\noindent By the monotonicity and symmetry assumption we can infer $\forall x\in S'. \: p(x+s) < p(x+s')$. It follows
\begin{align*}
    \Rightarrow& \sum_{x\in S'}  e^\eps p(x+s) \leq \sum_{x \in S'} e^\eps p(x+s')\\
    \Rightarrow& \sum_{x\in S'} (p(x) - e^\eps p(x+s)) \geq \sum_{x \in S'} (p(x) - e^\eps p(x+s'))\\
    \Rightarrow& \sum_{x\in S'} (p(x) - e^\eps p(x+s)) \leq \delta\\
    & \Longrightarrow \sum_{x \in S'} (p(x) - e^\eps p(x+s'))\leq \delta
\end{align*}

\noindent By showing that $\forall x\in S\setminus S'. \: p(x) \leq p(x+s) \Longrightarrow p(x) \leq p(x+s')$ and $\sum_{x\in S'} (p(x) - e^\eps p(x+s)) \leq \delta \Longrightarrow \sum_{x \in S'} (p(x) - e^\eps p(x+s'))\leq \delta$ we have shown that $\sum_{x\in S} p(x)\! \leq\! \delta \!+ e^\eps \sum_{x\in S} p(x+s)\!\Longrightarrow\!\sum_{x\in S} \!p(x) \!\leq\! \delta \!+\! e^\eps\! \sum_{x\in S} p(x\!+\!s')$. \new{Applying \cref{lemma:discretization-reduction} generalises this result to any query output and drawn noise in $\mathbb R$} which concludes the proof.
\end{proof}

	\begin{figure*}[htb!]
	    \centering
	    \begin{subfigure}[t]{0.42\textwidth}
    	    \includegraphics[width=1.0\textwidth]{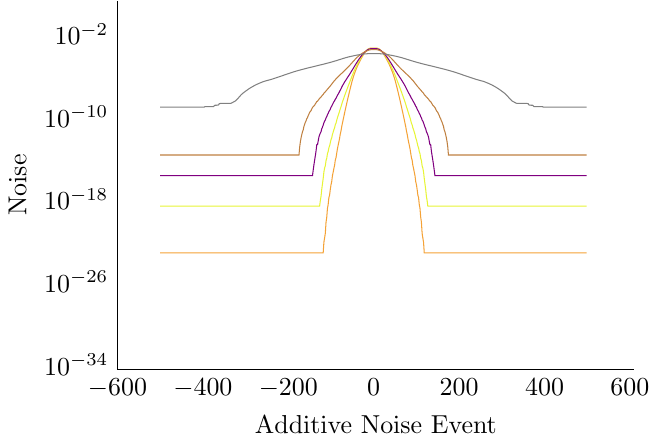}
	        \caption{$l_{B,A}^{ADP}$ for multiple number of compositions}
            \label{fig:dp_sgd_l2:pb_adp_b_a}
	    \end{subfigure}
	    \begin{subfigure}[t]{0.42\textwidth}
    	    \includegraphics[width=1.0\textwidth]{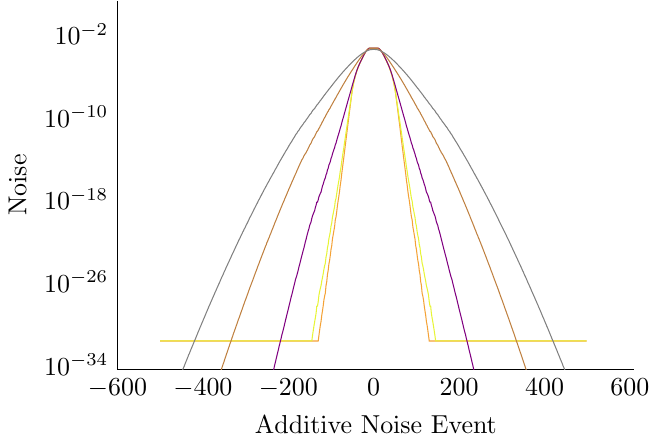}
	        \caption{$l_{B,A}^{MA}$ for multiple number of compositions}
            \label{fig:dp_sgd_l2:ma_b_a}
	    \end{subfigure}
	    \begin{subfigure}[t]{0.12\textwidth}\vspace{-14em}
    	    \includegraphics[width=0.9\textwidth]{images/texfiles/legend/legend_dp_sgd}
	    \end{subfigure}
	    \hfill\phantom{*}
	    \begin{subfigure}[t]{0.42\textwidth}
    	    \includegraphics[width=1.0\textwidth]{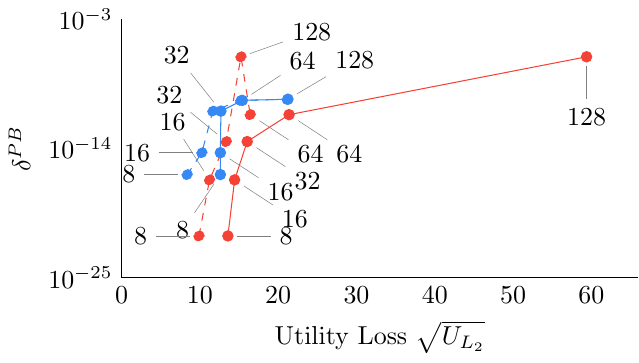}
	        \caption{Generated noise distributions compared to a truncated Gaussian with same $\delta^{PB}$\!\!. Numbers indicate compositions.}
            \label{fig:dp_sgd_l2:optimality}
	    \end{subfigure}
	    \begin{subfigure}[t]{0.42\textwidth}
    	    \includegraphics[width=1.0\textwidth]{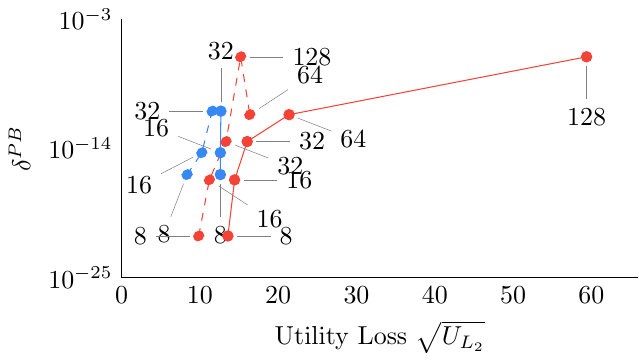}
	        \caption{Same as (c) but generated noise distributions and the corresponding truncated Gaussian are equally range-reduced.}
            \label{fig:dp_sgd_l2:optimality_cut}
	    \end{subfigure}
	    \begin{subfigure}[t]{0.12\textwidth}\vspace{-12em}
    	    \includegraphics[width=0.9\textwidth]{images/texfiles/legend/legend_dp_sgd_truncated}
	    \end{subfigure}
	    \hfill\phantom{*}
	    \begin{subfigure}[t]{0.42\textwidth}
    	    \includegraphics[width=1.0\textwidth]{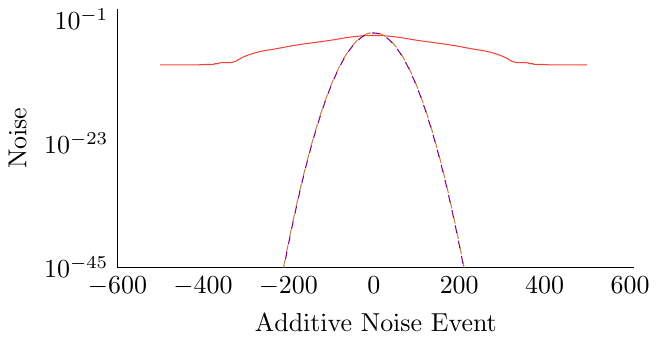}
	        \caption{
	        Illustration of generated vs. truncated vs. range-reduced truncated noise with equal $\delta^{PB}$\!\! for $l^{ADP}$\!\!\!, $n$=$128$ and $U_{L_2}$. 
	        Relative difference between truncated Gaussians std-dev: $0.23\%$.
	       }
            \label{fig:dp_sgd_l2:comparison:pb_adp}
	    \end{subfigure}
	    \begin{subfigure}[t]{0.42\textwidth}
    	    \includegraphics[width=1.0\textwidth]{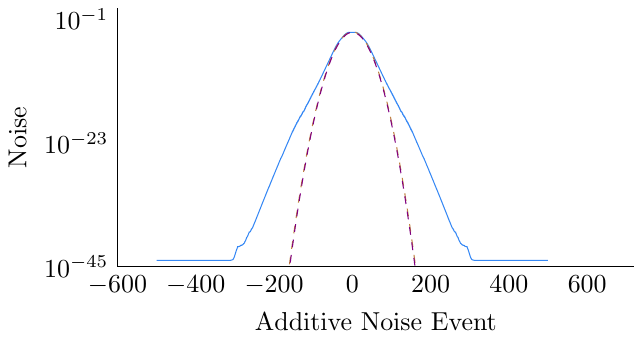}
	        \caption{Same as (e) $l^{MA}$ but for $n=32$ compositions (because no plateaus for $n$$\geq$$64$). Relative difference between truncated Gaussians std-dev: $0.75\%$.
	        }
            \label{fig:dp_sgd_l2:comparision:ma}
	    \end{subfigure}
	    \begin{subfigure}[t]{0.12\textwidth}\vspace{-12em}
    	    \includegraphics[width=0.9\textwidth]{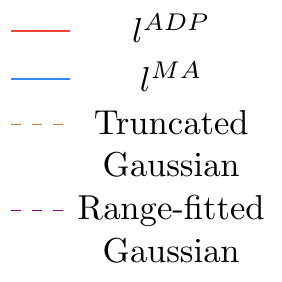}
	    \end{subfigure}
	    \hfill\phantom{*}
	    \begin{subfigure}[t]{0.42\textwidth}
    	    \includegraphics[width=1.0\textwidth]{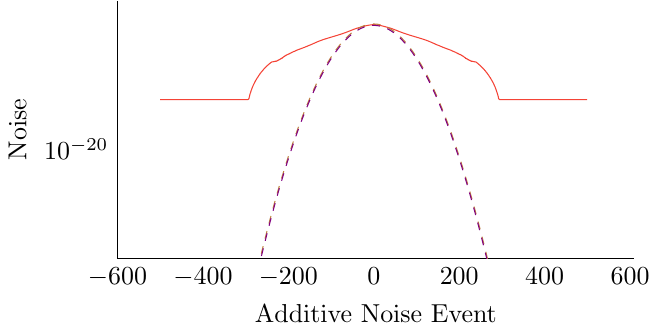}
	        \caption{
	        Illustration of generated vs. truncated vs. range-reduced truncated noise with equal $\delta^{PB}$\!\! for $l^{ADP}$\!\!\!, $n$=$128$ and $U_{L_1}$. 
	        Relative difference between truncated Gaussians std-dev: $0.33\%$.
	       }
           \label{fig:dp_sgd:comparison:pb_adp}
	    \end{subfigure}
	    \begin{subfigure}[t]{0.42\textwidth}
    	    \includegraphics[width=1.0\textwidth]{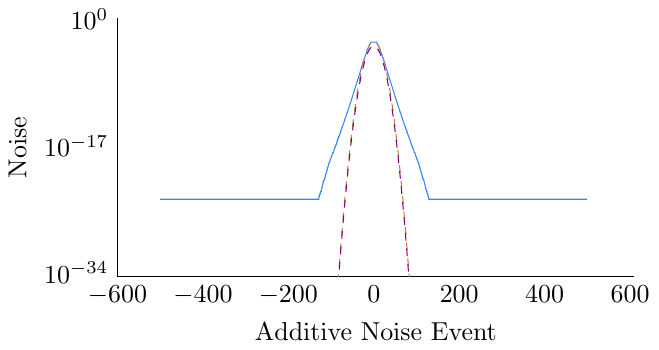}
	        \caption{Same as (g) $l^{MA}$ but for $n=8$ compositions (because no plateaus for $n$$\geq$$16$). Relative difference between truncated Gaussians std-dev: $0.15\%$.
	        }
            \label{fig:dp_sgd:comparision:ma}
	    \end{subfigure}
	    \hfill\phantom{*}
	    \caption{Illustration of DP-SGD (sub-sampling probability $q=0.1$) for multiple number of compositions $n$ and clipping distance $C=1$ for $U_{L_2}$ and comparison between the truncated and range-reduced Gaussian to the generated noise for $U_{L_1}$. PrivacyBuckets-ADP and Extended Moments Accountant (MA) are compared to the truncated Gaussian mechanism with same $\delta^{PB}$\!. Note that \textbf{(a)-(f)}: $U_{L_2}$, and \textbf{(g), (h)}: $U_{L_1}$. $\eps=0.3$}
        \label{fig:dp_sgd_l2}
    \end{figure*}

\vfill
\end{document}